\documentclass[reqno]{amsart}
\usepackage{amsmath}
\usepackage{amssymb}
\usepackage{amsthm}
\usepackage{amsfonts}
\usepackage{mathtools}
\usepackage{bm}
\usepackage[normalem]{ulem}  
\usepackage{fancyhdr}
\usepackage{hyperref}
\usepackage{framed}
\usepackage{float}
\usepackage[pdftex,final]{graphicx}
\usepackage{tikz}
\usepackage{comment}
\usepackage{sidecap}
\usepackage{xcolor}
\usepackage{accents}
\usepackage{pgfplots}
\usepackage{fixmath} 

\usetikzlibrary{calc,arrows}
\usetikzlibrary{patterns}
\usetikzlibrary{decorations.pathmorphing} 

\usepackage[iso-8859-7]{inputenc}

\theoremstyle{plain}
\newtheorem{thm}{Theorem}[section]
\newtheorem{corollary}[thm]{Corollary}
\newtheorem{lemma}[thm]{Lemma}
\newtheorem{prop}[thm]{Proposition}
\newtheorem*{theorem*}{Theorem}
\theoremstyle{definition}
\newtheorem{rem}[thm]{Remark}
\newtheorem{dfn}[thm]{Definition}
\newtheorem{exm}[thm]{Example}
\numberwithin{equation}{section}

\newcommand{\mytilde}{\raise.17ex\hbox{$\scriptstyle\mathtt{\sim}$}}
\newcommand{\RgO}{\ensuremath{\mathbb{R}_{>0}}}
\newcommand{\I}{\ensuremath{\mathcal{I}}}
\renewcommand{\S}{\ensuremath{\mathcal{S}}}

\newcommand{\G}{\ensuremath{\mathcal{G}}}

\newcommand{\E}{\ensuremath{\mathcal{E}}}

\newcommand{\N}{\ensuremath{\mathcal{N}}}

\renewcommand{\bf}[1]{\textbf{#1}}
\newcommand{\mbf}[1]{\mathbold{#1}}

\newcommand{\RgeO}{\ensuremath{\mathbb{R}_{ \geq 0}}}
\newcommand{\Rat}[1]{\ensuremath{\mathbb{R}^{#1}}}
\newcommand{\ubar}[1]{\underaccent{\bar}{#1}}

\renewcommand{\emph}[1]{\textit{#1}}

\begin{document}

\author{D. Boskos}
\address{Department of Automatic Control, School of Electrical Engineering, KTH Royal Institute of Technology, Osquldas v\"ag 10, 10044, Stockholm, Sweden}
\email{boskos@kth.se}

\author{D. V. Dimarogonas}
\address{Department of Automatic Control, School of Electrical Engineering, KTH Royal Institute of Technology, Osquldas v\"ag 10, 10044, Stockholm, Sweden}
\email{dimos@kth.se}

\begin{abstract}
In this report, we aim at the development of a decentralized abstraction framework for multi-agent systems under coupled constraints, with the possibility for a varying degree of decentralization. The methodology is based on the analysis employed in our recent work, where decentralized abstractions based exclusively on the information of each agent's neighbors were derived. In the first part of this report, we define the notion each agent's $m$-neighbor set, which constitutes a measure for the employed degree of decentralization. Then, sufficient conditions are provided on the space and time discretization that provides the abstract system's model, which guarantee the extraction of a transition system with quantifiable transition possibilities.
\end{abstract}

\keywords{hybrid systems, multi-agent systems, abstractions, transition systems.}

\title{Abstractions of Varying Decentralization Degree for Coupled Multi-Agent Systems}

\maketitle

\section{Introduction}

The analysis and control of multi-agent systems constitutes an active area of research with numerous applications, ranging from the analysis of power networks to the automatic deployment of robotic teams \cite{CyDxSaBc12}. Of central interest in this field is the problem of high level task planning by exploiting tools from formal verification  \cite{LsKk04}. In order to follow this approach for dynamic multi-agent systems it is required to provide an abstract model of the system which can serve as a tool for analysis and control, as well as high level planning. In particular, the use of a suitable discrete representation of the system allows the automatic synthesis of discrete plans that guarantee satisfaction of the high level specifications. Then, under appropriate relations between the continuous system and its discrete analogue, these plans can be converted to low level primitives such as sequences of feedback controllers, and hence, enable the continuous system to implement the corresponding tasks.

The need for a formal approach to the aforementioned control synthesis problem has lead to a considerable research effort for the extraction of discrete state symbolic models, also called abstractions, which capture properties of interest of continuous state dynamical and control systems, while ignoring detail. Results in this direction for the nonlinear single plant case have been obtained in the papers \cite{PgGaTp08} and \cite{ZmPgMmTp12}, which are based on partitioning and exploit approximate simulation and bisimulation relations. Symbolic models for piecewise affine systems on simplices and rectangles were introduced in  \cite{HlVj01} and have been further studied in \cite{BmGm14}. Closer related to the control framework that we adopt here for the abstraction, are the papers \cite{HmCp14b}, \cite{HmCp15} which build on the notion of In-Block Controllability \cite{CpWy95}. Other abstraction techniques for nonlinear systems include \cite{Rg11}, where discrete time systems are studied in a behavioral framework and \cite{AaTaSa09}, where box abstractions are studied for polynomial and other classes of systems. It is also noted that certain of the aforementioned approaches have been extended to switched systems \cite{GaPgTp10}, \cite{GeDxLmBc14}. Furthermore, abstractions for interconnected systems have been recently developed in \cite{TyIj08}, \cite{PgPpDm14},  \cite{PgPpDm16}, \cite{RmZm15}, \cite{Mp15a}, \cite{DeTp15} and rely mainly on compositional approaches  based on small gain arguments. Finally, in \cite{Mp15}, a compositional approach with a varying selection of subsystems for the abstraction is exploited, providing a tunable tradeoff between  complexity reduction and model accuracy.

In this framework, we focus on multi-agent systems and assume that the agents' dynamics consist of feedback interconnection terms and additional bounded input terms, which we call free inputs and provide the ability for motion planning under the coupled constraints. We generalize the corresponding results of our recent work \cite{BdDd15a}, where each agent's abstract model has been based on the knowledge of the discrete positions of its neighbors, by allowing the agent to have this information for all members of the network up to a certain distance in the communication graph. The latter  provides an improved estimate of the potential evolution of its neighbors and allows for more accurate discrete agent models, due to the reduction of the part of the available control which is required for the manipulation of the coupling terms. In addition, the derived abstractions are coarser than the ones in \cite{BdDd15a} and can reduce the computational complexity of high level task verification. Finally, we note that this report includes the proofs of its companion submitted conference version \cite{BdDd15c}, which have been completely omitted therein due to space constraints, as well as certain generalizations and additional material.


The rest of the report is organized as follows. Basic notation and preliminaries are introduced in Section 2. In Section 3, we define well posed abstractions for single integrator multi-agent systems and prove that the latter provide solutions consistent with the design requirement on the systems' free inputs. Section 4 is devoted to the study of the control laws that realize the transitions of the proposed discrete system's model. In Section 5 we quantify space and time discretizations which guarantee well posed transitions with motion planning capabilities. The framework is illustrated through an example with simulation results in Section 6 and we conclude in Section 7.

\section{Preliminaries and Notation}

We use the notation $|x|$ for the Euclidean norm of a vector $x\in\Rat{n}$ and ${\rm int}(S)$ for the interior of a set $S\subset\Rat{n}$. Given $R>0$ and $x\in\Rat{n}$, we denote by $B(R)$ the closed ball with center $0\in\Rat{n}$ and radius $R$, namely $B(R):=\{x\in\Rat{n}:|x|\le R \}$ and $B(x;R):=\{y\in\Rat{n}:|x-y|\le R \}$.

Consider a multi-agent system with $N$ agents. For each agent $i\in\N:=\{1,\ldots,N\}$ we use the notation $\mathcal{N}_i\subset\N\setminus\{i\}$ for the set of its neighbors and $N_i$ for its cardinality. We also consider an ordering of the agent's neighbors which is denoted by $j_1,\ldots,j_{N_i}$ and define the $N_i$-tuple $j(i)=(j_1(i),\ldots,j_{N_i}(i))$. Whenever it is clear from the context, the argument $i$ will be omitted from the latter notation. The agents' network is represented by a directed graph $\G:=(\N,\E)$, with vertex set $\N$ the agents' index set and edge set $\E$ the ordered pairs $(\ell,i)$ with $i,\ell\in\N$ and $\ell\in\N_i$. The sequence $i_0i_1\cdots i_m$ with $(i_{\kappa-1},i_{\kappa})\in\E$, $\kappa=1,\ldots,m$, namely, consisting of $m$ consecutive edges in $\G$, forms \textit{a path of length} $m$ in $\G$.
We will use the notation $\N_i^m$ to denote for each $m\ge 1$ the set of agents from which $i$ is reachable through a path of length $m$ and not by a shorter one, excluding also the possibility to reach itself through a cycle. Notice that $\N_i^1=\N_i$. We also define $\N_i^0:=\{i\}$ and  for each $m\ge 1$ the set $\bar{\N}_i^m:=\bigcup_{\ell=0}^m\N_i^{\ell}$, namely, the set of all agents from which $i$ is reachable by a path of length at most $m$, including $i$. With some abuse of language, we will use the terminology $m$-\textit{neighbor set} of agent $i$ for the set $\bar{\N}_i^{m}$, since it always contains the agent itself and will also refer to the rest of the agents in  $\bar{\N}_i^{m}$ as the $m$-neighbors of $i$. Finally, we denote by $\bar{N}_i^m$ and $N_i^m$ the cardinality of the sets $\bar{\N}_i^m$ and $\N_i^m$, respectively.

\begin{dfn} \label{mneighbor:set:ordering}
For each agent $i\in\N$, we consider a strict total order $\prec$ on each set $\bar{\N}_i^{m}$ that satisfies $i'\prec i''$ for each $i'\in\N_{i}^{m'}$, $i''\in\N_{i}^{m''}$, with $0\le m'<m''\le m$.

\end{dfn}

\begin{exm} \label{example}
In this example, we consider a network of 8 agents as depicted in Figure \ref{fig:graph} and illustrate the sets $\N_i^m$, $\bar{\N}_i^m$ for agents 1 and 5 up to paths of length $m=3$ as well as a candidate order for the corresponding sets $\bar{\N}_i^m$, in accordance to Definition \ref{mneighbor:set:ordering}.   

\begin{alignat*}{3}
& \bar{\N}_1^1=\{1,2,6\}, \quad && \N_1^1=\{2,6\}, \quad && 1\prec 6\prec 2, \\
& \bar{\N}_1^2=\{1,2,3,5,6,7\}, \quad && \N_1^2=\{3,5,7\}, \quad && 1\prec 6\prec 2\prec 5\prec 3\prec 7, \\
& \bar{\N}_1^3=\{1,2,3,4,5,6,7,8\}, \quad && \N_1^3=\{4,8\}, \quad && 1\prec 6\prec 2\prec 5\prec 3\prec 7\prec 4\prec 8, \\
& \bar{\N}_5^1=\{3,5\}, \quad && \N_5^1=\{3\}, \quad && 5\prec 3, \\
& \bar{\N}_5^2=\{2,3,4,5\}, \quad && \N_5^2=\{2,4\}, \quad && 5\prec 3\prec 2\prec 4, \\
& \bar{\N}_5^3=\{2,3,4,5\}, \quad && \N_5^3=\emptyset, \quad && 5\prec 3\prec 2\prec 4.
\end{alignat*}
\end{exm}

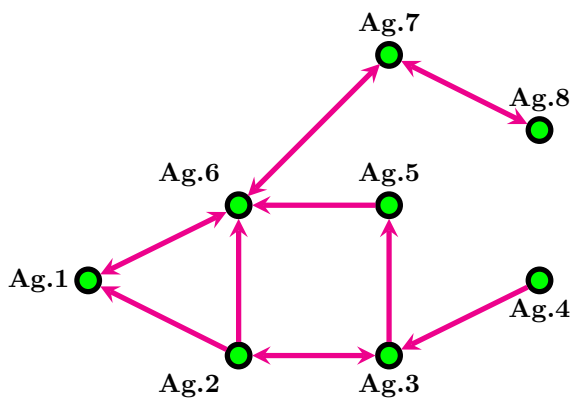
\begin{figure}[H]
\begin{center}
\begin{tikzpicture} [scale=1]

\filldraw[fill=green, line width=.07cm]  (-1,1) circle (0.15cm);
\node (agent 1) at (-1,1) [label=left:\textbf{Ag.1}] {};
\filldraw[fill=green, line width=.07cm]  (1,0) circle (0.15cm);
\node (agent 2) at (1,0) [label=below left:\textbf{Ag.2}] {};
\filldraw[fill=green, line width=.07cm]  (1,2) circle (0.15cm);
\node (agent 3) at (1,2) [label=above left:\textbf{Ag.6}] {};
\filldraw[fill=green, line width=.07cm]  (3,0) circle (0.15cm);
\node (agent 4) at (3,0) [label=below:\textbf{Ag.3}] {};
\filldraw[fill=green, line width=.07cm]  (3,2) circle (0.15cm);
\node (agent 5) at (3,2) [label=above:\textbf{Ag.5}] {};
\filldraw[fill=green, line width=.07cm]  (3,4) circle (0.15cm);
\node (agent 6) at (3,4) [label=above:\textbf{Ag.7}] {};
\filldraw[fill=green, line width=.07cm]  (5,1) circle (0.15cm);
\node (agent 7) at (5,1) [label=below:\textbf{Ag.4}] {};
\filldraw[fill=green, line width=.07cm]  (5,3) circle (0.15cm);
\node (agent 8) at (5,3) [label=above:\textbf{Ag.8}] {};

\draw[color=magenta,line width=.07cm,<-,>=stealth ] (-.84,.92)  -- (.84,.08);
\draw[color=magenta,line width=.07cm,<->,>=stealth ] (-.84,1.08) -- (.84,1.91);

\draw[color=magenta,line width=.07cm,->,>=stealth ] (1,.18) -- (1,1.82);
\draw[color=magenta,line width=.07cm,<->,>=stealth ] (1.18,0) -- (2.82,0);

\draw[color=magenta,line width=.07cm,<-,>=stealth ] (1.18,2) -- (2.82,2);
\draw[color=magenta,line width=.07cm,<->,>=stealth ] (1.12,2.12) -- (2.88,3.88);

\draw[color=magenta,line width=.07cm,->,>=stealth ] (3,.18) -- (3,1.82);
\draw[color=magenta,line width=.07cm,<-,>=stealth ] (3.16,.08) -- (4.84,.91);

\draw[color=magenta,line width=.07cm,<->,>=stealth ] (3.16,3.92)  -- (4.84,3.08);

\end{tikzpicture}
\end{center}
\caption{Illustration of the network topology for Example \ref{example}.}  \label{fig:graph}
\end{figure}

Given an index set $\I$, an agent $i\in\N$, its $m$-neighbor set $\bar{\N}_i^{m}$, and a strict total ordering of $\bar{\N}_i^{m}$ as in Definition \ref{mneighbor:set:ordering}, it follows that the elements of $\bar{\N}_i^{m}$ are ordered as $i\prec j_1^1(i)\prec\cdots\prec j_{N_i}^1(i)\prec j_1^2(i)\prec\cdots\prec j_{N_i^2}^2(i)\cdots\prec j_1^m(i)\cdots\prec j_{N_i^m}^m(i)$ with $j_{\kappa}^{m'}(i)\in\N_i^{m'}$ for all $m'\in\{1,\ldots,m\}$ and $\kappa\in\{1,\ldots,N_i^{m'}\}$. Whenever it is clear from the context we will remove the argument $i$ from the latter ordered elements. In addition, we define the mapping ${\rm pr}_i:\I^N\to\I^{\bar{N}_i^m}$ which assigns to each $N$-tuple $(l_1,\ldots,l_N)\in\I^N$ the $\bar{N}_i^m$-tuple $(l_i,l_{j_1^1},\ldots,l_{j_{N_i}^1},$ $l_{j_1^2},\ldots,l_{j_{N_i^2}^2},\ldots,l_{j_1^m},\ldots,l_{j_{N_i^m}^m})\in\I^{\bar{N}_i^m}$, i.e., the indices of agent $i$ and its $m$-neighbor set in accordance to the ordering.


\begin{dfn}
A transition system is a tuple $TS:=(Q,Act,\longrightarrow)$, where:

\textbullet\; $Q$ is a set of states.

\textbullet\; $Act$ is a set of actions.

\textbullet\; $\longrightarrow$ is a transition relation with $\longrightarrow\subset Q\times Act\times Q$.

\noindent The transition system is said to be finite, if $Q$ and $Act$ are finite sets. We also use the (standard) notation $q\overset{a}{\longrightarrow} q'$ to denote an element $(q,a,q')\in\longrightarrow$. For every $q\in Q$ and $a\in Act$ we use the notation ${\rm Post}(q;a):=\{q'\in Q:(q,a,q')\in\longrightarrow\}$.
\end{dfn}

\section{Abstractions for Multi-Agent Systems}

We consider multi-agent systems with single integrator dynamics
\begin{equation}\label{general:feedback:law}
\dot{x}_i=f_i(x_{i},\bf{x}_j)+v_{i}, i\in\N,
\end{equation}

\noindent that are governed by decentralized control laws consisting of two terms, a feedback term $f_{i}(\cdot)$ which depends on the states of $i$ and its neighbors, which we compactly denote by $\bf{x}_j(=\bf{x}_{j(i)}):=(x_{j_{1}},\ldots,x_{j_{N_i}})\in\Rat{N_i n}$ (see Section 2 for the notation $j(i)$), and an extra input term $v_{i}$, which we call free input. 
The dynamics $f_i(x,\bf{x}_j)$ are encountered in a large set of multi-agent protocols \cite{MmEm10}, including consensus, connectivity maintenance, collision avoidance and formation control. In addition they may represent internal dynamics of the system as for instance in the case of smart buildings (see e.g., \cite{AmDdShJk14}) , where the temperature $T_i$, $i\in\N$ of each room evolves according to $\dot{T}_i=\sum_{j\in\N_i} a_{ij}(T_j-T_i)+v_i$, with $a_{ij}$ representing the heat conductivity between rooms $i$ and $j$ and $v_i$ the heating/cooling capabilities of the room.

In what follows, we consider a cell decomposition of the state space $\Rat{n}$ and adopt a modification of the corresponding definition from \cite[p 129]{Gl02}.

\begin{dfn} \label{cell:decomposition}
Let $D$ be a domain of $\Rat{n}$. A \textit{cell decomposition} $\mathcal{S}=\{S_{l}\}_{l\in\mathcal{I}}$ of $D$, where $\mathcal{I}$ is a finite or countable index set, is a family of uniformly bounded and connected sets $S_{l}$, $l\in\mathcal{I}$, such that ${\rm int}(S_{l})\cap {\rm int}(S_{\hat{l}})=\emptyset$ for all $l\ne\hat{l}$ and $\cup_{l\in\mathcal{I}} S_{l}=D$.
\end{dfn}

\noindent Throughout the report, we consider a fixed $m\in\mathbb{N}$ which specifies the $m$-neighbor set of each agent and will refer to it as the \textit{degree of decentralization}. Also, given a cell decomposition $\S :=\{S_l\}_{l\in\I}$ of $\Rat{n}$, we use the notation $\bf{l}_i=(l_i,l_{j_1^1},\ldots,l_{j_{N_i}^1},l_{j_1^2},\ldots,l_{j_{N_i^2}^2},\ldots,l_{j_1^m},\ldots,l_{j_{N_i^m}^m})\in\I^{\bar{N}_i^m}$ to denote the indices of the cells where agent $i$ and its $m$-neighbors  belong at a certain time instant and call it the $m$-\textit{cell configuration} of agent $i$. We will also use the shorter notation $\bf{l}_i=(l_i,l_{j_1^1},\ldots,l_{j_{N_i^m}^m})$, or just $\bf{l}_i$. Similarly, we use the notation $\bf{l}=(l_{1},\ldots,l_{N})\in\mathcal{I}^{N}$ to specify the indices of the cells where all the $N$ agents belong at a given time instant and call it the cell configuration (of all agents). Thus, given a cell configuration $\bf{l}$, it is possible to determine the cell configuration of agent $i$ as $\bf{l}_i={\rm pr}_i(\bf{l})$ (see Section 2 for the definition of ${\rm pr}_i(\cdot)$).

Our aim is to derive finite or countable abstractions for each individual agent in the coupled system \eqref{general:feedback:law}, through the selection of a cell decomposition and a time discretization step $\delta t>0$. These will be  based on the knowledge of each agent's neighbors discrete positions up to a certain distance in the network graph, with the latter as specified by the degree of decentralization. Informally, we would like to consider for each agent $i$ the transition system with states the possible cells of the state partition, actions all the possible cells of the agents in its $m$-neighbor set and transition relation defined as follows. A final cell is reachable from an initial one, \textit{if for all states in the initial cell there is a free input such that the trajectory of $i$ will reach the final cell at time $\delta t$ for all possible initial states of its} $m$-\textit{neighbors in their cells and their corresponding free inputs}. For planning purposes we will require each individual transition system to be well posed, in the sense that for each initial cell it is possible to perform a transition to at least one final cell.

We next illustrate the concept of a well posed space-time discretization, namely, a discretization which generates for each agent a meaningful transition system as  discussed above. Consider a cell decomposition as depicted in Fig. 2 and a time step $\delta t$. The tips of the arrows in the figure are the endpoints of agent $i$'s trajectories at time $\delta t$. In both cases in the figure we focus on agent $i$ and consider the same 2-cell configuration for the 2-neighbor set $\bar{\N}_i^2=\{i,j_1,j_2\}$ of $i$, where $\N_i=\{j_1\}$ and $\N_{j_1}=\{j_2\}$, as depicted by the straight arrows in the figure. However, we consider different dynamics for Cases (i) and (ii). In Case (i), we observe that for the three distinct initial positions in cell $S_{l_i}$, it is possible to drive agent $i$ to cell $S_{l_i'}$ at time $\delta t$. We assume that this is possible for all initial conditions in this cell and irrespectively of the initial conditions of $j_1$ and $j_2$ in their cells and the inputs they choose. We also assume that this property holds for all possible 2-cell configurations of $i$ and for all the agents of the system. Thus we have a well posed discretization for System (i). On the other hand, for the same cell configuration and System (ii), we observe the following. For three distinct initial conditions of $i$ the corresponding reachable sets at $\delta t$, which are enclosed in the dashed circles, lie in different cells. Thus, it is not possible given this cell configuration of $i$ to find a cell in the decomposition which is reachable from every point in the initial cell and we conclude that the discretization is not well posed for System (ii).

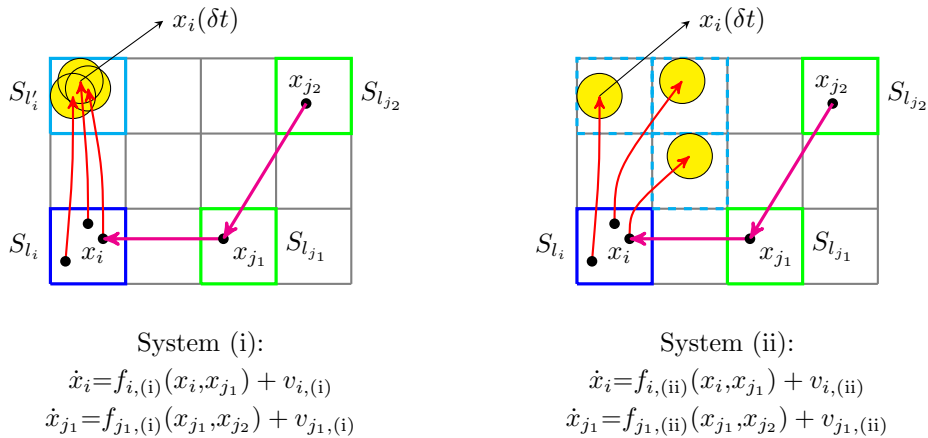
\begin{figure}[H]
\begin{center}
\begin{tikzpicture} [scale=1]

\draw[color=gray,thick] (0,0) -- (4,0);
\draw[color=gray,thick] (0,1) -- (4,1);
\draw[color=gray,thick] (0,2) -- (4,2);
\draw[color=gray,thick] (0,3) -- (4,3);

\draw[color=gray,thick] (0,0) -- (0,3);
\draw[color=gray,thick] (1,0) -- (1,3);
\draw[color=gray,thick] (2,0) -- (2,3);
\draw[color=gray,thick] (3,0) -- (3,3);
\draw[color=gray,thick] (4,0) -- (4,3);

\draw[color=blue,very thick] (0,0) -- (1,0) -- (1,1) -- (0,1) -- (0,0);
\draw[color=green,very thick] (2,0) -- (3,0) -- (3,1) -- (2,1) -- (2,0);
\draw[color=green,very thick] (3,2) -- (4,2) -- (4,3) -- (3,3) -- (3,2);
\draw[color=cyan,very thick] (0,2) -- (1,2) -- (1,3) -- (0,3) -- (0,2);


\fill[yellow] (0.3,2.5) circle (0.3cm);
\fill[yellow] (0.4,2.7) circle (0.3cm);
\fill[yellow] (0.5,2.6) circle (0.3cm);

\draw[black] (0.3,2.5) circle (0.3cm);
\draw[black] (0.4,2.7) circle (0.3cm);
\draw[black] (0.5,2.6) circle (0.3cm);

\draw [color=red,thick,->,>=stealth'](0.2,0.3) .. controls (0.3,1.5) .. (0.3,2.5);
\fill[black] (0.2,0.3) circle (2pt);

\draw [color=red,thick,->,>=stealth'](0.5,0.8) .. controls (0.5,1.5) .. (0.4,2.7);
\fill[black] (0.5,0.8) circle (2pt);

\draw [color=red,thick,->,>=stealth'](0.7,0.6) .. controls (0.7,1.5) .. (0.5,2.6);
\fill[black] (0.7,0.6) circle (2pt);

\draw [color=magenta,very thick,->,>=stealth'] (2.3,.6) -- (0.7,0.6);

\coordinate [label=left:$S_{l_i}$] (A) at (0,0.5);
\coordinate [label=above left:$x_{i}$] (A) at (0.85,0.15);
\coordinate [label=right:$S_{l_{j_1}}$] (A) at (3,.5);
\fill[black] (2.3,.6) circle (2pt) node[below right]{$x_{j_{1}}$};
\coordinate [label=right:$S_{l_{j_2}}$] (A) at (4,2.5);
\fill[black] (3.4,2.4) circle (2pt) node[above]{$x_{j_{2}}$};
\coordinate [label=left:$S_{l_i'}$] (A) at (0,2.5);

\draw [color=magenta,very thick,->,>=stealth'] (3.4,2.4) -- (2.3,.6);

\draw[color=gray,thick] (7,0) -- (11,0);
\draw[color=gray,thick] (7,1) -- (11,1);
\draw[color=gray,thick] (7,2) -- (11,2);
\draw[color=gray,thick] (7,3) -- (11,3);

\draw[color=gray,thick] (7,0) -- (7,3);
\draw[color=gray,thick] (8,0) -- (8,3);
\draw[color=gray,thick] (9,0) -- (9,3);
\draw[color=gray,thick] (10,0) -- (10,3);
\draw[color=gray,thick] (11,0) -- (11,3);

\draw[color=blue,very thick] (7,0) -- (8,0) -- (8,1) -- (7,1) -- (7,0);
\draw[color=green,very thick] (9,0) -- (10,0) -- (10,1) -- (9,1) -- (9,0);
\draw[color=green,very thick] (10,2) -- (11,2) -- (11,3) -- (10,3) -- (10,2);
\draw[color=cyan,dashed,very thick] (7,2) -- (8,2) -- (8,3) -- (7,3) -- (7,2);
\draw[color=cyan,dashed,very thick] (8,2) -- (9,2) -- (9,3) -- (8,3) -- (8,2);
\draw[color=cyan,dashed,very thick] (8,1) -- (9,1) -- (9,2) -- (8,2) -- (8,1);

\fill[yellow] (7.3,2.5) circle (0.3cm);
\fill[yellow] (8.4,2.7) circle (0.3cm);
\fill[yellow] (8.5,1.7) circle (0.3cm);

\draw[black] (7.3,2.5) circle (0.3cm);
\draw[black] (8.4,2.7) circle (0.3cm);
\draw[black] (8.5,1.7) circle (0.3cm);

\draw [color=red,thick,->,>=stealth'](7.2,0.3) .. controls (7.3,1.5) .. (7.3,2.5);
\fill[black] (7.2,0.3) circle (2pt);

\draw [color=red,thick,->,>=stealth'](7.5,0.8) .. controls (7.5,1.5) .. (8.4,2.7);
\fill[black] (7.5,0.8) circle (2pt);

\draw [color=red,thick,->,>=stealth'](7.7,0.6) .. controls (7.7,1) .. (8.5,1.7);
\fill[black] (7.7,0.6) circle (2pt);

\coordinate [label=below:System (i):] (A) at (2,-0.5);
\coordinate [label=below: $\dot{x}_{i}\text{$=$}f_{i,({\rm i})}(x_{i}\text{$,$}x_{j_{1}})+v_{i,({\rm i})}$] (A) at (2,-1);
\coordinate [label=below: $\dot{x}_{j_1}\text{$=$}f_{j_1,({\rm i})}(x_{j_1}\text{$,$}x_{j_{2}})+v_{j_1,({\rm i})}$] (A) at (2,-1.5);

\coordinate [label=below:System (ii):] (A) at (9,-0.5);
\coordinate [label=below: $\dot{x}_{i}\text{$=$}f_{i,({\rm ii})}(x_{i}\text{$,$}x_{j_{1}})+v_{i,({\rm ii})}$] (A) at (9,-1);
\coordinate [label=below: $\dot{x}_{j_1}\text{$=$}f_{j_1,({\rm ii})}(x_{j_1}\text{$,$}x_{j_{2}})+v_{j_1,({\rm ii})}$] (A) at (9,-1.5);

\draw [color=magenta,very thick,->,>=stealth'] (9.3,.6) -- (7.7,0.6);

\coordinate [label=left:$S_{l_i}$] (A) at (7,0.5);
\coordinate [label=above left:$x_{i}$] (A) at (7.85,0.15);
\coordinate [label=right:$S_{l_{j_1}}$] (A) at (10,.5);
\fill[black] (9.3,.6) circle (2pt) node[below right]{$x_{j_{1}}$};
\coordinate [label=right:$S_{l_{j_2}}$] (A) at (11,2.5);
\fill[black] (10.4,2.4) circle (2pt) node[above]{$x_{j_{2}}$};

\draw [color=magenta,very thick,->,>=stealth'] (10.4,2.4) -- (9.3,.6);

\draw[->,>=stealth] (0.4,2.7) -- (1.5,3.5);
\coordinate [label=right:$x_{i}(\delta t)$] (A) at (1.5,3.5);

\draw[->,>=stealth] (7.3,2.5) -- (8.5,3.5);
\coordinate [label=right:$x_{i}(\delta t)$] (A) at (8.5,3.5);

\end{tikzpicture}
\vspace{1em}
\caption{Illustration of a space-time discretization which is well posed for System (i) but non-well posed for System (ii).} \label{fig:well:posed}
\end{center}
\end{figure}

In order to provide meaningful decentralized abstractions we follow parts of the approach employed in \cite{BdDd15a} and design appropriate hybrid feedback laws in place of the $v_{i}$'s in order to guarantee well posed transitions. Before proceeding to the necessary definitions related to the problem formulation, we provide some bounds on the dynamics of the multi-agent system. We assume that there exists a constant $M>0$ such that
\begin{equation} \label{dynamics:bound}
|f_{i}(x_{i},\bf{x}_j)|\le M, \forall (x_{i},\bf{x}_j)\in \Rat{(N_i+1)n},i\in\N \\
\end{equation}

\noindent and that $f_{i}(\cdot)$ are globally Lipschitz functions. Thus, there exist Lipschitz constants $L_1>0$ and $L_2>0$, such that
\begin{align}
|f_{i}(x_{i},\bf{x}_j)-f_{i}(x_{i},\bf{y}_j)| & \le L_1 |(x_{i},\bf{x}_j)-(x_{i},\bf{y}_j)|,   \label{dynamics:bound1}  \\
|f_{i}(x_{i},\bf{x}_j)-f_{i}(y_{i},\bf{x}_j)| & \le L_2 |(x_{i},\bf{x}_j)-(y_{i},\bf{x}_j)|,  \label{dynamics:bound2} \\
\forall x_{i},y_{i} \in \Rat{n},\bf{x}_j,\bf{y}_j & \in \Rat{N_in},i\in\N. \nonumber
\end{align}

\noindent Furthermore, we assume that each input $v_{i}$, $i\in\N$ is piecewise continuous and satisfies the bound
\begin{equation}\label{input:bound}
|v_{i}(t)|\le v_{\max},\forall t\ge 0.
\end{equation}

\noindent where $ v_{\max}<M$. Finally, based on the uniform bound on the diameters of the cells in the decomposition of the workspace, we can define the diameter $d_{\max}$ of the cell decomposition as
\begin{equation} \label{dmax:dfn}
d_{\max}:=\inf\{R>0:\forall l\in\I,\exists x\in S_l\;{\rm with}\;S_l\subset B(x;\tfrac{R}{2})\}.
\end{equation}

\noindent This definition for the diameter of the cell decomposition is in general different from the one adopted in \cite{BdDd15a}. In particular, according to \eqref{dmax:dfn} it is defined as the ``minimum" among the diameters of the balls that can cover each cell in the decomposition, whereas in \cite{BdDd15a}, it is given as the ``minimum" among the diameters of the cells. However, for certain types of decompositions as for instance hyper-rectangles in $\Rat{n}$ and hexagons in $\Rat{2}$, these two definitions coincide.   

Given a cell decomposition we will consider a fixed selection of a reference point $x_{l,G}\in S_l$ for each cell $S_l$, $l\in\I$. For each agent $i$ and $m$-cell configuration of $i$, the corresponding reference points of the cells in the configuration will provide a trajectory which is indicative of the agent's reachability capabilities over the interval of the time discretization. In particular, for appropriate space-time discretizations the agent will be capable of reaching all points inside a ball with center the endpoint of this trajectory at the end of the time step. Furthermore, the selected reference points provide an estimate of the corresponding trajectories of the agent's neighbors, for $m$-cell configurations of its neighbors that are consistent with the cell configuration of $i$, namely for which the ``common agents" belong to the same cells, as formally given in Definition \ref{consistent:configurations} below. Before proceeding to Definition \ref{consistent:configurations} we provide a lemma which  establishes certain useful properties of the agents' $m$-neighbor sets.

\begin{lemma} \label{lemma:mneighbor:contains:mmin1}
\textit{(i)} For each agent $i\in\N$, neighbor $\ell\in\N_i$ of $i$ and $m\ge 1$, it holds $\bar{\N}_{\ell}^{m-1}\subset \bar{\N}_i^m$.

\noindent \textit{(ii)} For each $i\in\N$, $m\ge 1$ and $\ell\in\bar{\N}_i^m$ it holds $\N_{\ell}\subset \bar{\N}_i^{m+1}$.

\noindent \textit{(iii)} Assume that for certain $i\in\N$ and $m\ge 1$ it holds $\N_i^{m+1}=\emptyset$. Then, for each $\ell\in\bar{\N}_i^m$ it holds $
\N_{\ell}\subset \bar{\N}_i^m$.

\end{lemma}

\begin{proof}
For the proof of Part (i) let any $i'\in\bar{\N}_{\ell}^{m-1}$. Then, since $\bar{\N}_{\ell}^{m-1}=\cup_{\kappa=0}^{m-1}\N_{\ell}^{\kappa}$, either $i'=\ell$, which implies that $i'\in\N_i$ and hence also $i'\in\bar{\N}_i^m$, or $i'\in\N_{\ell}^{m'}$ for certain $m'\in\{1,\ldots,m\}$. In addition, if $i'=i$, then $i'\in \bar{\N}_i^m$ and hence, it remains to consider the case where $i'\ne i$ and $i'\in\N_{\ell}^{m'}$ for some $m'\in\{1,\ldots,m\}$. The latter implies that there exists  a shortest path $i_0\ldots i_{m'}$ with $i_0=i'$ and $i_{m'}=\ell$ from $i'$ to $\ell$. Then, since $i_{m'}=\ell\in\N_i$, it follows that $i_0\ldots i_{m'}i$ is a path of length $m'+1\le m$ from $i'$ to $i$. Thus, either it is a shortest path, implying that $i'\in\N_i^{m'+1}\subset\bar{\N}_i^m$, or, since $i'\ne i$, there exist $0<m''<m'+1$ and a shortest path of length $m''$ joining $i'$ and $i$. In the latter case, it follows that $i'\in\N_i^{m'}$ and hence again that $i'\in \bar{\N}_i^m$. The proof of Part (i) is now complete.

For the proof of Part (ii) let $\ell\in\bar{\N}_i^m$ and  $i'\in\N_{\ell}$. If $i'=i$ then $i'\in\bar{\N}_i^{m+1}$. Otherwise, since $\ell\in\bar{\N}_i^m$  there exists $1\le m'\le m$ and a path $i_0\ldots i_{m'}$ of length $m'$ with $i_0=\ell$ and $i_{m'}=i$, implying that $i'i_0\ldots i_{m'}$ is a path of length $m'+1\le m+1$ from $i'$ to $i$. Thus, it follows that $i'\in\bar{\N}_i^{m+1}$.

For the proof of Part (iii) let $\ell\in\bar{\N}_i^m$ and $i'\in\N_{\ell}$. If $\ell=i$ or $i'=i$, then the result follows directly from the facts that $\N_i\subset\bar{\N}_i^m$ and $i\in\bar{\N}_i^m$, respectively. Otherwise, there exists $1\le m'\le m$ such that $\ell\in\N_i^{m'}$, implying that there exists a path $i_0\ldots i_{m'}$ of length $m'$ with $i_0=\ell$ and $i_{m'}=i$. Thus, since $i'\in\N_{\ell}$, we get that $i'i_0\ldots i_{m'}$ is a path of length $m'+1$ from $i'$ to $i$. If it is a shortest path, then it follows that $m'<m$, because otherwise, since $i'\ne i$, we would have a shortest path of length $m+1$ joining $i'$ and $i$, implying that $\N_i^{m+1}\ne\emptyset$ and contradicting the hypothesis of Part (iii). Hence, it holds in this case that $i'\in\N_i^{m'+1}$ for certain $1\le m'< m$ and thus, that $i'\in\bar{\N}_i^m$. Finally, if $i'i_0\ldots i_{m'}$ is not a shortest path, then there exists a shortest path of length $1\le m''< m$ joining $i'$ and $i$, which implies that $i'\in\N_i^{m''}$ and hence again, that $i'\in\bar{\N}_i^m$.
\end{proof}

\begin{dfn}\label{consistent:configurations}
Consider an agent $i\in\N$, a neighbor $\ell\in\N_i$ of $i$ and $m$-cell configurations $\bf{l}_i=(l_i,l_{j_1^1},\ldots,l_{j_{N_i^m}^m})$ and $\bf{l}_{\ell}=(\bar{l}_{\ell},\bar{l}_{j(\ell)_1^1},\ldots,\bar{l}_{j(\ell)_{N_{\ell}^m}^m})$ of $i$ and $\ell$, respectively. We say that $\bf{l}_{\ell}$ \textit{is consistent with} $\bf{l}_i$ if for all $\kappa\in \bar{\N}_i^m\cap \bar{\N}_{\ell}^m$  it holds $l_{\kappa}=\bar{l}_{\kappa}$. 
$\;\triangleleft$
\end{dfn}

The following definition provides for each agent $i$ its reference trajectory and the estimates of its neighbors' reference trajectories, based on $i$'s  $m$-cell configuration.

\begin{dfn}\label{reference:trajectories}
\noindent Given a cell decomposition $\S=\{S_l\}_{l\in\I}$ of $\Rat{n}$, a reference point $x_{l,G}\in S_l$ for each $l\in\I$, a time step $\delta t$ and a nonempty subset $W$ of $\Rat{n}$, consider an agent $i\in\N$, its $m$-neighbor set $\bar{\N}_i^{m}$ and an $m$-cell configuration $\bf{l}_i=(l_i,l_{j_1^1},\ldots,l_{j_{N_i^m}^m})$ of $i$. We define the functions $\chi_i(t)$, $\mbf{\chi}_j(t):=(\chi_{j_1}(t),\ldots,\chi_{j_{N_i}}(t))$, $t\ge 0$, through the solution of the following \textit{initial value problem}, specified by Cases (i) and (ii) below:

\noindent \textit{Case (i).} $\N_i^{m+1}=\emptyset$. Then we have the initial value problem
\begin{align}
\dot{\chi}_{\ell}(t) = & f_{\ell}(\chi_{\ell}(t),\mbf{\chi}_{j(\ell)}(t))=f_{\ell}(\chi_{\ell}(t),\chi_{j(\ell)_1}(t),\ldots,\chi_{j(\ell)_{N_{\ell}}}(t)), t\ge 0, \ell \in \bar{\N}_i^{m}, \nonumber \\
\chi_{\ell}(0)= & x_{l_{\ell},G},\forall \ell \in \bar{\N}_i^{m}, \label{reference:IVP1}
\end{align}

\noindent where $j(\ell)_1,\ldots,j(\ell)_{N_{\ell}}$ denote the corresponding neighbors of each agent  $\ell \in \bar{\N}_i^{m}$.

\noindent \textit{Case (ii).} $\N_i^{m+1}\ne\emptyset$. Then we have the initial value problem
\begin{align}
\dot{\chi}_{\ell}(t) = & f_{\ell}(\chi_{\ell}(t),\mbf{\chi}_{\ell}(t))=f_{\ell}(\chi_{\ell}(t),\chi_{j(\ell)_1}(t),\ldots,\chi_{j(\ell)_{N_i}}(t)), t\ge 0, \ell \in \bar{\N}_i^{m-1}, \nonumber \\
\chi_{\ell}(0)= & x_{l_{\ell},G},\forall \ell \in \bar{\N}_i^{m-1}, \label{reference:IVP2}
\end{align}

\noindent with the terms $\chi_{\ell}(\cdot)$, $\ell \in \N_i^{m}$ defined as
\begin{equation} \label{reference:IVP2:constant:terms}
\chi_{\ell}(t):=x_{l_{\ell},G},\forall t\ge 0,\ell\in\N_i^{m}.\;\triangleleft
\end{equation}
\end{dfn}

\begin{rem} \label{remark:control:class:well:posed}
\textit{(i)} Notice, that in Case (i) for the initial value problem in Definition \ref{reference:trajectories}, the requirement $\N_i^{m+1}=\emptyset$ implies by vitue of Lemma \ref{lemma:mneighbor:contains:mmin1}(iii) for each agent $\ell\in\bar{\N}_i^m$ its neighbors $j(\ell)_1,\ldots,j(\ell)_{N_{\ell}}$ also belong to $\bar{\N}_i^m$. Hence, the subsystem formed by the  agents in $\bar{\N}_i^m$ is \textit{decoupled from the other agents in the system}
and the initial value problem \eqref{reference:IVP1} is well defined.

\noindent \textit{(ii)} In Case (ii), the  subsystem formed by the agents in $\bar{\N}_i^m$ is not decoupled from the other agents in the system. However, by considering the agents in $\N_i^{m}$ fixed at their reference points as imposed by \eqref{reference:IVP2:constant:terms}, it follows that the initial value problem \eqref{reference:IVP2}-\eqref{reference:IVP2:constant:terms} is also well defined. 

\noindent \textit{(iii)} Apart from the notation $\chi_i(\cdot)$ and $\mbf{\chi}_j(\cdot)$ above, we will use the notation $\chi_{\ell}^{[i]}(\cdot)$  for the trajectory of each agent $\ell\in\bar{\N}_i^m$ (including the functions $\chi_{\ell}^{[i]}(\cdot)$  as defined by \eqref{reference:IVP2:constant:terms} when $\N_i^{m+1}\ne\emptyset$), as specified by the initial value problem corresponding to the $m$-cell configuration of $i$ in Definition \ref{reference:trajectories}. We will refer to $\chi_i(\cdot)\equiv\chi_{\ell}^{[i]}(\cdot)$ as the \textit{reference trajectory} of agent $i$ and to each $\chi_{\ell}^{[i]}(\cdot)$ with $\ell\in\bar{\N}_i^m\setminus\{i\}$ as the \textit{estimate of $\ell$'s reference trajectory by $i$}.
\end{rem}

\begin{exm}
In this example we demonstrate the IVPs of Definition \ref{reference:trajectories} for $m=3$, as specified by Cases (i) and (ii) for agents 5 and 1, respectively. The topology of the network is the same as in Example \ref{example} and the agents involved in the solution of the correposnding initial value problems are depicted in Fig. \ref{graph:IVP}. 
\begin{align*}
\textup{Agent 5 - Case (i)} & \left\lbrace\begin{array}{l}
\dot{\chi}_5(t) =f_5(\chi_5(t),\chi_3(t)) \\
\dot{\chi}_3(t) =f_3(\chi_3(t),\chi_2(t),\chi_4(t)) \\
\dot{\chi}_2(t) =f_2(\chi_2(t),\chi_3(t)) \\
\dot{\chi}_4(t) =f_4(\chi_4(t)) \\
\chi_{\kappa}(0) =x_{\kappa,G},\kappa=5,3,2,4
\end{array}
\right. \\
\textup{Agent 1 - Case (ii)} & \left\lbrace\begin{array}{l}
\dot{\chi}_1(t) =f_1(\chi_1(t),\chi_6(t),\chi_2(t)) \\
\dot{\chi}_6(t) =f_6(\chi_6(t),\chi_1(t),\chi_2(t),\chi_5(t),\chi_7(t)) \\
\dot{\chi}_2(t) =f_2(\chi_2(t),\chi_3(t)) \\
\dot{\chi}_5(t) =f_5(\chi_5(t),\chi_3(t)) \\
\dot{\chi}_3(t) =f_3(\chi_3(t),\chi_2(t),x_{4,G}) \\
\dot{\chi}_7(t) =f_7(\chi_7(t),\chi_6(t),x_{8,G}) \\
\chi_{\kappa}(0) =x_{\kappa,G},\kappa=1,6,2,5,3,7
\end{array}
\right.
\end{align*}
\end{exm}

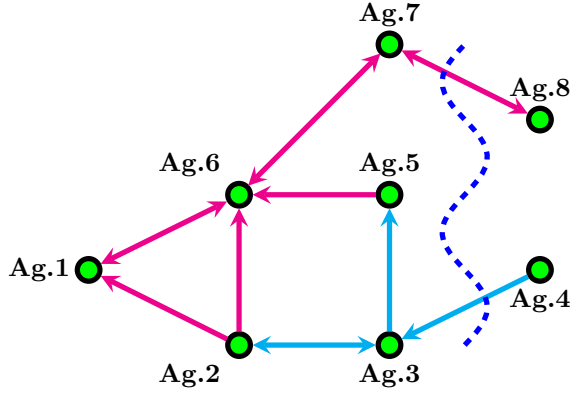
\begin{figure}[H]
\begin{center}
\begin{tikzpicture} [scale=1]

\filldraw[fill=green, line width=.07cm]  (-1,1) circle (0.15cm);
\node (agent 1) at (-1,1) [label=left:\textbf{Ag.1}] {};
\filldraw[fill=green, line width=.07cm]  (1,0) circle (0.15cm);
\node (agent 2) at (1,0) [label=below left:\textbf{Ag.2}] {};
\filldraw[fill=green, line width=.07cm]  (1,2) circle (0.15cm);
\node (agent 3) at (1,2) [label=above left:\textbf{Ag.6}] {};
\filldraw[fill=green, line width=.07cm]  (3,0) circle (0.15cm);
\node (agent 4) at (3,0) [label=below:\textbf{Ag.3}] {};
\filldraw[fill=green, line width=.07cm]  (3,2) circle (0.15cm);
\node (agent 5) at (3,2) [label=above:\textbf{Ag.5}] {};
\filldraw[fill=green, line width=.07cm]  (3,4) circle (0.15cm);
\node (agent 6) at (3,4) [label=above:\textbf{Ag.7}] {};
\filldraw[fill=green, line width=.07cm]  (5,1) circle (0.15cm);
\node (agent 7) at (5,1) [label=below:\textbf{Ag.4}] {};
\filldraw[fill=green, line width=.07cm]  (5,3) circle (0.15cm);
\node (agent 8) at (5,3) [label=above:\textbf{Ag.8}] {};

\draw[color=magenta,line width=.07cm,<-,>=stealth ] (-.84,.92)  -- (.84,.08);
\draw[color=magenta,line width=.07cm,<->,>=stealth ] (-.84,1.08) -- (.84,1.91);

\draw[color=magenta,line width=.07cm,->,>=stealth ] (1,.18) -- (1,1.82);

\draw[color=cyan,line width=.07cm,<->,>=stealth ] (1.18,0) -- (2.82,0);

\draw[color=magenta,line width=.07cm,<-,>=stealth ] (1.18,2) -- (2.82,2);
\draw[color=magenta,line width=.07cm,<->,>=stealth ] (1.12,2.12) -- (2.88,3.88);

\draw[color=cyan,line width=.07cm,->,>=stealth ] (3,.18) -- (3,1.82);
\draw[color=cyan,line width=.07cm,<-,>=stealth ] (3.16,.08) -- (4.84,.91);

\draw[color=magenta,line width=.07cm,<->,>=stealth ] (3.16,3.92)  -- (4.84,3.08);

\draw[color=blue,line width=0.07cm, dashed] plot [domain=0:90,samples=50]({4+0.3*sin(8*\x)},{\x/90*4});

\end{tikzpicture}
\end{center}
\caption{The neighbors on the right of the dotted line are considered fixed for the IVP of the 3-cell configuration of agent 1.} \label{graph:IVP}
\end{figure}

We next show that under certain conditions on the structure of the network graph in a neighborhood of each agent $i$, the reference trajectories of agent $i$'s neighbors coincide with their estimates by $i$, for consistent cell configurations. The following lemma provides this result.

\begin{lemma} \label{lemma:zero:deviation}
Assume that for agent $i\in\N$ it holds $\N_i^{m+1}=\emptyset$ and let $\bf{l}_i$ be an $m$-cell configuration of $i$. Then, for every agent $\ell\in\N_i$ with $\N_{\ell}^{m+1}=\emptyset$ and $m$-cell configuration $\bf{l}_{\ell}$ of $\ell$ consistent with $\bf{l}_i$, it holds $\chi_{\ell}^{[\ell]}(t)=\chi_{\ell}^{[i]}(t)$, for all $t\ge 0$, with $\chi_{\ell}^{[\ell]}(\cdot)$ and $\chi_{\ell}^{[i]}(\cdot)$ as determined by the initial value problem of Definition \ref{reference:trajectories} for the $m$-cell configurations $\bf{l}_{\ell}$ and $\bf{l}_i$, respectively.
\end{lemma}

\begin{proof}
Since $\N_i^{m+1}=\emptyset$, it follows that the initial value problem which corresponds to the $m$-cell configuration of agent $i$ and specifies the trajectory $\chi_{\ell}^{[i]}(\cdot)$ of $\ell$ is provided by Case (i) of Definition \ref{reference:trajectories}, namely, by \eqref{reference:IVP1}. We rewrite \eqref{reference:IVP1} in the compact form
\begin{equation} \label{IVP:agenti}
\dot{X}=F(X);\;F=(f_{\kappa_1},\ldots,f_{\kappa_{\bar{N}_i^m}}),\;X=(\chi_{\kappa_1},\ldots,\chi_{\kappa_{\bar{N}_i^m}}),
\end{equation}

\noindent with initial condition $\chi_{\kappa_{\nu}}(0)=x_{l_{\kappa_{\nu}},G}$, $\nu=1,\ldots,\bar{N}_i^m$, and $\kappa_{\nu}$ being the $\nu$-th index of $\bar{\N}_i^m$ according to the total order of Definition \ref{mneighbor:set:ordering}. Next, since $\N_{\ell}^{m+1}=\emptyset$, we similarly obtain that the initial value problem which corresponds to the $m$-cell configuration of agent $\ell$ and specifies its reference trajectory $\chi_{\ell}^{[\ell]}(\cdot)$  is given as
\begin{equation} \label{IVP:agentihat}
\dot{X}_1=F_1(X_1);\;F_1=(f_{\kappa_1'},\ldots,f_{\kappa_{\bar{N}_{\ell}^m}'}),\;X_1=(\chi_{\kappa_1'},\ldots,\chi_{\kappa_{\bar{N}_{\ell}^m}'}),
\end{equation}

\noindent with initial condition $\chi_{\kappa_{\nu}'}(0)=x_{l_{\kappa_{\nu}}',G}$, $\nu=1,\ldots,\bar{N}_{\ell}^m$, and $\kappa_{\nu}'$ being the corresponding $\nu$-th index of $\bar{\N}_{\ell}^m$. By taking into account that $\N_i^{m+1}=\emptyset$, we get from Lemma \ref{lemma:mneighbor:contains:mmin1}(i) that $\bar{\N}_{\ell}^m\subset \bar{\N}_i^{m+1}=\bar{\N}_i^m\cup\N_i^{m+1}=\bar{\N}_i^m$.  By assuming that without any loss of generality the inclusion is strict, we obtain from \eqref{IVP:agenti}, \eqref{IVP:agentihat} and Remark \ref{remark:control:class:well:posed}(i) for $\ell$ (namely, that the subsystem formed by the agents in $\bar{\N}_{\ell}^m$ is decoupled from the other agents), that by a reordering of the components $\chi_{\kappa_{\nu}}$, $\nu=1,\ldots,\bar{N}_i^m$, \eqref{IVP:agenti} can be cast in the form
\begin{align}
\dot{X}_1 & =F_1(X_1) \nonumber \\
\dot{X}_2 & =F_2(X_1,X_2), \label{IVP:agneti:equiv}
\end{align}

\noindent with $X_1$ and $F_1(\cdot)$ as in \eqref{IVP:agentihat}, $X_2=(\chi_{\kappa_{\bar{N}_{\ell}^m+1}'},\ldots,\chi_{\kappa_{\bar{N}_i^m}'})$, $F_2=(f_{\kappa_{\bar{N}_{\ell}^m+1}'},\ldots,f_{\kappa_{\bar{N}_i^m}'})$, and the same initial condition for the $X_1$ part as \eqref{IVP:agentihat}, due to the consistency of $\bf{l}_{\ell}$ with $\bf{l}_i$. Hence, since the $X_1$ part of the solution in \eqref{IVP:agneti:equiv} is independent of $X_2$, it follows that the reference trajectory $\chi_{\ell}^{[\ell]}(\cdot)$ of agent $\ell$ as given by \eqref{IVP:agentihat} and its corresponding estimate $\chi_{\ell}^{[i]}(\cdot)$ obtained from the first (decoupled) subsystem in \eqref{IVP:agneti:equiv} coincide.
\end{proof}

We next check the conditions of Lemma \ref{lemma:zero:deviation} for the graphs depicted in Fig. \ref{fig:cases} below. Consider first the graph of Case I, and assume that $(i,\ell)\notin\E$, where the interconnection between $i$ and $\ell$ is depicted through the non-dashed arrow. Then, by selecting $m=3$ and $m=4$, respectively, it turns out that $\N_i^4=\N_i^5=\emptyset$. However, we have  $\N_{\ell}^4=\{j_4\}$ and $\N_{\ell}^5=\{j_5\}$, respectively, which implies that the conditions of  Lemma \ref{lemma:zero:deviation} are not fulfilled. The same observation holds also for the graph of Case II when $(i,\ell)\notin\E$. On the other hand, when  $(i,\ell)\in\E$ it follows in Case I that for both  $m=3$ and $m=4$ it holds $\N_i^4=\N_i^5=\emptyset$ and $\N_{\ell}^4=\N_{\ell}^5=\emptyset$, namely the requirements of  Lemma \ref{lemma:zero:deviation} are satisfied. However, for Case II (again  when  $(i,\ell)\in\E$), these are only satisfied for $m=4$ where  $\N_i^5=\emptyset$ and $\N_{\ell}^5=\emptyset$, since for $m=3$ it holds $\N_i^4=\emptyset$ and $\N_{\ell}^4=\{j_5\}\ne\emptyset$.

\begin{figure}[H] 
\begin{center}
\begin{tikzpicture}[scale=.9]

\filldraw[fill=green, line width=.07cm]  (0,1) circle (0.15cm);
\node (agent 1) at (0,1) [label=left:\textbf{Ag.}$i$] {};
\filldraw[fill=green, line width=.07cm]  (2,2) circle (0.15cm);
\node (agent 2) at (2,2) [label=above:\textbf{Ag.}$\ell$] {};
\filldraw[fill=green, line width=.07cm]  (4,2) circle (0.15cm);
\node (agent 3) at (4,2) [label=above:\textbf{Ag.}$j_1$] {};
\filldraw[fill=green, line width=.07cm]  (6,2) circle (0.15cm);
\node (agent 4) at (6,2) [label=above:\textbf{Ag.}$j_2$] {};
\filldraw[fill=green, line width=.07cm]  (6,0) circle (0.15cm);
\node (agent 5) at (6,0) [label=below:\textbf{Ag.}$j_3$] {};
\filldraw[fill=green, line width=.07cm]  (4,0) circle (0.15cm);
\node (agent 6) at (4,0) [label=below:\textbf{Ag.}$j_4$] {};
\filldraw[fill=green, line width=.07cm]  (2,0) circle (0.15cm);
\node (agent 7) at (2,0) [label=below:\textbf{Ag.}$j_5$] {};

\draw[color=blue,line width=.07cm,<->,>=stealth,dashed] (.16,1.3)  -- (1.82,2.14);
\draw[color=magenta,line width=.07cm,<-,>=stealth ] (.16,1.08)  -- (1.82,1.92);
\draw[color=magenta,line width=.07cm,<-,>=stealth ] (2.18,2) -- (3.82,2);
\draw[color=magenta,line width=.07cm,<-,>=stealth ] (4.18,2) -- (5.82,2);
\draw[color=magenta,line width=.07cm,<-,>=stealth ] (6,1.82) -- (6,.18);
\draw[color=magenta,line width=.07cm,<->,>=stealth ] (4.18,0) -- (5.82,0);
\draw[color=magenta,line width=.07cm,<->,>=stealth ] (2.18,0) -- (3.82,0);
\draw[color=magenta,line width=.07cm,<-,>=stealth ] (.16,.92)  -- (1.82,.08);

\filldraw[fill=green, line width=.07cm]  (8,0) circle (0.15cm);
\node (agent 1) at (8,0) [label=below left:\textbf{Ag.}$i$] {};
\filldraw[fill=green, line width=.07cm]  (8,2) circle (0.15cm);
\node (agent 2) at (8,2) [label=left:\textbf{Ag.}$\ell$] {};
\filldraw[fill=green, line width=.07cm]  (10,4) circle (0.15cm);
\node (agent 3) at (10,4) [label=right:\textbf{Ag.}$j_1$] {};
\filldraw[fill=green, line width=.07cm]  (10,2) circle (0.15cm);
\node (agent 4) at (10,2) [label=right:\textbf{Ag.}$j_2$] {};
\filldraw[fill=green, line width=.07cm]  (10,0) circle (0.15cm);
\node (agent 5) at (10,0) [label=below:\textbf{Ag.}$j_3$] {};
\filldraw[fill=green, line width=.07cm]  (12,0) circle (0.15cm);
\node (agent 6) at (12,0) [label=below right:\textbf{Ag.}$j_4$] {};
\filldraw[fill=green, line width=.07cm]  (12,2) circle (0.15cm);
\node (agent 7) at (12,2) [label=right:\textbf{Ag.}$j_5$] {};

\draw[color=blue,line width=.07cm,<->,>=stealth,dashed] (7.75,.18)  -- (7.75,1.82);
\draw[color=magenta,line width=.07cm,<-,>=stealth ] (8,.18)  -- (8,1.82);
\draw[color=magenta,line width=.07cm,<-,>=stealth ] (8.12,2.12) -- (9.88,3.88);
\draw[color=magenta,line width=.07cm,<-,>=stealth ] (10,3.82) -- (10,2.18);
\draw[color=magenta,line width=.07cm,<-,>=stealth ] (10,1.82) -- (10,.18);
\draw[color=magenta,line width=.07cm,<-,>=stealth ] (8.18,0) -- (9.82,0);
\draw[color=magenta,line width=.07cm,<-,>=stealth ] (10.18,0) -- (11.82,0);
\draw[color=magenta,line width=.07cm,<-,>=stealth ] (12,.18)  -- (12,1.82);

\node (Case I) at (3,-1) [label=center:\textbf{Case I.}] {};
\node (Case II) at (10,-1) [label=center:\textbf{Case II.}] {};

\end{tikzpicture}
\end{center}
\caption{Illustration of graph topologies that satisfy the conditions of Lemma \ref{lemma:mhatneighbors:empty} for $m=4$.} \label{fig:cases}
\end{figure}
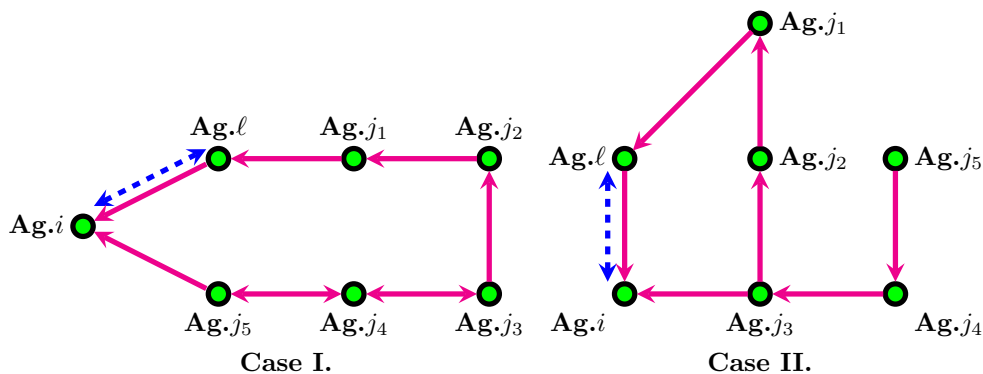

We note that in both cases discussed above, the fact that when $(i,\ell)\in\E$, the condition $\N_i^4=\emptyset$ implies that $\N_{\ell}^5=\emptyset$ is justified by the following lemma, which provides sufficient conditions for the assumptions of Lemma \ref{lemma:zero:deviation} to hold.

\begin{lemma} \label{lemma:mhatneighbors:empty}
Assume that for agent $i$ it holds $\N_i^{m}=\emptyset$ for certain $m\ge 2$. \textbf{(i)} Then, for every agent $\ell\in\N_i$ it holds
\begin{equation} \label{neighbor:mset:inclusion}
\bar{\N}_{\ell}^{m+\kappa}\subset \bar{\N}_i^m,\forall\kappa\ge 0.
\end{equation}

\noindent \textbf{(ii)} Furthermore, if $i\in\N_{\ell}$, then it also holds
\begin{equation}
\N_{\ell}^{m+1}=\emptyset.
\end{equation}
\end{lemma}

\begin{proof}

For Part (i), we first show that 
\begin{equation} \label{mset:inclusion}
\N_i^{m+\kappa}=\emptyset,\forall\kappa\ge 1.
\end{equation} 

\noindent Indeed, assume on the contrary that $i'\in \N_i^{m+\kappa}\ne\emptyset$ for certain $\kappa\ge 1$ and thus, there exists a shortest path $i_0i_1\ldots i_{m-1+\kappa}i_{m+\kappa}$ of length $m+\kappa$ in $\G$ with $i_0=i'$ and $i_{m+\kappa}=i$. Then, it follows that $i_{\kappa}i_{\kappa+1}\ldots i_{m-1+\kappa}i_{m+\kappa}$ is a shortest path of length $m$ from $i_{\kappa}$ to $i$, because otherwise, if there were a shorter one $i_{\kappa}i_{\kappa+1}'\ldots i_{m_0-1+\kappa}'i_{m+\kappa}$ with $m_0<m$, $i_0i_1\ldots i_{\kappa}i_{\kappa+1}'\ldots i_{m_0-1+\kappa}'i_{m+\kappa}$ would be a path of length $m_0+\kappa<m+\kappa$ from $i'$ to $i$, contradicting minimality of $i_0i_1\ldots i_{m-1+\kappa}i_{m+\kappa}$. Thus, we get that $i_{\kappa}\in\N_i^m$, which contradicts the assumption that $\N_i^m\ne\emptyset$ and we conclude that \eqref{mset:inclusion} is fulfilled. By exploiting the latter, we get that for each $\kappa\ge0$ \eqref{neighbor:mset:inclusion} holds, since by virtue of Lemma \ref{lemma:mneighbor:contains:mmin1}(i) we have $\bar{\N}_{\ell}^{m+\kappa}\subset\bar{\N}_i^{m+\kappa+1}=\bar{\N}_i^m\cup\bigcup_{\kappa'=1}^{\kappa+1}\N_i^{m+\kappa'}=\bar{\N}_i^m$.

For the proof of Part (ii), we proceed again by contradiction. Hence, consider a shortest path $i_0i_1\ldots i_mi_{m+1}$ with $i_{m+1}=\ell$, and notice that by virtue of \eqref{neighbor:mset:inclusion} with $\kappa=1$, it holds $i_0\in\bar{\N}_i^m$. We consider two cases. If $i_0=i$, then by virtue of the fact that $i\in\N_{\ell}$, which implies that $(i,\ell)\in\E$, we have that $i_0\ell$ is a path of length 1 joining $i_0$ and $\ell$, contradicting minimality of  $i_0i_1\ldots i_mi_{m+1}$. Now consider the case where  $i_0\ne i$. Then, it follows from the facts that $i_0\in\bar{\N}_i^m$, $i_0\ne i$ and the hypothesis  $\N_i^{m}=\emptyset$, that there exists a shortest path $i_0i_1'\ldots i_{m_0-1}'i$ of length $m_0<m$, joining $i_0$ and $i$. Thus, since $(i,\ell)\in\E$ we obtain that $i_0i_1'\ldots i_{m_0-1}'i\ell$ is a path of length $m_0+1\le m$ joining  $i_0$ and $\ell$, contradicting minimality of  $i_0i_1\ldots i_mi_{m+1}$. The proof is now complete.
\end{proof}

Despite the result of Lemma \ref{lemma:zero:deviation}, in principle, the actual trajectory of each agent's neighbor and its estimate, based on the solution of the initial value problem  for the reference trajectory of the specific agent do not coincide. Explicit bounds for the deviation between these trajectories and their estimates are given in Proposition \ref{proposition:neighbor:rt:deviation} below, whose proof requires the following auxiliary lemma.

\begin{lemma} \label{lemma:reftraj:worst:case:dev}
Consider the agent $i\in\N$ and let $\bf{l}_i$ be an $m$-cell configuration of $i$. Also, pick $\ell\in\N_i$ and any $m$-cell configuration $\bf{l}_{\ell}$ of $\ell$ consistent with $\bf{l}_i$. Finally, let $t^*$ be the unique positive solution of the equation
\begin{equation} \label{time:tstar}
e^{L_2t^*}-\left(L_2+\frac{L_2^2}{L_1\sqrt{N_{\max}}}\right)t^*-1=0,
\end{equation}

\noindent with
\begin{equation} \label{Nmax}
\quad N_{\max}:=\max\{N_i:i\in\N\}.
\end{equation}

\noindent Then, for each $\kappa\in\bar{\N}_{\ell}^{m-1}$ (recall that $\bar{\N}_{\ell}^{m-1}\subset\bar{\N}_i^m$ by Lemma \ref{lemma:mneighbor:contains:mmin1}(i)) it holds:
\begin{equation} \label{reftraj:deviation:all}
|\chi_{\kappa}^{[i]}(t)-\chi_{\kappa}^{[\ell]}(t)|\le Mt,\forall t\in [0,t^*],
\end{equation}

\noindent where $\chi_{\kappa}^{[i]}(\cdot)$ and $\chi_{\kappa}^{[\ell]}(\cdot)$ are determined by the initial value problem of Definition \ref{reference:trajectories} for the $m$-cell configurations $\bf{l}_i$ and $\bf{l}_{\ell}$, respectively.
\end{lemma}

\begin{proof}
The proof of the lemma is given in the Appendix.
\end{proof}

\begin{prop} \label{proposition:neighbor:rt:deviation}
Consider the agent $i\in\N$ and let $\bf{l}_i$ be an $m$-cell configuration of $i$. Also, pick $\ell\in\N_i$ and any $m$-cell configuration $\bf{l}_{\ell}$ of $\ell$ consistent with $\bf{l}_i$. Then the difference $|\chi_{\ell}^{[i]}(\cdot)-\chi_{\ell}^{[\ell]}(\cdot)|$ satisfies the bound:
\begin{equation} \label{neighbor:rt:deviation}
|\chi_{\ell}^{[i]}(t)-\chi_{\ell}^{[\ell]}(t)|\le H_m(t), \forall t\in [0,t^*],
\end{equation}

\noindent where $t^*$ is given in \eqref{time:tstar}, the functions $H_{\kappa}(\cdot)$, $\kappa\ge 1$, are defined recursively as
\begin{equation} \label{functions:Hkappa}
H_1(t):=Mt,t\ge 0;\quad H_{\kappa}(t):=\int_0^te^{L_2(t-s)}L_1\sqrt{N_{\max}}H_{\kappa-1}(s)ds, t\ge 0
\end{equation}

\noindent and $\chi_{\ell}^{[i]}(\cdot)$, $\chi_{\ell}^{[\ell]}(\cdot)$ are determined by the initial value problem of Definition \ref{reference:trajectories} for the $m$-cell configurations $\bf{l}_i$ and $\bf{l}_{\ell}$, respectively.
\end{prop}

\begin{proof}
The proof is carried out by induction and is based on the result of Lemma \ref{lemma:reftraj:worst:case:dev}. We will prove the following induction hypothesis:

\noindent \textit{IH.} For each $m'\in\{1,\ldots,m\}$ and $\iota\in\bar{\N}_{\ell}^{m-m'}$ it holds
\begin{equation} \label{induction:hypothesis}
|\chi_{\iota}^{[i]}(t)-\chi_{\iota}^{[\ell]}(t)|\le H_{m'}(t), \forall t\in [0,t^*].
\end{equation}

\noindent Notice that for $m'=m$ the Induction Hypothesis implies the desired \eqref{neighbor:rt:deviation}. Also, by virtue of Lemma \ref{lemma:reftraj:worst:case:dev}, IH is valid for $m'=1$. In order to prove the general step, assume that IH is fulfilled for certain $m'\in\{1,\ldots,m-1\}$ and consider any $\iota\in\bar{\N}_{\ell}^{m-(m'+1)}$. Notice first, that since $m-(m'+1)\le m-2$, both $\mbf{\chi}_{j(\iota)}^{[i]}(\cdot)$ and $\mbf{\chi}_{j(\iota)}^{[\ell]}(\cdot)$  are well defined and the differences $|\chi_{\nu}^{[i]}(\cdot)-\chi_{\nu}^{[\ell]}(\cdot)|$, $\nu\in\N_{j(\iota)}$ of their respective components satisfy \eqref{induction:hypothesis} with $m'$. It then follows that  $|\mbf{\chi}_{j(\iota)}^{[i]}(t)-\mbf{\chi}_{j(\iota)}^{[\ell]}(t)|=\left(\sum_{\nu\in\N_{j(i)}}|\chi_{\nu}^{[i]}(t)-\chi_{\nu}^{[\ell]}(t)|^2\right)^{\frac{1}{2}}\le \sqrt{N_{\iota}}H_{m'}(t)$ for all $t\in(0,t^*)$. Thus, by evaluating   $|\chi_{\iota}^{[i]}(\cdot)-\chi_{\iota}^{[\ell]}(\cdot)|$ as in \eqref{tildexi:difference:integral:ineq} in the Appendix, we obtain that
\begin{equation*} 
|\chi_{\iota}^{[i]}(t)-\chi_{\iota}^{[\ell]}(t)| \le \int_0^tL_1\sqrt{N_{\iota}}H_{m'}(s)ds+\int_0^tL_2|\chi_{\iota}^{[i]}(s)-\chi_{\iota}^{[\ell]}(s)| ds
\end{equation*}

\noindent By exploiting Fact I used in the proof of Lemma \ref{lemma:reftraj:worst:case:dev} in the Appendix, we obtain from \eqref{induction:hypothesis}and \eqref{functions:Hkappa} that
\begin{align*}
|\chi_{\iota}^{[i]}(t)-\chi_{\iota}^{[\ell]}(t)| & \le\int_0^te^{L_2(t-s)}L_1\sqrt{N_{\iota}}H_{m'}(s)ds \\
& \le\int_0^te^{L_2(t-s)}L_1\sqrt{N_{\max}}H_{m'}(s)ds=H_{m'+1}(t),\forall t\in[0,t^*],
\end{align*}

\noindent which establishes the general step of the induction procedure. The proof is now complete.
\end{proof}

\begin{rem}
Explicit formulas for the functions $H_{\kappa}(\cdot)$ in \eqref{functions:Hkappa} are provided in the Appendix. In addition, we provide some linear upper bounds for them in the following section, which are then further utilized for the explicit derivation of acceptable space-time discretizations.
\end{rem}

In order to provide the definition of well posed transitions for the individual agents, we will consider for each agent $i\in\N$ the following system with disturbances:
\begin{equation} \label{system:disturbances}
\dot{x}_i=f_i(x_i,\bf{d}_j)+v_i,
\end{equation}

\noindent where $d_{j_1},\ldots,d_{j_{N_i}}:[0,\infty)\to\Rat{n}$ (also denoted $d_{\ell}$, $\ell\in\N_i$) are continuous functions. Also, before defining the notion of a well posed space-time discretization we provide a class of hybrid feedback laws which are assigned to the free inputs $v_i$ in order to obtain meaningful discrete transitions. For each agent, these control laws are parameterized by the agent's initial conditions and a set of auxiliary parameters belonging to a nonempty subset $W$ of $\Rat{n}$. These parameters, as will be clarified in the next section, are exploited for motion planning. 
In addition, for each agent $i$, the feedback laws in the following definition depend on the selection of the cells where $i$ and its $m$-neighbors belong.

\begin{dfn}\label{control:class}
\noindent Given a cell decomposition $\S=\{S_l\}_{l\in\I}$ of $\Rat{n}$ and a nonempty subset $W$ of $\Rat{n}$, consider an agent $i\in\N$ and an initial cell configuration $\bf{l}_i$ of $i$. For each $x_{i0}\in S_{l_i}$ and $w_i\in W$, consider the mapping  $k_{i,\bf{l}_i}(\cdot,\cdot,\cdot;x_{i0},w_i):[0,\infty)\times \Rat{(N_i+1)n}\to\Rat{n}$, parameterized by $x_{i0}\in S_{l_i}$ and $w_i\in W$. We say that $k_{i,\bf{l}_i}(\cdot)$ satisfies \textit{Property (P)}, if the following conditions are satisfied.

\noindent\textit{(P1)} The mapping $k_{i,\bf{l}_i}(t,x_i,\bf{x}_j;x_{i0},w_i)$ is continuous on $[0,\infty)\times \Rat{(N_i+1)n}\times S_{l_i}\times W$.

\noindent\textit{(P2)} The mapping $k_{i,\bf{l}_i}(t,\cdot,\cdot;x_{i0},w_i)$ is globally Lipschitz continuous on $(x_i,\bf{x}_j)$ (uniformly with respect to $t\in[0,\infty)$, $x_{i0}\in S_{l_i}$ and $w_i\in W$). $\triangleleft$
\end{dfn}

\noindent The next definition characterizes the bounds on the deviation between the reference trajectory of each agent's neighbor and its estimate obtained from the solution of the initial value problem for the specific agent.

\begin{dfn} \label{dfn:nrtdb}
Consider an agent $i\in\N$. We say that a continuous function $\alpha_i:[0,\delta t]\to\RgeO$ satisfies the \textit{neighbor reference trajectory deviation bound}, if for each cell configuration $\bf{l}_i$ of $i$, neighbor $\ell\in\N_i$ of $i$ and cell configuration of $\ell$ consistent with $\bf{l}_i$ it holds:
\begin{equation} \label{alpha:property}
|\chi_{\ell}^{[i]}(t)-\chi_{\ell}^{[\ell]}(t)|\le \alpha_i(t),\forall t\in[0,\delta t]. \;\triangleleft
\end{equation}
\end{dfn}

\noindent We are now in position to formalize our requirement which describes the possibility for an agent to perform a discrete transition, based on the knowledge of its $m$-cell configuration. The corresponding definition below includes certain bounds on the evolution of the agent, which we sharpen, by requiring the reference points $x_{l,G}$, $l\in\I$ of the cell decomposition to satisfy
\begin{equation} \label{reference:points}
|x_{l,G}-x|\le\frac{d_{\max}}{2},\forall x\in S_l,l\in\I.
\end{equation}

\noindent Due to \eqref{dmax:dfn}, the latter is always possible. The definition exploits the auxiliary system with disturbances \eqref{system:disturbances}, which is inspired by the approach adopted in \cite{GaMs12}, where a nonlinear system is modeled by means of a piecewise affine system with disturbances.

\begin{dfn} \label{conditionCC}
Consider a cell decomposition $\S=\{S_l\}_{l\in\I}$ of $\Rat{n}$, a time step $\delta t$, a nonempty subset $W$ of $\Rat{n}$, and a continuous function $\beta:[0,\delta t]\to\RgeO$ satisfying
\begin{equation} \label{beta:property}
\frac{d_{\max}}{2}\le\beta(0); \beta(\delta t)\le v_{\max}\delta t.
\end{equation}

\noindent Also, consider an agent $i\in\N$, a continuous function $\alpha_i:[0,\delta t]\to\RgeO$ satisfying the neighbor reference trajectory deviation bound \eqref{alpha:property},
 an $m$-cell configuration $\bf{l}_i$ of $i$ and the solution of the initial value problem of Definition \ref{reference:trajectories}. Then, given a control law
\begin{equation} \label{feedback:for:i}
v_i=k_{i,\bf{l}_i}(t,x_{i},\bf{x}_j;x_{i0},w_i)
\end{equation}

\noindent as in Definition \ref{control:class}, that satisfies Property (P), a vector $w_i\in W$, and a cell index $l'\in\I$, we say that the \textit{Consistency Condition} is satisfied if the following hold. The set  ${\rm int}(B(\chi_i(\delta t);\beta(\delta t)))\cap S_{l_i'}$ is nonempty, and there exists a point $x_i'\in{\rm int}(B(\chi_i(\delta t);\beta(\delta t)))\cap S_{l_i'}$, such that for each initial condition $x_{i0}\in {\rm int}(S_{l_i})$ and selection of continuous functions $d_{\ell}:\RgeO\to\Rat{n}$, $\ell\in\N_i$ satisfying
\begin{equation} \label{disturbance:bounds}
|d_{\ell}(t)-\chi_{\ell}^{[i]}(t)|\le\alpha_i(t)+\beta(t),\forall t\in[0,\delta t],
\end{equation}

\noindent where $\chi_{\ell}^{[i]}(\cdot)$,  $\ell\in\N_i$ correspond to $i$'s estimates of its neighbors' reference trajectories, the solution $x_i(\cdot)$ of the system with disturbances \eqref{system:disturbances} with $v_i=k_{i,\bf{l}_i}(t,x_{i},\bf{d}_j;x_{i0},w_i)$, satisfies

\begin{equation} \label{xi:consistency:bounds}
|x_i(t)-\chi_i^{[i]}(t)|<\beta(t), \forall t\in[0,\delta t].
\end{equation}

\noindent Furthermore, it holds
\begin{equation} \label{xi:in:finalcell}
x_i(\delta t)=x_i'\in S_{l_i'}
\end{equation}

\noindent and
\begin{equation} \label{controler:consistency}
|k_{i,\bf{l}_i}(t,x_{i}(t),\bf{d}_j(t);x_{i0},w_i)|\le v_{\max},\forall  t\in[0,\delta t]. \quad\triangleleft
\end{equation}
\end{dfn}

Notice, that when the Consistency Condition is satisfied, agent $i$ can be driven to cell $S_{l_i'}$ precisely in time $\delta t$ under the feedback law $k_{i,\bf{l}_i}(\cdot)$ corresponding to the given parameter $w_i$ in the definition. The latter is possible for all disturbances which satisfy \eqref{disturbance:bounds} and capture the possibilities for the evolution of $i$'s neighbors over the time interval $[0,\delta t]$, given the knowledge of the $m$-cell configuration of $i$.

We next provide some intuition behind the Consistency Condition and the particular selection of the neighbor reference trajectory deviation bounds $\alpha_i(\cdot)$ and the function $\beta(\cdot)$ introduced therein, through an example with four agents as illustrated in Fig. \ref{fig:CC}. The agents $i$, $\ell$, $j_1$ and $j_2$ form a path graph, and their neighbor sets are given as $\N_i=\{\ell\}$, $\N_{\ell}=\{j_1\}$,   $\N_{j_1}=\{j_2\}$ and  $\N_{j_2}=\emptyset$, respectively. The degree of decentralization is assumed to be $m=3$ and we focus primarily on agent $i$. Agent $i$ is depicted at the left of the figure together with its reference trajectory $\chi_i(\cdot)=\chi_i^{[i]}(\cdot)$ initiated from the reference point $x_{l_i,G}$. The latter constitutes the center of the ball with minimal radius that encloses the cell $S_{l_i}$ where agent $i$ is contained, as specified by \eqref{reference:points} and depicted with the dashed circle in the figure. The (green) area enclosing agent $i$ is the set $\cup_{t\in [0,\delta t]}B(\chi_i(t);\beta(t))$, namely, the union of all possible positions of the agent whose distance from the reference trajectory at each time $t$ does not exceed the bound $\beta(t)$, or equivalently satisfy  \eqref{xi:consistency:bounds}. This is the restriction that the Consistency Condition imposes on agent $i$, and implicitly through the acceptable disturbances that satisfy \eqref{disturbance:bounds}, to its neighbor $\ell$. 
The latter is depicted through the larger (red) area $\cup_{t\in [0,\delta t]}B(\chi_{\ell}^{[i]}(t);\alpha_i(t)+\beta(t))$, which encloses the reference point $x_{l_{\ell},G}$ of agent $\ell$. The interior (darker red) part of this area is $\cup_{t\in [0,\delta t]}B(\chi_{\ell}^{[i]}(t);\alpha_i(t))$, namely, the union of the points with distance from the (dashed) reference trajectory $\chi_{\ell}^{[i]}(t)$ of $\ell$ as estimated by $i$ is no more than the reference trajectory deviation bound $a_i(t)$ at each $t$. The exterior part depicts the inflation ot the interior one by $\beta(\cdot)$. Thus, when the Consistency Condition is applied to agent $\ell$, namely, when its motion over $[0,\delta t]$ with respect to its own reference trajectory $\chi_{\ell}^{[\ell]}(t)$ is bounded at each $t$ by $\beta(t)$, as depicted with the dotted area in the figure, agent $\ell$ will remain within the larger area enclosing its reference point $x_{l_{\ell},G}$. Therefore, in order to capture the behaviour of agent $\ell$, it is assumed that the disturbances $d_{\ell}(\cdot)$ which model the possible trajectories of $\ell$ satisfy \eqref{disturbance:bounds} and thus belong to the lager (red) area containing the reference point $x_{l_{\ell},G}$. The reason for this hypothesis comes from the fact that the $m$-cell configuration of $i$ allows only for the computation of $\chi_{\ell}^{[i]}(\cdot)$ and not the actual reference trajectory $\chi_{\ell}^{[\ell]}(\cdot)$ of $\ell$.

On the right hand side of the figure we also illustrate the trajectories of the reference solutions  $\chi_{j_1}^{[j_1]}(\cdot)$,  $\chi_{j_1}^{[\ell]}(\cdot)$ and  $\chi_{j_1}^{[i]}(\cdot)$ for agent $j_1$, as evaluated by the agent itself and agents $\ell$ and $i$, respectively. As was the case for agent $\ell$, the reference trajectory $\chi_{j_1}^{[j_1]}(t)$ of $j_1$ and its estimate  $\chi_{j_1}^{[\ell]}(t)$ by $\ell$ do not exceed the reference trajectory deviation bound $a_i(t)$ at each time $t$. However, as shown in the figure, the same does not necessarily hold for the corresponding estimate $\chi_{j_1}^{[i]}(\cdot)$ by $i$, since $\chi_{j_1}^{[j_1]}(\cdot)$ lies outside the dotted area depicting the reference trajectory deviation bound around $\chi_{j_1}^{[i]}(\cdot)$. 
Finally, it is noted that according to the Consistency Condition, agent $i$ can be assigned a feedback law in order to reach the final cell $S_{l_i'}$ from each point inside the initial cell $S_{l_i}$ and for any disturbance satisfying the left hand side of  \eqref{xi:consistency:bounds}. Such a path of agent $i$ and a corresponding disturbance are depicted in the figure. More details on the specific feedback laws that are applied for this purpose are provided in the next section.   

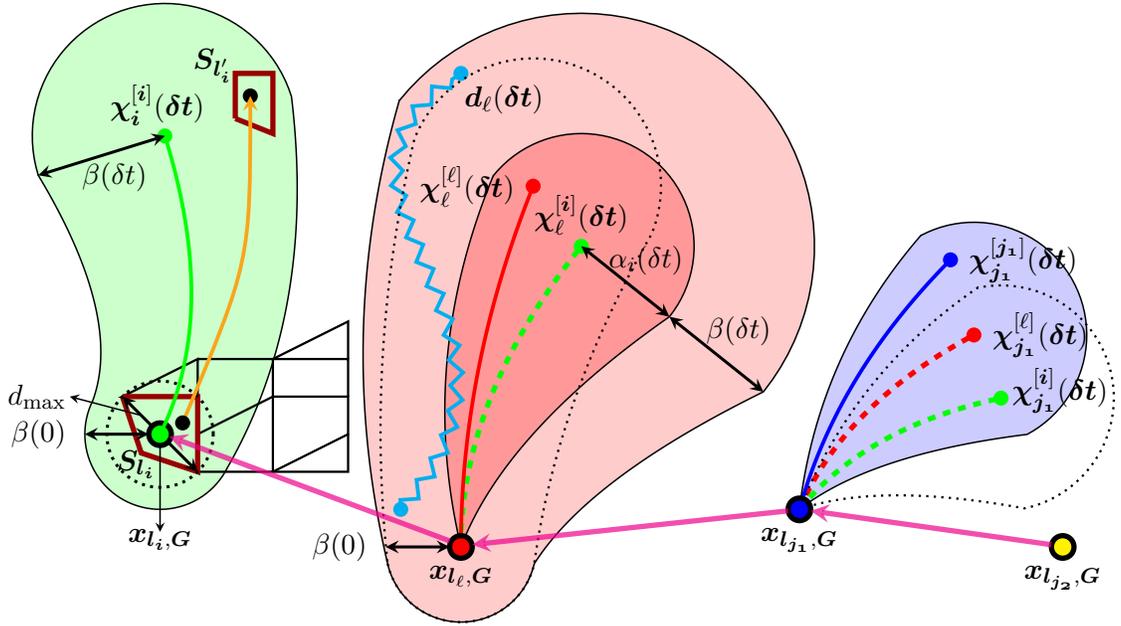
\begin{figure}[H]
\begin{center}

\begin{tikzpicture}[scale=1]


\filldraw[fill=yellow, line width=.07cm]  (12,0) circle (0.15cm);
\node (agent 4) at (12,0) [label=below:{\large $\mbf{x_{l_{j_2},G}}$}] {};

\def\tend{4}
\def\aone{.4}
\def\atwo{.8}
\def\bone{\atwo}
\def\btwo{-\aone}
\def\baratio{.12}

\pgfmathsetmacro\aonetip{\aone*\tend-\baratio*\bone*\tend^2}
\pgfmathsetmacro\atwotip{1.5+\atwo*\tend-\baratio*\btwo*\tend^2}

\pgfmathsetmacro\aonefinaux{\aone-2*\baratio*\bone*\tend}
\pgfmathsetmacro\atwofinaux{\atwo-2*\baratio*\btwo*\tend}

\pgfmathsetmacro\aonefin{\aonefinaux/sqrt(\aonefinaux^2+\atwofinaux^2)}
\pgfmathsetmacro\atwofin{\atwofinaux/sqrt(\aonefinaux^2+\atwofinaux^2)}

\pgfmathsetmacro\coneaux{\atwo}
\pgfmathsetmacro\ctwoaux{-\aone}


\pgfmathsetmacro\cone{\coneaux/sqrt(\coneaux^2-\ctwoaux^2)}
\pgfmathsetmacro\ctwo{\ctwoaux/sqrt(\coneaux^2-\ctwoaux^2)}

\def\done{1.6*\atwofin}
\def\dtwo{-1.6*\aonefin}

\pgfmathsetmacro\phi{atan(\aone/\atwo)}

\draw ({\aonetip},{\atwotip}) -- ({\aonetip+\done},{\atwotip+\dtwo});
\draw (0,1.5) -- (\cone,1.5+\ctwo);


\filldraw[line width=.02cm,  fill=green!20!white] 
plot [domain=180-\phi:360-\phi, variable=\theta,samples=50]({cos(\theta)},{1.5+sin(\theta)}) 
-- 
plot [domain=0:\tend, variable=\t,samples=50]({\aone*\t-\baratio*\bone*\t^2+\t/\tend*\done+(\tend-\t)/\tend*\cone},{1.5+\atwo*\t-\baratio*\btwo*\t^2+\t/\tend*\dtwo+(\tend-\t)/\tend*\ctwo})
--  
plot [domain=0:180, variable=\theta,samples=50]({\aonetip+\done*cos(\theta)-\dtwo*sin(\theta)},{\atwotip+\done*sin(\theta)+\dtwo*cos(\theta)})
--
plot [domain=0:\tend, variable=\t,samples=50]({\aone*(\tend-\t)-\baratio*\bone*(\tend-\t)^2-\t/\tend*\cone-(\tend-\t)/\tend*\done},{1.5+\atwo*(\tend-\t)-\baratio*\btwo*(\tend-\t)^2-\t/\tend*\ctwo-(\tend-\t)/\tend*\dtwo});


\draw [line width=.03cm] (.5,2) -- (-.5,2) -- (.5,2.5) -- (.5,1) -- (1.5,1) -- (1.5,2.5) -- (.5,2.5);
\draw [line width=.03cm] (.5,1.5) -- (1.5,2);
\draw [line width=.03cm] (1.5,1) -- (2.5,1.5) -- (2.5,2) -- (1.5,2);
\draw [line width=.03cm] (2.5,2) -- (2.5,2.5) -- (1.5,2.5);
\draw [line width=.03cm] (1.5,1) -- (2.5,1) -- (2.5,1.5);
\draw [line width=.03cm] (1.5,2.5) -- (2.5,3) --(2.5,2.5);

\draw [line width=.07cm, red!60!black] (-.5,2) -- (-.25,1.25) -- (.5,1) -- (.5,2) -- (-.5,2);

\draw [line width=.04cm,<->,>=stealth] (-.5,2) -- (.5,1); 
\draw [line width=.04cm, dotted] (0,1.5) circle (0.707 cm);

\draw [line width=.07cm, red!60!black] (1,5.7) -- (1,6.3) -- (1.5,6.3)-- (1.5,5.5) -- (1,5.7);

\filldraw[fill=green, line width=.07cm]  (0,1.5) circle (0.15cm);

\draw[line width=.05cm, green] plot [domain=0:\tend, variable=\t,samples=50]({\aone*\t-\baratio*\bone*\t^2},{1.5+\atwo*\t-\baratio*\btwo*\t^2});

\fill[fill=green] (\aonetip,\atwotip) circle (0.1cm);

\draw [line width=.04cm,<->,>=stealth] (-1,1.5) -- (-.15,1.5); 
\node (irrel) at (-1,1.5) [label=left:{\large $\beta(0)$}] {};

\draw [line width=.04cm,<->,>=stealth] (\aonetip,\atwotip) -- (\aonetip-\done,\atwotip-\dtwo); 
\node (irrel) at (\aonetip-.8*\done,\atwotip-.2*\dtwo) [label=below right:{\large $\beta(\delta t)$}] {};

\draw [line width=.02cm,->,>=stealth] (-.25,1.75)  -- (-1.2,2); 
\node (irrel) at (-1,2) [label=left:{\large $d_{\max}$}] {};
\node (irrel) at (\aonetip-.1,\atwotip-.1)  [label=above:{\large $\mbf{\chi_i^{[i]}(\delta t)}$}] {};
\node (irrel) at (-.32,1.1) [label=center:{\large $\mbf{S_{l_i}}$}] {};
\node (irrel) at (.7,6.4) [label=center:{\large $\mbf{S_{l_i'}}$}] {};

\node (agent i) at (0,0.1) [label=center:{\large $\mbf{x_{l_i,G}}$}] {};

\draw [line width=.02cm,->,>=stealth] (0,1.35)  -- (0,.2); 

\fill (1.2,6) circle (0.1cm);
\draw[line width=.05cm, yellow!60!red,->,>=stealth] (.3,1.65) .. controls (1.2,3.8) .. (1.2,6) ;
\fill (.3,1.65) circle (0.1cm);


\def\aone{0}
\def\atwo{1}
\pgfmathsetmacro\bone{\atwo}
\pgfmathsetmacro\btwo{-\aone}
\def\baratio{.1}

\def\aonetip{4+\aone*\tend+\baratio*\bone*\tend^2}
\def\atwotip{\atwo*\tend+\baratio*\btwo*\tend^2}

\pgfmathsetmacro\aonefinaux{\aone+2*\baratio*\bone*\tend}
\pgfmathsetmacro\atwofinaux{\atwo+2*\baratio*\btwo*\tend}

\pgfmathsetmacro\aonefin{\aonefinaux/sqrt(\aonefinaux^2+\atwofinaux^2)}
\pgfmathsetmacro\atwofin{\atwofinaux/sqrt(\aonefinaux^2+\atwofinaux^2)}

\pgfmathsetmacro\bonefin{1.5*\atwofin}
\pgfmathsetmacro\btwofin{-1.5*\aonefin}

\pgfmathsetmacro\coneaux{\atwo+\btwofin/\tend\aone+\bonefin/\tend}
\pgfmathsetmacro\ctwoaux{-(\aone+\bonefin/\tend)}


\pgfmathsetmacro\cone{\coneaux/sqrt(\coneaux^2-\ctwoaux^2)}
\pgfmathsetmacro\ctwo{\ctwoaux/sqrt(\coneaux^2-\ctwoaux^2)}

\pgfmathsetmacro\done{3.1*\atwofin}
\pgfmathsetmacro\dtwo{-3.1*\aonefin}

\pgfmathsetmacro\phi{atan(\ctwoaux/\coneaux)}

\filldraw[line width=.02cm, fill=red!20!white] plot [domain=180-\phi:360+\phi, variable=\theta,samples=50]({4+cos(\theta)},{sin(\theta)}) 
-- 
plot [domain=0:\tend, variable=\t,samples=50]({4+\aone*\t+\baratio*\bone*\t^2+\t/\tend*\done+(\tend-\t)/\tend*\cone},{\atwo*\t+\baratio*\btwo*\t^2+\t/\tend*\dtwo+(\tend-\t)/\tend*\ctwo})
 --  
plot [domain=0:180, variable=\theta,samples=50]({\aonetip+\done*cos(\theta)-\dtwo*sin(\theta)},{\atwotip+\done*sin(\theta)+\dtwo*cos(\theta)})
--
plot [domain=0:\tend, variable=\t,samples=50]({4+\aone*(\tend-\t)+\baratio*\bone*(\tend-\t)^2-\t/\tend*\cone-(\tend-\t)/\tend*\done},{\atwo*(\tend-\t)+\baratio*\btwo*(\tend-\t)^2+\t/\tend*\ctwo-(\tend-\t)/\tend*\dtwo});

\filldraw[line width=.02cm, fill=red!40!white] plot [domain=0:\tend, variable=\t,samples=50]({4+\aone*\t+\baratio*\bone*\t^2+\t/\tend*\bonefin},{\atwo*\t+\baratio*\btwo*\t^2+\t/\tend*\btwofin}) --  plot [domain=0:180, variable=\theta,samples=50]({\aonetip+\bonefin*cos(\theta)-\btwofin*sin(\theta)},{\atwotip+\bonefin*sin(\theta)+\btwofin*cos(\theta)}) -- plot [domain=0:\tend, variable=\t,samples=50]({4+\aone*(\tend-\t)+\baratio*\bone*(\tend-\t)^2-(\tend-\t)/\tend*\bonefin},{\atwo*(\tend-\t)+\baratio*\btwo*(\tend-\t)^2-(\tend-\t)/\tend*\btwofin}); 

\draw[line width=.07cm, green,dashed] plot [domain=0:\tend, variable=\t,samples=50]({4+\aone*\t+\baratio*\bone*\t^2},{\atwo*\t+\baratio*\btwo*\t^2});

\fill[fill=green] (\aonetip,\atwotip) circle (0.1cm);

\draw[line width=.04cm,<->,>=stealth] (\aonetip,\atwotip) -- ({\aonetip+\bonefin},{\atwotip+\btwofin});

\draw[line width=.04cm,<->,>=stealth] ({\aonetip+\bonefin},{\atwotip+\btwofin}) -- ({\aonetip+\done},{\atwotip+\dtwo});

\node (irrel) at ({\aonetip+.1*\bonefin},{\atwotip+.2*\btwofin}) [label=right:{\large $\alpha_i(\delta t)$}] {};

\node (irrel) at ({\aonetip+\bonefin+.1*\done)},{\atwotip+\btwofin+.1*\dtwo}) [label=right:{\large $\beta(\delta t)$}] {};


\draw[line width=.04cm,<->,>=stealth] (3,0) -- (3.85,0);
\node (irrel) at (3,0)  [label=left:{\large $\beta(0)$}] {};

\node (irrel) at (\aonetip,{\atwotip-.1})  [label=above:{\large $\mbf{\chi_{\ell}^{[i]}(\delta t)}$}] {};

\draw[line width=.05cm, decoration={zigzag}, decorate,  cyan] (3.2,.5) .. controls (5.5,3) and (1.5,4.7) .. (4,6.3);
\fill[fill=cyan]  (3.2,.5) circle (0.1cm);
\fill[fill=cyan]  (4,6.3) circle (0.1cm);

\node (irrel) at (3.8,6.4) [label=below right:{\large $\mbf{d_{\ell}(\delta t)}$}] {};


\def\aone{0}
\def\atwo{1.2}
\pgfmathsetmacro\bone{\atwo}
\pgfmathsetmacro\btwo{-\aone}
\def\baratio{.05}

\def\aonetip{4+\aone*\tend+\baratio*\bone*\tend^2}
\def\atwotip{\atwo*\tend+\baratio*\btwo*\tend^2}

\pgfmathsetmacro\aonefinaux{\aone+2*\baratio*\bone*\tend}
\pgfmathsetmacro\atwofinaux{\atwo+2*\baratio*\btwo*\tend}

\pgfmathsetmacro\aonefin{\aonefinaux/sqrt(\aonefinaux^2+\atwofinaux^2)}
\pgfmathsetmacro\atwofin{\atwofinaux/sqrt(\aonefinaux^2+\atwofinaux^2)}


\pgfmathsetmacro\cone{\coneaux/sqrt(\coneaux^2-\ctwoaux^2)}
\pgfmathsetmacro\ctwo{\ctwoaux/sqrt(\coneaux^2-\ctwoaux^2)}

\pgfmathsetmacro\done{1.7*\atwofin}
\pgfmathsetmacro\dtwo{-1.7*\aonefin}

\pgfmathsetmacro\phi{atan(\ctwoaux/\coneaux)}

\draw[line width=.03cm, dotted] plot [domain=180-\phi:360+\phi, variable=\theta,samples=50]({4+cos(\theta)},{sin(\theta)}) 
-- 
plot [domain=0:\tend, variable=\t,samples=50]({4+\aone*\t+\baratio*\bone*\t^2+\t/\tend*\done+(\tend-\t)/\tend*\cone},{\atwo*\t+\baratio*\btwo*\t^2+\t/\tend*\dtwo+(\tend-\t)/\tend*\ctwo})
 --  
plot [domain=0:180, variable=\theta,samples=50]({\aonetip+\done*cos(\theta)-\dtwo*sin(\theta)},{\atwotip+\done*sin(\theta)+\dtwo*cos(\theta)})
--
plot [domain=0:\tend, variable=\t,samples=50]({4+\aone*(\tend-\t)+\baratio*\bone*(\tend-\t)^2-\t/\tend*\cone-(\tend-\t)/\tend*\done},{\atwo*(\tend-\t)+\baratio*\btwo*(\tend-\t)^2+\t/\tend*\ctwo-(\tend-\t)/\tend*\dtwo});

\draw[line width=.05cm, red] plot [domain=0:\tend, variable=\t,samples=50]({4+\aone*\t+\baratio*\bone*\t^2},{\atwo*\t+\baratio*\btwo*\t^2});

\filldraw[fill=red, line width=.07cm]  (4,0) circle (0.15cm);
\node (agent 2) at (4,0) [label=below:{\large $\mbf{x_{l_{\ell},G}}$}] {};

\fill[fill=red] (\aonetip,\atwotip) circle (0.1cm);

\node (irrel) at (\aonetip,\atwotip)  [label=left:{\large $\mbf{\chi_{\ell}^{[\ell]}(\delta t)}$}] {};


\def\aone{.3}
\def\atwo{.7}
\pgfmathsetmacro\bone{\atwo}
\pgfmathsetmacro\btwo{-\aone}
\def\baratio{.1}

\def\aonetip{8.5+\aone*\tend+\baratio*\bone*\tend^2}
\def\atwotip{.5+\atwo*\tend+\baratio*\btwo*\tend^2}

\pgfmathsetmacro\aonefinaux{\aone+2*\baratio*\bone*\tend}
\pgfmathsetmacro\atwofinaux{\atwo+2*\baratio*\btwo*\tend}

\pgfmathsetmacro\aonefin{\aonefinaux/sqrt(\aonefinaux^2+\atwofinaux^2)}
\pgfmathsetmacro\atwofin{\atwofinaux/sqrt(\aonefinaux^2+\atwofinaux^2)}

\pgfmathsetmacro\bonefin{1.5*\atwofin}
\pgfmathsetmacro\btwofin{-1.5*\aonefin}

\pgfmathsetmacro\phi{atan(\ctwoaux/\coneaux)}

\filldraw[line width=.02cm, fill=blue!20!white] plot [domain=0:\tend, variable=\t,samples=50]({8.5+\aone*\t+\baratio*\bone*\t^2+\t/\tend*\bonefin},{.5+\atwo*\t+\baratio*\btwo*\t^2+\t/\tend*\btwofin}) --  plot [domain=0:180, variable=\theta,samples=50]({\aonetip+\bonefin*cos(\theta)-\btwofin*sin(\theta)},{\atwotip+\bonefin*sin(\theta)+\btwofin*cos(\theta)}) -- plot [domain=0:\tend, variable=\t,samples=50]({8.5+\aone*(\tend-\t)+\baratio*\bone*(\tend-\t)^2-(\tend-\t)/\tend*\bonefin},{.5+\atwo*(\tend-\t)+\baratio*\btwo*(\tend-\t)^2-(\tend-\t)/\tend*\btwofin}); 

\draw[line width=.07cm, red,dashed] plot [domain=0:\tend, variable=\t,samples=50]({8.5+\aone*\t+\baratio*\bone*\t^2},{.5+\atwo*\t+\baratio*\btwo*\t^2});

\fill[fill=red] (\aonetip,\atwotip) circle (0.1cm);

\node (irrel) at (\aonetip,\atwotip)  [label=right:{\large $\mbf{\chi_{j_1}^{[\ell]}(\delta t)}$}] {};


\def\aone{.45}
\def\atwo{.55}
\pgfmathsetmacro\bone{\atwo}
\pgfmathsetmacro\btwo{-\aone}
\def\baratio{.1}

\def\aonetip{8.5+\aone*\tend+\baratio*\bone*\tend^2}
\def\atwotip{.5+\atwo*\tend+\baratio*\btwo*\tend^2}

\pgfmathsetmacro\aonefinaux{\aone+2*\baratio*\bone*\tend}
\pgfmathsetmacro\atwofinaux{\atwo+2*\baratio*\btwo*\tend}

\pgfmathsetmacro\aonefin{\aonefinaux/sqrt(\aonefinaux^2+\atwofinaux^2)}
\pgfmathsetmacro\atwofin{\atwofinaux/sqrt(\aonefinaux^2+\atwofinaux^2)}

\pgfmathsetmacro\bonefin{1.5*\atwofin}
\pgfmathsetmacro\btwofin{-1.5*\aonefin}

\pgfmathsetmacro\phi{atan(\ctwoaux/\coneaux)}

\draw[line width=.03cm,dotted] plot [domain=0:\tend, variable=\t,samples=50]({8.5+\aone*\t+\baratio*\bone*\t^2+\t/\tend*\bonefin},{.5+\atwo*\t+\baratio*\btwo*\t^2+\t/\tend*\btwofin}) --  plot [domain=0:180, variable=\theta,samples=50]({\aonetip+\bonefin*cos(\theta)-\btwofin*sin(\theta)},{\atwotip+\bonefin*sin(\theta)+\btwofin*cos(\theta)}) -- plot [domain=0:\tend, variable=\t,samples=50]({8.5+\aone*(\tend-\t)+\baratio*\bone*(\tend-\t)^2-(\tend-\t)/\tend*\bonefin},{.5+\atwo*(\tend-\t)+\baratio*\btwo*(\tend-\t)^2-(\tend-\t)/\tend*\btwofin}); 

\draw[line width=.07cm, green,dashed] plot [domain=0:\tend, variable=\t,samples=50]({8.5+\aone*\t+\baratio*\bone*\t^2},{.5+\atwo*\t+\baratio*\btwo*\t^2});

\fill[fill=green] (\aonetip,\atwotip) circle (0.1cm);

\node (irrel) at (\aonetip-.1,\atwotip+.1)  [label=right:{\large $\mbf{\chi_{j_1}^{[i]}(\delta t)}$}] {};


\def\aone{.25}
\def\atwo{.9}
\pgfmathsetmacro\bone{\atwo}
\pgfmathsetmacro\btwo{-\aone}
\def\baratio{.07}

\def\aonetip{8.5+\aone*\tend+\baratio*\bone*\tend^2}
\def\atwotip{.5+\atwo*\tend+\baratio*\btwo*\tend^2}

\draw[line width=.05cm, blue] plot [domain=0:\tend, variable=\t,samples=50]({8.5+\aone*\t+\baratio*\bone*\t^2},{.5+\atwo*\t+\baratio*\btwo*\t^2});

\filldraw[fill=blue, line width=.07cm]  (8.5,.5) circle (0.15cm);
\node (agent 3) at (8.5,.5) [label=below:{\large $\mbf{x_{l_{j_1},G}}$}] {};

\fill[fill=blue] (\aonetip,\atwotip) circle (0.1cm);

\draw[color=magenta,line width=.07cm,->,>=stealth,opacity=.7 ] (3.85,0.08)  -- (0.12,1.47);

\draw[color=magenta,line width=.07cm,->,>=stealth,opacity=.7 ] (8.35,.48) -- (4.15,0.03);

\draw[color=magenta,line width=.07cm,->,>=stealth,opacity=.7 ] (11.85,0.03) --(8.65,.48);

\node (irrel) at (\aonetip,\atwotip)  [label=right:{\large $\mbf{\chi_{j_1}^{[j_1]}(\delta t)}$}] {};
\end{tikzpicture}
\end{center}
\caption{Illustration of the estimates on agent's $i$ reachable set over the interval $[0,\delta t]$ and the reference trajectory deviation bound.}  \label{fig:CC}
\end{figure}

We next provide the definition of a well posed space-time discretization. This definition formalizes the acceptable space and time discretizations through the possibility to assign a feedback law to each agent, in order to enable a meaningful transition from an initial to a final cell in accordance to the Consistency Condition above.

\begin{dfn}\label{well:posed:discretization}
Consider a cell decomposition $\S=\{S_l\}_{l\in\I}$ of $\Rat{n}$, a time step $\delta t$ and a nonempty subset $W$ of $\Rat{n}$.

\noindent \textit{(i)} Given a continuous function $\beta:[0,\delta t]\to\RgeO$ that satisfies \eqref{beta:property}, an agent $i\in\N$, a continuous function $\alpha_i:[0,\delta t]\to\RgeO$ satisfying \eqref{alpha:property}, an initial $m$-cell configuration $\bf{l}_i$ of $i$, and a cell index $l_i'\in\I$, we say that \textit{the transition} $l_i\overset{\bf{l}_i}{\longrightarrow}l_i'$ \textit{is well posed with respect to the space-time discretization $\S-\delta t$}, if there exist a feedback law  $v_i=k_{i,\bf{l}_i}(\cdot,\cdot,\cdot;x_{i0},w_i)$ as in Definition \ref{control:class} that satisfies Property (P), and a vector $w_i\in W$, such that the Consistency Condition of Definition \ref{conditionCC} is fulfilled.

\noindent \textit{(ii)} We say that the \textit{space-time discretization} $\S-\delta t$ \textit{is well posed}, if there exists a continuous function $\beta:[0,\delta t]\to\RgeO$ that satisfies \eqref{beta:property}, such that for each agent $i\in\N$ there exists a continuous function $\alpha_i:[0,\delta t]\to\RgeO$ satisfying \eqref{alpha:property}, in a way that for each cell configuration $\bf{l}_i$ of $i$, there exists a cell index $l_i'\in\I$ such that the transition $l_i\overset{\bf{l}_i}{\longrightarrow}l_i'$ is well posed with respect to $\S-\delta t$.
\end{dfn}

Given a space-time discretization $\S-\delta t$ and based on Definition \ref{well:posed:discretization}(i), it is now possible to provide an exact definition of the discrete transition system which serves as an abstract model for the behaviour of each agent.

\begin{dfn} \label{individual:ts}
For each agent $i$, its individual transition system $TS_i:=(Q_i,Act_i,\longrightarrow_i)$ is defined as follows:

\noindent \textbullet\; $Q_i:=\I$ (the indices of the cell decomposition)

\noindent \textbullet\; $Act_i:=\I^{\bar{N}_i^m}$ (the set of all $m$-cell configurations of $i$)

\noindent \textbullet\; $l_i\overset{\bf{l}_i}{\longrightarrow_i}l_i'$ iff $l_i\overset{\bf{l}_i}{\longrightarrow}l_i'$ is well posed for each
$l_i,l_i'\in Q_i$ and $\bf{l}_i=(l_i,l_{j_1^1},\ldots,l_{j_{N_i}^1},$ $l_{j_1^2},\ldots,l_{j_{N_i^2}^2},$ $\ldots,l_{j_1^m},\ldots,l_{j_{N_i^m}^m})\in\I^{\bar{N}_i^m}$. $\triangleleft$
\end{dfn}

\begin{rem} \label{remark:post:nonempty}
Given a well posed space-time discretization $\S-\delta t$ and an initial cell configuration $\bf{l}=(l_{1},\ldots,l_{N})\in\mathcal{I}^{N}$, it follows from Definitions \ref{well:posed:discretization} and  \ref{individual:ts} that for each agent $i\in\N$ it holds ${\rm Post}_i(l_{i};{\rm pr}_{i}(\bf{l}))\ne\emptyset$ (${\rm Post}_i(\cdot)$ refers to the transition system $TS_i$ of each agent).
\end{rem}

According to Definition \ref{well:posed:discretization}, a well posed space-time discretization requires the existence of a well posed transition for each agent $i$ and $m$-cell configuration of $i$, and the latter reduces to the selection of an appropriate feedback controller for $i$, which also satisfies Property (P) and guarantees that the auxiliary system with disturbances \eqref{system:disturbances} satisfies the Consistency Condition. We next show, that given an initial cell configuration and a well posed transition for each agent, it is possible to choose a local feedback law for each agent, so that the resulting closed-loop system will guarantee all these well posed transitions (for all possible initial conditions in the cell configuration). At the same time, it will follow that the magnitude of the agents' hybrid control laws evaluated along the solutions of the system does not exceed the maximum allowed magnitude $v_{\max}$ of the free inputs on $[0,\delta t]$, and hence, establishes consistency with \eqref{input:bound}.

\begin{prop}\label{discrete:transitions:result}
Consider system \eqref{general:feedback:law}, let $\bf{l}=(l_{1},\ldots,l_{N})\in\mathcal{I}^{N}$ be an initial cell configuration and assume that the space-time discretization $S-\delta t$ is well posed, which according to Remark \ref{remark:post:nonempty} implies that for all $i\in\N$ it holds that ${\rm Post}_i(l_{i};{\rm pr}_{i}(\bf{l}))\ne\emptyset$. Then, for every final cell configuration $\bf{l}'=(l_1',\ldots,l_N')\in{\rm Post}_1(l_1;{\rm pr}_1(\bf{l}))\times\cdots\times{\rm Post}_N(l_N;{\rm pr}_N(\bf{l}))$, there exist feedback laws
\begin{equation} \label{feedback:for:all}
v_i=k_{i,{\rm pr}_i(\bf{l})}(t,x_{i},\bf{x}_j;x_{i0},w_i),i\in\N,
\end{equation}

\noindent satisfying Property (P), $w_1,\ldots,w_N\in W$ and $x_1'\in S_{l_1},\ldots,x_N'\in S_{l_N}$, such that  the solution of the closed-loop system \eqref{general:feedback:law}, \eqref{feedback:for:all} (with $v_{\kappa}=k_{\kappa,{\rm pr}_{\kappa}(\bf{l})}$, $\kappa\in\N$) is well defined on $[0,\delta t]$, and for each $i\in\N$, its $i$-th component satisfies
\begin{equation} \label{contoler:compatibility}
x_{i}(\delta t,x(0))=x_i'\in S_{l_i'}, \forall x(0)\in \Rat{Nn}: x_{\kappa}(0)=x_{\kappa 0}\in S_{l_{\kappa}},\kappa\in\N.
\end{equation}

\noindent Furthermore, it follows that each control law $k_{i,\bf{l}_i}$ evaluated along the corresponding solution of the system satisfies
\begin{equation} \label{feedback:consistency}
|k_{i,{\rm pr}_i(\bf{l})}(t,x_{i}(t),\bf{x}_j(t);x_{i0},w_i)|\le v_{\max},\forall t\in[0,\delta t], i\in\N,
\end{equation}

\noindent which provides the desired consistency with the design requirement \eqref{input:bound} on the $v_i$'s.
\end{prop}

\begin{proof}
Indeed, consider a final cell configuration $\bf{l}'=(l_1',\ldots,l_N')$ as in the statement of the proposition. By the definitions of the operators ${\rm Post}_i(\cdot)$, $i\in\N$ and the transition relation of each corresponding agent's individual transition system, it follows that there exist continuous functions $\beta:[0,\delta t]\to\RgeO$ satisfying \eqref{beta:property} and $\alpha_i:[0,\delta t]\to\RgeO$, $i\in\N$ satisfying  \eqref{alpha:property}, such that each transition $l_i\overset{{\rm pr}_i(\bf{l})}{\longrightarrow}l_i'$ is well posed in the sense of Definition \ref{well:posed:discretization}(i). Hence, we can pick for each agent $i\in\N$ a control law $k_{i,{\rm pr}_i(\bf{l})}(\cdot)$ that satisfies Property (P) and vectors $w_i\in W$, $x_i'\in S_{l_i}'$, such that the requirements of the Consistency Condition are fulfilled.

After the selection of $k_{i,{\rm pr}_i(\bf{l})}(\cdot)$, $w_i$ and $x_i'$, we pick for each agent $i$ an initial condition $x_{i0}\in{\rm int}(S_{l_i})$. Notice, that by virtue of Property (P2) of the control laws, the solution of the closed loop system is defined for all $t\ge 0$. Furthermore, by recalling that for each $i\in\N$ it holds $x_{i0}\in{\rm int}(S_{l_i})$,  we obtain from \eqref{reference:points} that  $|x_{i0}-x_{l_i,G}|<\frac{d_{\max}}{2}$ for all $i\in\N$. Hence, by continuity of the solution of the closed-loop system \eqref{general:feedback:law}, \eqref{feedback:for:all} we deduce from \eqref{beta:property} that there exists $\delta\in(0,\delta t]$, such that for all $i\in\N$ it holds
\begin{equation} \label{xi:bound:uptodelta}
|x_{i}(t)-\chi_{i}^{[i]}(t)|< \beta(t),\forall t\in[0,\delta],
\end{equation}

\noindent where $x_i(\cdot)$ is the $i$-th component of the solution and $\chi_i^{[i]}(\cdot)$ is the reference trajectory of $i$ corresponding to the $m$-cell configuration ${\rm pr}_i(\bf{l})$ of $i$, with initial condition $\chi_i^{[i]}(0)=x_{l_i,G}$. We claim that for each $i\in\N$, \eqref{xi:bound:uptodelta} holds with $\delta=\delta t$. Indeed, suppose on the contrary that there exist an agent $\ell\in\N$ and a time $T\in(0,\delta t]$, such that
\begin{equation} \label{xi:gebeta}
|x_{\ell}(T)-\chi_{\ell}^{[\ell]}(T)|\ge \beta(T).
\end{equation}

\noindent Next, we define
\begin{equation} \label{time:tau:consistency:result}
\tau:=\sup\{\bar{t}\in(0,\delta t]:|x_i(t)-\chi_i^{[i]}(t)|<\beta(t),\forall t\in[0,\bar{t}],i\in\N\}
\end{equation}

\noindent Then, it follows from \eqref{xi:bound:uptodelta}, \eqref{xi:gebeta} and \eqref{time:tau:consistency:result}, that $\tau$ is well defined and satisfies
$$
0<\tau\le\delta t,
$$

\noindent and that there exists $\ell\in\N$ such that
\begin{equation} \label{xiell:attau}
|x_{\ell}(\tau)-\chi_{\ell}^{[\ell]}(\tau)|=\beta(\tau).
\end{equation}

\noindent Next, notice that by the definition of $\tau$, it holds
\begin{equation} \label{xiell:betabound}
|x_{\kappa}(t)-\chi_{\kappa}^{[\kappa]}(t)|\le\beta(t),\forall t\in[0,\tau],\kappa\in\N_{\ell}.
\end{equation}

\noindent Also, since for each $\kappa\in\N_{\ell}$ the $m$-cell configuration ${\rm pr}_{\kappa}(\bf{l})$ of $\kappa$ is consistent with the configuration  ${\rm pr}_{\ell}(\bf{l})$ of $\ell$, it follows from \eqref{alpha:property} and the fact that $\tau<\delta t$, that
\begin{equation} \label{xiell:alphabound}
|\chi_{\kappa}^{[\ell]}(t)-\chi_{\kappa}^{[\kappa]}(t)|\le\alpha_{\ell}(t),\forall t\in[0,\tau],\kappa\in\N_{\ell}.
\end{equation}

\noindent Hence, we obtain that
\begin{equation} \label{xiell:alphabound}
|x_{\kappa}(t)-\chi_{\kappa}^{[\ell]}(t)|\le\beta(t)+\alpha_{\ell}(t),\forall t\in[0,\tau],\kappa\in\N_{\ell}.
\end{equation}

\noindent By setting $d_{\kappa}(t):=x_{\kappa}(t)$, $t\ge 0$, $\kappa\in\N_{\ell}$, it follows from standard uniqueness results from ODE theory, that $x_{\ell}(\cdot)$ is also the solution of the system with disturbances \eqref{system:disturbances}, with $i=\ell$, $v_i=k_{i,\bf{l}_i}(t,x_{i},\bf{d}_j;x_{i0},w_i)$, and initial condition $x_{\ell0}\in{\rm int}(S_{l_{\ell}})$. Thus, by exploiting causality of \eqref{system:disturbances} with respect to the disturbances and the fact that due to \eqref{xiell:alphabound} the disturbances satisfy \eqref{disturbance:bounds} for all $t\in[0,\tau)$, it follows from  \eqref{xi:consistency:bounds} that $|x_{\ell}(\tau)-\chi_{\ell}^{[\ell]}(\tau)|<\beta(\tau)$, which contradicts \eqref{xiell:attau}. Hence, we conclude that \eqref{xi:bound:uptodelta} holds for $\delta=\delta t$.

Next, by using the same arguments as above, we can deduce that for each agent $i$, the $i$-th component of the solution of the closed loop system \eqref{general:feedback:law}, \eqref{feedback:for:all}, is the same as the solution of system \eqref{system:disturbances} for $i$, with disturbances $d_{\kappa}(\cdot)$, $\kappa\in\N_i$ being the components $x_{\kappa}(\cdot)$, $\kappa\in\N_i$ of the solution corresponding to $i$'s neighbors. Furthermore, it follows that the disturbances satisfy \eqref{disturbance:bounds}. Hence, from the Consistency Condition, and the fact that the components of the solution of the closed loop system and the corresponding solutions of the systems with disturbances are identical, we obtain that  \eqref{contoler:compatibility} and \eqref{feedback:consistency} are satisfied.

For the general case where $x_{i0}\in S_{l_i}$ (not necessarily the interior of $S_{l_i}$) for all $i\in\N$, we can exploit for each $i$ continuity of the $i$-th component of the solution $x_i(\cdot)$ with respect to initial conditions and parameters, in order to prove validity of the proposition. In particular, $x_i(t,x_0;x_0,w)$ depends by virtue of Property (P1) continuously on $x_0$ (both as the initial condition and as a parameter) and on the parameters $w=(w_1,\ldots,w_N)\in W^N$. Thus, by selecting a sequence $\{x_{0,\nu}\}_{\nu\in\mathbb{N}}$ with $x_{0,\nu}\to x_0$ and $x_{0i,\nu}\in{\rm int}(S_{l_i})$, $\forall i\in\N,\nu\in\mathbb{N}$, it follows that the components $x_i(\cdot,x_{0,\nu};x_{0,\nu},w)$ of the solution of the closed loop system \eqref{general:feedback:law}, \eqref{feedback:for:all} satisfy
\begin{equation} \label{continuity:property}
x_i(t,x_{0,\nu};x_{0,\nu},w)\to x_i(t,x_0;x_0,w), \forall t\in[0,\delta t],i\in\N,
\end{equation}

\noindent with the vector $w$ of the parameters $w_i$, $i\in\N$ as selected at the beginning of the proof. In addition, since $x_{0i,\nu}\in{\rm int}(S_{l_i})$, $\forall i\in\N$, $\nu\in\mathbb{N}$, it follows from the first part of the proof that the functions $x_i(\cdot,x_{0,\nu};x_{0,\nu},w)$ satisfy \eqref{contoler:compatibility} and \eqref{feedback:consistency}, namely, it holds
$$
x_{i}(\delta t,x_{0,\nu};x_{0,\nu},w)=x_i'\in S_{l_i'}, \forall i\in\N,
$$

\noindent and
$$
|k_{i,{\rm pr}_i(\bf{l})}(t,x_{i}(t,x_{0,\nu};x_{0,\nu},w),\bf{x}_j(t,x_{0,\nu};x_{0,\nu},w);x_{i0},w_i)|\le v_{\max},\forall t\in[0,\delta t], i\in\N.
$$

\noindent Hence, we deduce by virtue of \eqref{continuity:property} that $x_i(\cdot,x_{0};x_{0},w)$, $i\in\N$ also satisfy  \eqref{contoler:compatibility} and \eqref{feedback:consistency}, which establishes the desired result for the general case. The proof is now complete.
\end{proof}

\section{Design of the Hybrid Control Laws}


Consider again system  \eqref{general:feedback:law}. We want to determine sufficient conditions which guarantee that the space-time discretization $\mathcal{S}-\delta t$ is well posed. 
According to Definition \ref{well:posed:discretization}, establishment of a well posed discretization is based on the design of appropriate feedback laws which guarantee well posed transitions for all agents and their possible cell configurations. We postpone the derivation of acceptable $d_{\max}$ and $\delta t$ that fulfill well posedness for the next section and proceed by defining the control laws that are exploited in order to derive well posed discretizations. Consider  a cell decomposition $\mathcal{S}=\{S_{l}\}_{l\in\mathcal{I}}$ of $\Rat{n}$, a time step $\delta t$ and a selection of a reference point $x_{l,G}$ for each cell, satisfying \eqref{reference:points}. For each agent $i$ and $m$-cell configuration $\bf{l}_i$ of $i$, we define the family of feedback laws  $k_{i,\bf{l}_i}:[0,\infty)\times \Rat{(N_i+1)n}\to\Rat{n}$ parameterized by $x_{i0}\in S_{l_i}$ and $w_i\in W$ as
\begin{equation} \label{feedback:ki}
k_{i,\bf{l}_i}(t,x_i,\bf{x}_j;x_{i0},w_i):=k_{i,\bf{l}_i,1}(t,x_i,\bf{x}_j)+k_{i,\bf{l}_i,2}(x_{i0})+k_{i,\bf{l}_i,3}(t;w_i),
\end{equation}

\noindent where
\begin{equation} \label{set:W}
W:=B(v_{\max})\subset\Rat{n},
\end{equation}

\noindent and
\begin{align}
& k_{i,\bf{l}_i,1}(t,x_{i},\bf{x}_j):=f_i(\chi_i(t),\mbf{\chi}_j(t))-f_i(x_i,\bf{x}_j), \label{feedback:ki1} \\
& k_{i,\bf{l}_i,2}(x_{i0}):=\frac{1}{\delta t}(x_{l_i,G}-x_{i0}), \label{feedback:ki2} \\
& k_{i,\bf{l}_i,3}(t;w_i):=\zeta(t)w_i,\label{feedback:ki3:plan} \\
& t\in [0,\infty),(x_i,\bf{x}_j)\in \Rat{(N_i+1)n},x_{i0}\in S_{l_i},w_i\in W, \nonumber \\
&\zeta_i:\RgeO\to[\ubar{\lambda},\bar{\lambda}],0\le\ubar{\lambda}\le\bar{\lambda}<1. \label{zeta:and:lambdas}
\end{align}

\noindent The functions  $\chi_i(\cdot)$ and in $\mbf{\chi}_j(\cdot)$ in \eqref{feedback:ki1} are defined for all $t\ge 0$ through the solution of the initial value problem of Definition \ref{reference:trajectories}, with $x_{l_\ell,G}$ satisfying \eqref{reference:points} for all $\ell\in\bar{\N}_i^m$. In particular, $\chi_i(\cdot)$ constitutes a reference trajectory, whose endpoint agent $i$ should reach at time $\delta t$, when the agent's initial condition lies in $S_{l_i}$ and the feedback $k_{i,\bf{l}_i}(\cdot)$ is applied with $w_i=0$ in \eqref{feedback:ki3:plan}. We also note that the feedback laws $k_{i,\bf{l}_i}(\cdot)$ depend on the cell of agent $i$ and specifically on its $m$-cell configuration $\bf{l}_i$, through the reference point $x_{l_i,G}$ in \eqref{feedback:ki2} and the trajectories $\chi_i(\cdot)$ and $\mbf{\chi}_j(\cdot)$ in  \eqref{feedback:ki1}, as provided by the initial value problem of Definition \ref{reference:trajectories}. The parameters $\ubar{\lambda}$ and $\bar{\lambda}$ in \eqref{zeta:and:lambdas} stand for the minimum and maximum portion of the free input, respectively, that can be further exploited for motion planning. In particular, for each $w_i\in W$ in \eqref{set:W}, the vector $\zeta_i(t)w_i$ provides the ``velocity" of a motion that we superpose to the reference trajectory $\chi_i(\cdot)$ of agent $i$ at time $t\in[0,\delta t]$. The latter allows the agent to reach all points inside a ball with center the position of the reference trajectory at time $\delta t$ by following the curve $\bar{x}_i(t):=x_i(t)+w_i\int_0^t
\zeta_i(s)ds$, as depicted in Fig. \ref{fig:controllers} below. Recall that the reference trajectory of $i$, starting from $x_{l_i,G}$, is obtained from the initial value problem of Definition \ref{reference:trajectories}, by considering the unforced solution of the subsystem formed by agent's $i$ $m$-neighbor set. Then, the feedback term $k_{i,\bf{l}_i,1}(\cdot)$ enforces the agent to move in parallel to its reference trajectory. Also, by selecting a vector $w_i$ in $W$ and a function $\zeta_i(\cdot)$ as in \eqref{zeta:and:lambdas}, we can exploit the extra terms $k_{i,\bf{l}_i,2}(\cdot)$ and $k_{i,\bf{l}_i,3}(\cdot)$ to navigate the agent to the point $x$ inside the ball $B(\chi_i(\delta t);r_i)$ (the dashed circle in Figure 1) at time $\delta t$ from any initial state $x_{i0}\in S_{l_i}$. In a similar way, it is possible to reach any point inside this ball by a different selection of $w_i$. This ball has radius
\begin{equation} \label{distance:r}
r_i:=\int_0^{\delta t}\zeta_i(s)dsv_{\max}\ge\ubar{\lambda}\delta tv_{\max},
\end{equation}

\noindent namely, the distance that the agent can cross in  time $\delta t$ by exploiting $k_{i,\bf{l}_i,3}(\cdot)$, which corresponds to the part of the free input that is available for planning. Hence, it is possible to perform a well posed transition to any cell which has a nonempty intersection with $B(\chi_i(\delta t);r_i)$ (the cyan cells in Fig. \ref{fig:controllers}). Notice that due to the assumption $v_{\max}<M$ below \eqref{input:bound}, it is in principle not possible to cancel the interconnection  terms. An example of a cooperative controller motivating this assumption can be found in our recent paper \cite{BdDd15b}, where appropriate coupling terms ensure robust connectivity of the multi-agent network for bounded free inputs.

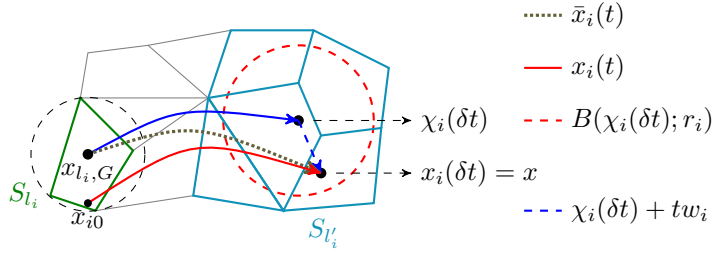
\begin{figure}[H]
\begin{center}

\begin{tikzpicture}[scale=1]

\draw[dashed, color=red,thick] (3,0.5) circle (1cm);

\draw[color=gray] (0.1,0.8) -- (0.2,1.4) -- (1,1.5) -- (1.8,0.8);
\draw[color=gray] (1,1.5) -- (2.1,1.7);
\draw[color=gray] (0.3,-0.7) -- (0.8,0.1) -- (1.8,0.8)-- (1.6,-0.5) -- (0.3,-0.7);
\draw[color=gray] (0.1,0.8) -- (1.8,0.8);
\draw[color=green!50!black,thick] (-0.3,-0.5) -- (0.3,-0.7) -- (0.8,0.1) -- (0.1,0.8)-- (-0.3,-0.5);
\draw[color=cyan!70!black,thick] (1.8,0.8) -- (2.8,-0.7)-- (1.6,-0.5) -- (1.8,0.8);
\draw[color=cyan!70!black,thick] (1.8,0.8)-- (2.8,-0.7)-- (3.3,0.3) -- (3,1) -- (1.8,0.8);
\draw[color=cyan!70!black,thick] (2.8,-0.7)-- (4,-0.6) -- (4.1,0.4) -- (3.3,0.3) ;
\draw[color=cyan!70!black,thick] (1.8,0.8) -- (2.1,1.7) -- (3.2,1.7);
\draw[color=cyan!70!black,thick] (3,1) -- (3.2,1.7) -- (4.2,1.2) -- (4.1,0.4);

\fill[black] (3,0.5) circle (2pt);
\fill[black] (3.3,-0.2) circle (2pt);

\draw[black, dashed] (0.2,0.05) circle (0.7566cm);


\coordinate [label=left:$\textcolor{green!50!black}{S_{l_i}}$] (A) at (-0.3,-0.5);
\coordinate [label=right:$\textcolor{cyan!70!black}{S_{l_i'}}$] (A) at (3,-1);

\draw[dashed,->,>=stealth'] (3.2,0.5) -- (4.5,0.5) node[right] {$\chi_{i}(\delta t)$};
\draw[dashed,->,>=stealth'] (3.5,-0.2) -- (4.5,-0.2) node[right] {$x_{i}(\delta t)=x$};

\draw[color=red,dashed,thick] (6,0.5) -- (6.5,0.5);
\coordinate [label=right:$B(\chi_{i}(\delta t);r_i)$] (A) at (6.5,0.5);

\draw[color=blue,dashed,thick] (6,-0.7) -- (6.5,-0.7);
\coordinate [label=right:$\chi_{i}(\delta t)+tw_i$] (A) at (6.5,-0.7);


\draw[color=red,thick] (6,1.2) -- (6.5,1.2);
\coordinate [label=right:$x_{i}(t)$] (A) at (6.5,1.2);

\draw[color=olive!50!black,very thick,densely dotted] (6,1.9) -- (6.5,1.9);
\coordinate [label=right:$\bar{x}_{i}(t)$] (A) at (6.5,1.9);

\draw[color=blue,dashed,->,thick,>=stealth'] (3,0.5) -- (3.3,-0.2);

\draw[color=blue,->,thick,>=stealth'] (0.2,0.05)  .. controls (1.4,0.7) .. (3,0.5);
\draw[color=olive!50!black,->,very thick,>=stealth',densely dotted] (0.2,0.05)  .. controls (1.6,0.5) .. (3.3,-0.2);
\draw[color=red,->,thick,>=stealth'] (0.2,-0.6) .. controls (1.6,0.3) .. (3.3,-0.2);

\fill[black] (0.2,-0.6) circle (1.5pt) node[below] {$x_{i0}$};
\fill[black] (0.2,0.05) circle (2pt) node[below]{$x_{l_i,G}$};

\end{tikzpicture}

\vspace{-0.4cm}

\end{center}
\caption{Consider any point $x$ inside the ball with center $\chi_i(\delta t)$. Then, for each initial condition $x_{i0}$ in the cell $S_{l_i}$, the endpoint of agent's $i$ trajectory $x_i(\cdot)$ coincides with the endpoint of the curve $\bar{x}_i(\cdot)$, which is precisely $x$, and lies in $S_{l_i'}$, namely, $x_i(\delta t)=\bar{x}_i(\delta t)=x\in S_{l_i'}$.} \label{fig:controllers}
\end{figure}

In order to verify the Consistency Condition for the derivation of well posed discretizations, we will select the function $\beta(\cdot)$ in Definition \ref{conditionCC} as
\begin{equation} \label{beta:dfn}
\beta(t):=\frac{d_{\max}(\delta t-t)}{2\delta t}+\bar{\lambda}v_{\max}t,t\in[0,\delta t]
\end{equation}

\noindent Furthermore, we select a constant $\bar{c}\in(0,1)$, which provides for each agent a measure of the deviation between each reference trajectory of its neighbors and its estimate by the agent (for corresponding consistent cell configurations). By defining
\begin{equation} \label{time:tbar}
\bar{t}:=\sup\left\lbrace t>0:e^{L_2t}-\left(L_2+\bar{c}\frac{L_2^2}{L_1\sqrt{N_{\max}}}\right)t-1<0\right\rbrace,
\end{equation}

\noindent it follows that
\begin{equation} \label{tbar:vs:tstar}
0<\bar{t}<t^*,
\end{equation}

\noindent where $t^*$ is defined in \eqref{time:tstar}, and that the function $H_m(\cdot)$ in \eqref{functions:Hkappa} satisfies
\begin{equation} \label{function:Hm:bound}
H_m(t)\le \bar{c}^{m-1}Mt,\forall t\in[0,\bar{t}]
\end{equation}

\noindent The proof of both \eqref{tbar:vs:tstar} and \eqref{function:Hm:bound} is based on elementary calculations and is provided in the Appendix. Hence, it follows from \eqref{function:Hm:bound} and the result of Proposition \ref{proposition:neighbor:rt:deviation}, namely, \eqref{neighbor:rt:deviation}, that if we select the functions $\alpha_i(\cdot)$ in Definition \ref{dfn:nrtdb} as $\alpha_i(\cdot)\equiv\alpha(\cdot)$, $\forall i\in\N$, with
\begin{equation} \label{alpha:dfn}
\alpha(t):=cMt,\forall t\in [0,\delta t]; \quad c:=\bar{c}^{m-1},
\end{equation}

\noindent then the neighbor reference trajectory deviation bound \eqref{alpha:property} is satisfied for all $0<\delta t\le\bar{t}$.

\begin{rem}
It follows from \eqref{alpha:dfn} that for a fixed constant $\bar{c}\in(0,1)$, the neighbor reference trajectory deviation bound decreases exponentially with respect to the degree of decentralization.
\end{rem}

\section{Well Posed Space-Time Discretizations with Motion Planning Capabilities}

In this section, we exploit the controllers introduced in Section 4 to provide sufficient conditions for well posed space-time discretizations. By exploiting the result of Proposition \ref{discrete:transitions:result} this  framework can be applied for motion planning, by specifying different transition possibilities for each agent through modifying its controller appropriately. As in the previous section, we consider the system \eqref{general:feedback:law}, a cell decomposition $\mathcal{S}=\{S_{l}\}_{l\in\mathcal{I}}$ of $\Rat{n}$, a time step $\delta t$ and a selection of reference points $x_{l,G}$, $l\in\I$ as in \eqref{reference:points}. For each agent $i\in\N$ and $m$-cell configuration $\bf{l}_i\in\I^{\bar{N}_i^m}$ of $i$, consider the family of feedback laws given in \eqref{feedback:ki1}, \eqref{feedback:ki2}, \eqref{feedback:ki3:plan}, and parameterized by $x_{i0}\in S_{l_i}$ and $w_i\in W$. As in the previous sections, $\chi_i(\cdot)$ is the reference solution of the initial value problem corresponding to the $m$-cell configuration of $i$ as in Definition \ref{reference:trajectories}. Recall that the parameters $\ubar{\lambda}$ and $\bar{\lambda}$ in \eqref{zeta:and:lambdas} provide the lower and upper portion of the free input that is exploited for planning. Thus, they can be regarded as a measure for the degree of control freedom that is chosen for the abstraction. We proceed by providing the desired sufficient conditions for a space-time discretization to be suitable for motion planning.

\begin{thm}\label{discretizations:for:planning}
Consider a cell decomposition $\S$ of $\Rat{n}$ with diameter $d_{\max}$, a time step $\delta t$, the constant $r_i$ defined in \eqref{distance:r} and the parameters $\ubar{\lambda}$, $\bar{\lambda}$ in \eqref{zeta:and:lambdas}. We assume that $d_{\max}$ and $\delta t$ satisfy the following restrictions:

\begin{align}
\delta t & \in\left(0,\min\left\lbrace\bar{t},\frac{(1-\ubar{\lambda})v_{\max}}{L_1\sqrt{N_{\max}}(cM+\bar{\lambda}v_{\max})+\ubar{\lambda}L_2v_{\max}}\right\rbrace\right) \label{deltat:interval} \\
d_{\max} & \in\left(0,\min\left\lbrace \frac{2(1-\ubar{\lambda})v_{\max}\delta t}{1+(L_1\sqrt{N_{\max}}+L_2)\delta t}, 2(1-\ubar{\lambda})v_{\max}\delta t \right.\right. \nonumber \\
& \left.\left.\qquad\qquad-2(L_1\sqrt{N_{\max}}(cM+\bar{\lambda}v_{\max})+\ubar{\lambda}L_2v_{\max}) \delta t^2 \right\rbrace\right], \label{dmax:interval}
\end{align}

\noindent with $L_1$, $L_2$, $M$, $v_{\max}$, $c$ and $\bar{t}$ as given in \eqref{dynamics:bound1}, \eqref{dynamics:bound2}, \eqref{dynamics:bound}, \eqref{input:bound}, \eqref{alpha:dfn} and \eqref{time:tbar}, respectively. Then, for each agent $i\in\N$ and cell configuration $\bf{l}_i$ of $i$ we have ${\rm Post}_i(l_i;\bf{l}_i)\ne\emptyset$, namely, the space-time discretization is well posed for the multi-agent system \eqref{general:feedback:law}. In particular, it holds
\begin{equation} \label{planning:condition}
{\rm Post}_i(l_i;\bf{l}_i)\supset\{l\in\I:S_l\cap B(\chi_i(\delta t);r_i)\ne\emptyset\},
\end{equation}

\noindent where $r_i$ is defined in \eqref{distance:r} with
\begin{equation}  \label{zeta:selection}
\zeta_i(t):=\ubar{\lambda}. 
\end{equation}
\end{thm}

\begin{proof}
In order to prove that the discretization is well posed, we will specify, according to Definition \ref{well:posed:discretization}(ii), continuous functions $\beta(\cdot)$ and $\alpha_i(\cdot)$, $i\in\N$, satisfying \eqref{beta:property} and \eqref{alpha:property}, respectively, in such a way that \eqref{planning:condition} holds for all $i\in\N$ and $\bf{l}_i\in\I^{\bar{N}_i^m}$. We pick $\beta(\cdot)$ as in \eqref{beta:dfn} and $\alpha_i(\cdot)\equiv\alpha(\cdot)$, for all $i\in\N$, with $\alpha(\cdot)$ as given in \eqref{alpha:dfn}. Notice first that $\beta(\cdot)$ satisfies \eqref{beta:property}. In addition, due to the requirement that $\delta t\le\bar{t}$ in \eqref{deltat:interval} and the discussion below \eqref{alpha:dfn}, the functions $\alpha_i(\cdot)$ satisfy the reference trajectory deviation bound \eqref{alpha:property}. Thus, it follows that the requirements of a well posed discretization on the mapping $\beta(\cdot)$ and $\alpha_i(\cdot)$ are fulfilled. 

Next, let $i\in\N$ and $\bf{l}_i\in\I^{\bar{N}_i^m}$. For the derivation of \eqref{planning:condition}, namely, that each transition $l_i\overset{\bf{l}_i}{\longrightarrow}l_i'$, with $S_{l_i'}\cap B(\chi_i(\delta t);r_i)\ne\emptyset$ is well posed, it suffices to show that for each $x\in B(\chi_i(\delta t);r_i)$ and $l_i'\in\I$ such that $x\in S_{l_i'}$, the transition $l_i\overset{\bf{l}_i}{\longrightarrow}l_i'$ is well posed. Thus, for each $x\in B(\chi_i(\delta t);r_i)$ and $l_i'\in\I$ such that $x\in S_{l_i'}$, we need according to Definition \ref{well:posed:discretization}(i) to find a feedback law \eqref{feedback:for:i} satisfying Property (P) and a vector $w_i\in W$, in such a way that the Consistency Condition is fulfilled.

Let $x\in B(\chi_i(\delta t);r_i)$ and define
\begin{equation}  \label{vector:wi}
w_i:= \frac{x-\chi(\delta t)}{\ubar{\lambda}\delta t}.
\end{equation}

\noindent with $\ubar{\lambda}$ as in \eqref{zeta:and:lambdas}. Then, it follows from \eqref{distance:r} and \eqref{zeta:selection} that $|w_i|\le \frac{r_i}{\ubar{\lambda}\delta t}\le v_{\max}$ and hence, by virtue of \eqref{set:W} that $w_i\in W$. We now select the feedback law $k_{i,\bf{l}_i}(\cdot)$ as given by \eqref{feedback:ki}, \eqref{feedback:ki1}, \eqref{feedback:ki2}, \eqref{feedback:ki3:plan} and with $w_i$ as defined in \eqref{vector:wi}, and we will show that for all $l_i'\in\I$ such that $x\in S_{l_i'}$ the Consistency Condition is satisfied. Notice first that $k_{i,\bf{l}_i}(\cdot)$ satisfies Property (P). In order to show the Consistency Condition we pick  $x_{i0}\in{\rm int}(S_{l_i})$, $x_i':=x$ with $x$ as selected above and prove that the solution $x_i(\cdot)$ of \eqref{system:disturbances} with $v_i=k_{i,\bf{l}_i}(t,x_i,\bf{d}_j;x_{i0},w_i)$ satisfies \eqref{xi:consistency:bounds}, \eqref{xi:in:finalcell} and \eqref{controler:consistency}, for any continuous $d_{j_1},\ldots,d_{j_{N_i}}:\RgeO\to\Rat{n}$ that satisfy \eqref{disturbance:bounds}. We break the subsequent proof in the following steps.

\noindent \textit{STEP 1:  Proof of \eqref{xi:consistency:bounds} and \eqref{xi:in:finalcell}.} By taking into account \eqref{system:disturbances}, \eqref{feedback:ki}, \eqref{feedback:ki1}-\eqref{feedback:ki3:plan} and \eqref{zeta:selection} we obtain for any continuous $d_{j_1},\ldots,d_{j_{N_i}}:\RgeO\to\Rat{n}$ the solution $x_i(\cdot)$ of \eqref{system:disturbances} with $v_i=k_{i,\bf{l}_i}$ as

\begin{align*}
x_i(t) & =x_{i0}+\int_0^t(f_i(x_i(s),\bf{d}_j(s))+k_{i,\bf{l}_i}(s,x_{i}(s),\bf{d}_j(s);x_{i0},w_i))ds \\
& = x_{i0}+\int_0^t\left(f_i(\chi_i(s),{\mbf{\chi}}_j(s))ds+\frac{1}{\delta t}(x_{l_i,G}-x_{i,0})+\ubar{\lambda}w_i\right)ds \\
& = x_{i0}+\chi_i(t)-x_{l_i,G}+\frac{t}{\delta t}(x_{l_i,G}-x_{i,0})+t\ubar{\lambda}w_i= \chi_i(t)+\frac{\delta t-t}{\delta t}(x_{i0}-x_{l_i,G})+t\ubar{\lambda}w_i,t\ge 0
\end{align*}

\noindent Hence, from the fact that $x_{i0}\in{\rm int}(S_{l_i})$, we deduce from \eqref{reference:points} that
\begin{equation} \label{solution:vs:reftraj}
|x_i(t)-\chi_i(t)|\le \left|\frac{\delta t-t}{\delta t}(x_{i0}-x_{l_i,G})|+t\ubar{\lambda}w_i\right|<\frac{(\delta t-t)d_{\max}}{2\delta t}+t\ubar{\lambda}v_{\max},\forall t\in[0,\delta t],
\end{equation}

\noindent which by virtue of  \eqref{beta:dfn} and \eqref{zeta:and:lambdas} establishes validity of \eqref{xi:consistency:bounds}. Furthermore, we get that $x_i(\delta t)=\chi_i(\delta t)+\delta t\ubar{\lambda}w_i=x=x_i'$ and thus, \eqref{xi:in:finalcell} also holds.

\noindent \textit{STEP 2: Estimation of bounds on $k_{i,\bf{l}_i,1}(\cdot)$, $k_{i,\bf{l}_i,2}(\cdot)$ and $k_{i,\bf{l}_i,3}(\cdot)$ along the solution $x_i(\cdot)$ of \eqref{system:disturbances} with $v_i=k_{i,\bf{l}_i}$ and $d_{j_1},\ldots,d_{j_{N_i}}$ satisfying  \eqref{disturbance:bounds}.} Pick any continuous disturbances $d_{j_1},\ldots,d_{j_{N_i}}$ that satisfy $|d_{j_{\kappa}}(t)-\chi_{j_{\kappa}}^{[i]}(t)|\le \alpha_i(t)+\beta(t)$, for all $t\in[0,\delta t]$. We first show that
\begin{align}
|k_{i,\bf{l}_i,1}(t,x_{i}(t),\bf{d}_j(t))| & \le L_{1}\sqrt{N_{\max}}\left(\frac{d_{\max}(\delta t-t)}{2\delta t}+(cM+\bar{\lambda}v_{\max})t\right) \nonumber \\
& +L_{2}\left(\frac{(\delta t-t)d_{\max}}{2\delta t}+\ubar{\lambda}v_{\max}t\right), \forall t\in[0,\delta t]. \label{ki1:bound}
\end{align}

\noindent Indeed, notice that by virtue of \eqref{feedback:ki1}, we have
\begin{equation}  \label{ki1:equiv}
k_{i,\bf{l}_i,1}(t,x_{i}(t),\bf{d}_j(t))=[f_i(\chi_i(t),\mbf{\chi}_j(t))-f_i(x_i(t),\mbf{\chi}_j(t))]+[f_i(x_i(t),\mbf{\chi}_j(t))-f_i(x_i(t),\bf{d}_j(t))].
\end{equation}

\noindent For the second difference on the right hand side of \eqref{ki1:equiv}, we obtain from \eqref{dynamics:bound1}, \eqref{Nmax}, \eqref{xi:consistency:bounds}, \eqref{beta:dfn} and \eqref{alpha:dfn} that
\begin{align*}
|f_i(x_i(t),\mbf{\chi}_j(t))-f_i(x_i(t),\bf{d}_j(t))| \le & L_{1}|(d_{j_{1}}(t)-\chi_{j_{1}}^{[i]}(t),\ldots,d_{j_{N_i}}(t)-\chi_{j_{N_i}}^{[i]}(t))| \\
\le & L_{1}\left(\sum_{\kappa=1}^{N_i}(\alpha(t)+\beta(t))^{2}\right)^{\frac{1}{2}} \\
\le & L_{1}\sqrt{N_{\max}}(\alpha(t)+\beta(t)) \\
\le & L_{1}\sqrt{N_{\max}}\left(\frac{d_{\max}(\delta t-t)}{2\delta t}+(cM+\bar{\lambda}v_{\max})t\right).
\end{align*}

\noindent For the other difference in \eqref{ki1:equiv}, it follows from \eqref{dynamics:bound2} and the obtained trajectory $x_i(\cdot)$ in Step 1, that
\begin{align*}
|f_i(x_i(t),\mbf{\chi}_j(t))-f_i(\chi_i(t),\mbf{\chi}_j(t))| \le & L_2\left|\left(\chi_i(t)+t\ubar{\lambda}w_i+\left(1-\tfrac{t}{\delta t}\right)(x_{i0}-x_{l_i,G})\right)-\chi_i(t)\right| \\
\le & L_{2}\left(\frac{(\delta t-t)d_{\max}}{2\delta t}+\ubar{\lambda}v_{\max}t\right).
\end{align*}

\noindent Hence, it follows from the evaluated bounds on the differences of the right hand side of \eqref{ki1:equiv} that \eqref{ki1:bound} holds. Next, by recalling that $x_{l_i,G}$ satisfies \eqref{reference:points}, it follows directly from \eqref{feedback:ki2} that
\begin{equation} \label{ki2:bound}
|k_{i,\bf{l}_i,2}(x_{i0})|=\frac{1}{\delta t}|x_{i0}-x_{l_i,G}|\le \frac{d_{\max}}{2\delta t},\forall x_{i0}\in S_{l_i}.
\end{equation}

\noindent Finally, for $k_{i,\bf{l}_i,3}(\cdot)$ we get from \eqref{feedback:ki3:plan}, \eqref{set:W} and \eqref{zeta:selection} that
\begin{equation} \label{ki3:bound:plan}
|k_{i,\bf{l}_i,3}(t;w_i)| = |\ubar{\lambda}w_i|\le \ubar{\lambda}v_{\max},\forall t\in[0,\delta t],w_i\in W.
\end{equation}

\noindent \textit{STEP 3: Verification of \eqref{controler:consistency}.} In this step we exploit the bounds obtained in Step 2 in order to show \eqref{controler:consistency} for any $d_{j_1},\ldots,d_{j_{N_i}}$ satisfying \eqref{disturbance:bounds}. By taking into account \eqref{feedback:ki}, \eqref{ki1:bound}, \eqref{ki2:bound} and \eqref{ki3:bound:plan} we want to prove that

\begin{align}
 L_{1}\sqrt{N_{\max}} & \left(\frac{d_{\max}(\delta t-t)}{2\delta t}+(cM+\bar{\lambda}v_{\max})t\right)+\frac{d_{\max}}{2\delta t} \nonumber \\
 +L_{2} & \left(\frac{(\delta t-t)d_{\max}}{2\delta t}+\ubar{\lambda}v_{\max}t\right)+\ubar{\lambda}v_{\max}\le v_{\max}, \forall t\in[0,\delta t].  \label{condition:dmax:vmax:plan}
\end{align}

\noindent Due to the linearity of the left hand side of \eqref{condition:dmax:vmax:plan} with respect to $t$, it suffices to verify it for $t=0$ and $t=\delta t$. For $t=0$ we obtain that
\begin{align*}
& L_{1}\sqrt{N_{\max}}\frac{d_{\max}}{2}+\frac{d_{\max}}{2\delta t}+L_{2}\frac{d_{\max}}{2}+\ubar{\lambda}v_{\max}\le v_{\max} \iff \\
& L_{1}\sqrt{N_{\max}}\delta t d_{\max}+d_{\max}+L_{2}\delta td_{\max}\le 2(1-\ubar{\lambda})v_{\max}\delta t
\end{align*}

\noindent whose validity is guaranteed by \eqref{dmax:interval}. For the case where $t=\delta t$, we have
\begin{align*}
& L_{1}\sqrt{N_{\max}}(cM+\bar{\lambda}v_{\max})\delta t+\frac{d_{\max}}{2\delta t}+L_{2}\ubar{\lambda}v_{\max}\delta t+\ubar{\lambda}v_{\max}\le v_{\max} \iff \\
& d_{\max}+2(L_{2}\ubar{\lambda}v_{\max}\delta t^2+L_{1}\sqrt{N_{\max}}(cM+\bar{\lambda}v_{\max})\delta t^2)\le  2(1-\ubar{\lambda})v_{\max}\delta t
\end{align*}

\noindent which also holds because of \eqref{dmax:interval}. Hence, we deduce that  \eqref{controler:consistency} is also fulfilled and the proof is complete.
\end{proof}

\begin{rem}
In Fig. \ref{fig:control:bounds} we depict the bounds on the feedback term $k_{i,\bf{l}_i}$ in order to provide some additional intuition behind the selection of $d_{\max}$ and $\delta t$. By rearranging terms in \eqref{condition:dmax:vmax:plan} we get  
\begin{align}
\frac{d_{\max}}{2\delta t}+\frac{d_{\max}}{2\delta t}(L_1\sqrt{N_{\max}}+L_2)(\delta t-t) & +L_1\sqrt{N_{\max}}(cM+\bar{\lambda}v_{\max})+L_{2}\ubar{\lambda}v_{\max})t \nonumber \\
& \le (1-\ubar{\lambda})v_{\max}, \forall t\in[0,\delta t]. \label{condition:dmax:vmax:plan2}
\end{align}

\noindent The third term of the sum in \eqref{condition:dmax:vmax:plan2} is maximized at $\delta t$, namely, becomes $L_1\sqrt{N_{\max}}(cM+\bar{\lambda}v_{\max})+L_{2}\ubar{\lambda}v_{\max})\delta t$ and is independent of the decomposition diameter $d_{\max}$. This term is responsible for the requirement that $\delta t$ is upper bounded by $\frac{(1-\ubar{\lambda})v_{\max}}{L_1\sqrt{N_{\max}}(cM+\bar{\lambda}v_{\max})+\ubar{\lambda}L_2v_{\max}}$ in \eqref{deltat:interval}, since it cannot exceed $ (1-\ubar{\lambda})v_{\max}$. Hence, assuming that we have selected a time step $\delta t$ which satisfies this constraint, we can exploit the remaining part of the input for the manipulation of the first two terms in \eqref{condition:dmax:vmax:plan2}. Notice that both these terms are proportional to $d_{\max}$ which implies that there will always exist a sufficiently small selection of $d_{\max}$ which guarantees consistency with the bound on the available control. In particular, $d_{\max}$ can be selected between zero and the largest value that will result in the full exploitation of the control with magnitude $(1-\ubar{\lambda})v_{\max}$, either at time $0$ or at time $\delta t$ as depicted in Fig. \ref{fig:control:bounds} right and left, respectively. The later holds because of the linearity of the right hand side of \eqref{condition:dmax:vmax:plan2} with respect to $t$ and provides the two possible upper bounds on $d_{\max}$ through the min operator in \eqref{dmax:interval}.
\end{rem}

\begin{figure}[H] 
\begin{center}
\begin{tikzpicture} [scale=1]

\fill[draw=black, fill=green, opacity=0.5] (0,0)  -- (5,0) -- (5,4) -- (0,0);
\fill[draw=black, fill=blue, opacity=0.5] (0,0)  -- (5,4) -- (5,5) -- (0,1) -- (0,0);
\fill[draw=black, fill=magenta, opacity=0.5] (5,5) -- (0,1) -- (0,4.5) -- (5,5);

\draw[line width=.07cm,->,>=stealth ] (-.5,0)  -- (6,0);
\draw[line width=.04cm,dashed] (5,0)  -- (5,6);

\draw[line width=.07cm,->,>=stealth ] (0,-.5)  -- (0,6);
\draw[line width=.04cm,dashed] (0,5)  -- (6,5);

\node (irrelevant) at (2.5,2.5) [label=center:$\frac{d_{\max}}{2\delta t}$] {};

\node (irrelevant) at (3,.5) 
{$\begin{aligned}
\scriptstyle{(L_1\sqrt{N_{\max}}(cM} & \scriptstyle{+\bar{\lambda}v_{\max})}  \\[-.5em]   & \scriptstyle{+L_{2}\ubar{\lambda}v_{\max})t}  \end{aligned}$} [label=center];

\node (irrelevant) at (1.8,4.1) [label=center:$\frac{d_{\max}}{2\delta t}\scriptstyle{(L_1\sqrt{N_{\max}}+L_2)(\delta t-t)}$] {};


\node (irrelevant) at (5,0) [label=below:$\delta t$] {};
\fill[black] (5,0) circle (2pt);
\node (irrelevant) at (0,5) [label=above right:$(1-\ubar{\lambda})v_{\max}$] {};
\fill[black] (0,5) circle (2pt);

\fill[draw=black, fill=green, opacity=0.5] (7,0)  -- (12,0) -- (12,2.5) -- (7,0);
\fill[draw=black, fill=blue, opacity=0.5] (7,0)  -- (12,2.5) -- (12,3.5) -- (7,1) -- (7,0);
\fill[draw=black, fill=magenta, opacity=0.5] (12,3.5) -- (7,1) -- (7,5) -- (12,3.5);

\draw[line width=.07cm,->,>=stealth ] (6.5,0)  -- (13,0);
\draw[line width=.04cm,dashed] (12,0)  -- (12,6);

\draw[line width=.07cm,->,>=stealth ] (7,-.5)  -- (7,6);
\draw[line width=.04cm,dashed] (7,5)  -- (13,5);

\node (irrelevant) at (9.8,1.9) [label=center:$\frac{d_{\max}}{2\delta t}$] {};

\node (irrelevant) at (10.2,.4) 
{$\begin{aligned}
\scriptstyle{(L_1\sqrt{N_{\max}}(cM} & \scriptstyle{+\bar{\lambda}v_{\max})}  \\[-.5em]   & \scriptstyle{+L_{2}\ubar{\lambda}v_{\max})t}  \end{aligned}$} [label=center];

\node (irrelevant) at (8.8,3.5) [label=center:$\frac{d_{\max}}{2\delta t}\scriptstyle{(L_1\sqrt{N_{\max}}+L_2)(\delta t-t)}$] {};


\node (irrelevant) at (12,0) [label=below:$\delta t$] {};
\fill[black] (12,0) circle (2pt);
\node (irrelevant) at (7,5) [label=above right:$(1-\ubar{\lambda})v_{\max}$] {};
\fill[black] (7,5) circle (2pt);

\end{tikzpicture}
\end{center}
\caption{Illustration of the worst case bounds on the feedback terms as given in \eqref{condition:dmax:vmax:plan2}.}  \label{fig:control:bounds}
\end{figure}
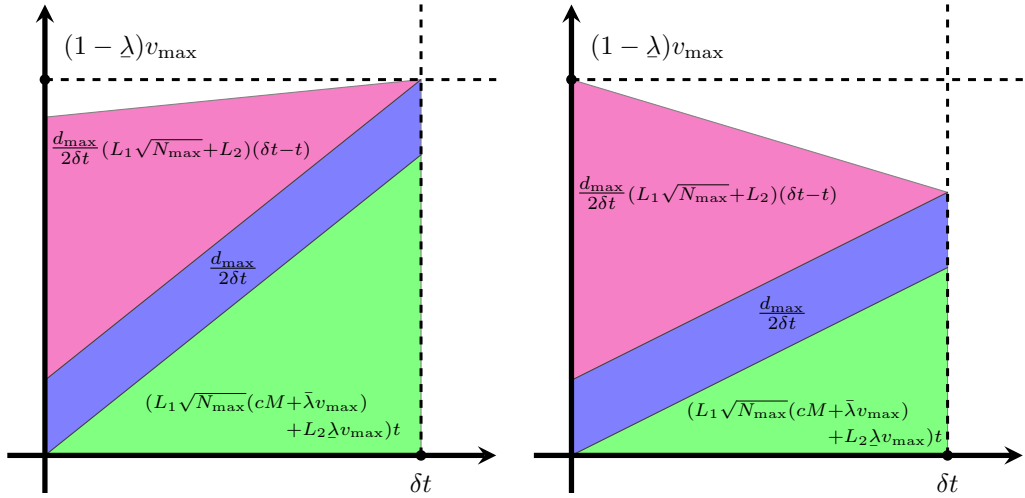

Once a cell decomposition diameter $d_{\max}$ and a time step $\delta t$ which satisfy the restrictions \eqref{deltat:interval} and \eqref{dmax:interval} of Theorem~\ref{discretizations:for:planning} have been selected, it is possible to improve the reachability properties of the abstraction by exploiting the remaining available part of the control. This is established through a modification of the function $\zeta_i(\cdot)$ which was selected as a constant in \eqref{zeta:selection}. The corresponding result is provided in Corollary \ref{corollary:planning} below.

\begin{corollary} \label{corollary:planning}
Assume that the hypotheses of Theorem~\ref{discretizations:for:planning} are satisfied and let any $d_{\max}$ and $\delta t$ that satisfy \eqref{deltat:interval} and \eqref{dmax:interval}. Then, for each agent $i\in\N$ and cell configuration $\bf{l}_i$ of $i$ it holds
\begin{equation} \label{planning:condition}
{\rm Post}_i(l_i;\bf{l}_i)\supset\{l\in\I:S_l\cap B(\chi_i(\delta t);r_i)\ne\emptyset\},
\end{equation}

\noindent where $r_i$ is defined in \eqref{distance:r} with 
\begin{align}
\zeta_i(t):= & \ubar{\lambda}+\xi_i(t),t\in[0,\delta t], \label{zeta:selection:2} \\
\xi_i(t):= & \min\left\lbrace \bar{\lambda}-\ubar{\lambda},\frac{A^L(\delta t-t)+A^R_it}{\delta t v_{\max}(1+t)}\right\rbrace \label{xi:dfn} 
\end{align}

\noindent and the constants $A^L,A^R_i\ge 0$ are given as
\begin{align}
A^L:= & (1-\ubar{\lambda})v_{\max}-L_{1}\sqrt{N_{\max}}\frac{d_{\max}}{2}-\frac{d_{\max}}{2\delta t}-L_{2}\frac{d_{\max}}{2}, \label{ALeft}\\
A^R_i:= & \left\lbrace \begin{array}{l}
(1-\ubar{\lambda})v_{\max}-\frac{d_{\max}}{2\delta t}-L_{2}\ubar{\lambda}v_{\max}\delta t-L_{1}\sqrt{N_{\max}}\bar{\lambda}v_{\max}\delta t, \\
{\rm if}\; \N_i^{m+1}=\emptyset\; {\rm and}\; \N_{\ell}^{m+1}=\emptyset \;\textup{for all}\; \ell\in\N_i, \label{ARight} \\
(1-\ubar{\lambda})v_{\max}-\frac{d_{\max}}{2\delta t}-L_{2}\ubar{\lambda}v_{\max}\delta t-L_{1}\sqrt{N_{\max}}(cM+\bar{\lambda}v_{\max})\delta t, \\
{\rm otherwise}.
\end{array}\right.
\end{align}  
\end{corollary}

\begin{proof}
Notice first that due to \eqref{dmax:interval}, namely that $d_{\max}\le \frac{2(1-\ubar{\lambda})v_{\max}\delta t}{1+(L_1\sqrt{N_{\max}}+L_2)\delta t}$ and $d_{\max}\le 2(1-\ubar{\lambda})v_{\max}\delta t-2(L_1\sqrt{N_{\max}}(cM+\bar{\lambda}v_{\max})+\ubar{\lambda}L_2v_{\max}) \delta t^2$, it follows that $A^L\ge 0$  and $A^R_i\ge 0$ for all $i\in\N$, respectively. The rest of the proof follows the same arguments with those of the proof of Theorem \ref{discretizations:for:planning} which certain modifications that are  provided below. First, the parameter $w_i$ in \eqref{vector:wi} is chosen as
$$
w_i=\frac{x-\chi_i(\delta t)}{\int_0^{\delta t}\zeta_i(s)ds}.
$$

\noindent In addition, when for certain $i\in\N$ it holds $\N_i^{m+1}=\emptyset$ and $\N_{\ell}^{m+1}=\emptyset$ for all $\ell\in\N_i$, it follows from Lemma \ref{lemma:zero:deviation} that  for corresponding consistent cell configurations $\bf{l}_i$ and $\bf{l}_{\ell}$ it holds $\chi_{\ell}^{[i]}(t)=\chi_{\ell}^{[\ell]}(t)$, $\forall t\ge 0$. Thus, we select the neighbor reference trajectory deviation bound $\alpha_i(\cdot)\equiv 0$ in that case and as in \eqref{alpha:dfn} otherwise. 

Due to the above modifications, the bounds on $k_{i,\bf{l}_i,1}(\cdot)$ and $k_{i,\bf{l}_i,3}(\cdot)$ are now given as $|k_{i,\bf{l}_i,1}(t,x_{i}(t),\bf{d}_j(t))| \le L_{1}\sqrt{N_{\max}}(\frac{d_{\max}(\delta t-t)}{2\delta t}+\bar{\lambda}v_{\max})t)+L_{2}(\frac{(\delta t-t)d_{\max}}{2\delta t}+(\ubar{\lambda}+\xi_i(t))v_{\max}t)$, $\forall t\in[0,\delta t]$  when $\N_i^{m+1}=\emptyset$ and $\N_{\ell}^{m+1}=\emptyset$ for all $\ell\in\N_i$, or $|k_{i,\bf{l}_i,1}(t,x_{i}(t),\bf{d}_j(t))| \le L_{1}\sqrt{N_{\max}}(\frac{d_{\max}(\delta t-t)}{2\delta t}+(cM+\bar{\lambda}v_{\max})t)+L_{2}(\frac{(\delta t-t)d_{\max}}{2\delta t}+(\ubar{\lambda}+\xi_i(t))v_{\max}t)$, $\forall t\in[0,\delta t]$ otherwise, and $|k_{i,\bf{l}_i,3}(t,w_i)|\le (\ubar{\lambda}+\xi_i(t))v_{\max} $,$\forall t\in[0,\delta t]$, $w_i\in W$, respectively. Hence, instead of \eqref{condition:dmax:vmax:plan}, we need to verify that
\begin{align}
 L_{1}\sqrt{N_{\max}} & \left(\frac{d_{\max}(\delta t-t)}{2\delta t}+(cM+\bar{\lambda}v_{\max})t\right)+\frac{d_{\max}}{2\delta t} \nonumber \\
 +L_{2} & \left(\frac{(\delta t-t)d_{\max}}{2\delta t}+(\ubar{\lambda}+\xi(t))v_{\max}t\right)+(\ubar{\lambda}+\xi(t))v_{\max}\le v_{\max}, \forall t\in[0,\delta t], \label{condition:dmax:vmax:plan:cor:a}
\end{align}

\noindent when $\N_i^{m+1}=\emptyset$ and $\N_{\ell}^{m+1}=\emptyset$ for all $\ell\in\N_i$, and 
\begin{align}
 L_{1}\sqrt{N_{\max}} & \left(\frac{d_{\max}(\delta t-t)}{2\delta t}+\bar{\lambda}v_{\max}t\right)+\frac{d_{\max}}{2\delta t} \nonumber \\
 +L_{2} & \left(\frac{(\delta t-t)d_{\max}}{2\delta t}+(\ubar{\lambda}+\xi(t))v_{\max}t\right)+(\ubar{\lambda}+\xi(t))v_{\max}\le v_{\max}, \forall t\in[0,\delta t], \label{condition:dmax:vmax:plan:cor:b}
\end{align}

\noindent otherwise. By taking into account \eqref{zeta:selection:2}, \eqref{xi:dfn}, \eqref{ALeft} and \eqref{ARight}, it follows that both \eqref{condition:dmax:vmax:plan:cor:a} and \eqref{condition:dmax:vmax:plan:cor:b} are fulfilled and the proof is complete.
\end{proof}

\begin{rem}
\textit{(i)} The exploitation of the remaining available control in Corollary \ref{corollary:planning}, after a specific space-time discretization has been selected, is illustrated in Fig. \ref{fig:remaining:control}. Specifically, Theorem \ref{discretizations:for:planning} provides for each agent $i$ sufficient conditions on the space-time discretization, which guarantee that the agent can reach a ball of radius $\ubar{\lambda}v_{\max}\delta t$, given that its neighbors may exploit the larger part $\bar{\lambda}v_{\max}$ of their free input for reachability purposes. Since the same reasoning is applied for all the agents, once a discretization has been chosen, the reachable ball of each agent can be increased by exploiting the remaining part $(\bar{\lambda}-\ubar{\lambda})v_{\max}$ of the free input, as long the latter does not violate consistency with the total magnitude of the feedback laws (to be bounded by $v_{\max}$). Thus, the lengths of the left and right bases of the upper trapezoid in Fig. \ref{fig:remaining:control}(left) are  $A^Lv_{\max}$ and  $A^R_iv_{\max}$, respectively, and determine the linearly dependent on $t$ available control over $[0,\delta t]$, as depicted in Fig. \ref{fig:remaining:control}(bottom right). Based on this available control, $\zeta(\cdot)$ is then selected under the conditions that its effect in the control modification does not exceed this threshold, and that it is bounded by $\bar{\lambda}v_{\max}$, as illustrated in Fig. \ref{fig:remaining:control}(up right). 

\noindent \textit{(ii)} Assume that heterogeneous Lipschitz constants $L_1(i)$, $L_2(i)$ and $M(i)$ are given for each agent $i\in\N$, which are upper bounded by $L_1$, $L_2$ and $M$, respectively, as the latter are considered in the report. Then the result of Corollary \ref{corollary:planning} remains valid with the  constants $L_1$, $L_2$, $M$ and $N_{\max}$ in \eqref{ALeft} and \eqref{ARight} being replaced by $L_1(i)$, $L_2(i)$, $M(i)$ and $\N_i$, respectively. Hence, it is possible to obtain improved reachability properties for agents with weaker couplings in their dynamics.
\end{rem}

\begin{figure}[H]
\begin{center}
\begin{tikzpicture} [scale=1]


\fill[draw=black, fill=green, opacity=0.5] (0,0)  -- (5,0) -- (5,3.3) -- (0,0);
\fill[draw=black, fill=blue, opacity=0.5] (0,0)  -- (5,3.3) -- (5,4) -- (0,.7) -- (0,0);
\fill[draw=black, fill=magenta, opacity=0.5] (5,4) -- (0,.7) -- (0,3) -- (5,4);
\fill[draw=black, fill=red, opacity=0.5] (0,3) -- (5,4) -- (5,5) -- (0,5) -- (0,3);
\fill[draw=black, fill=yellow, opacity=0.7] (0,5) -- (5,5) -- (5,6) -- (0,6) -- (0,5);

\node (irrelevant) at (2.5,2) [label=center:$\frac{d_{\max}}{2\delta t}$] {};

\node (irrelevant) at (2.9,.55) 
{$\begin{aligned}
\scriptstyle{(L_1\sqrt{N_{\max}}(cM} & \scriptstyle{+\bar{\lambda}v_{\max})}  
\\[-.5em]   & \scriptstyle{+L_{2}\ubar{\lambda}v_{\max})t} 
\\[-.5em]  \scriptstyle{{\rm or} \;(L_1\sqrt{N_{\max}}\bar{\lambda}v_{\max}} & \scriptstyle{+L_{2}\ubar{\lambda}v_{\max})t} \end{aligned}$} [label=center];

\node (irrelevant) at (1.45,2.6) {$\begin{aligned}
\textstyle{\frac{d_{\max}}{2\delta t}} & \scriptstyle{(L_1\sqrt{N_{\max}}+L_2)}  \\[-.5em] 
\scriptstyle{\times} & \scriptstyle{(\delta t-t)} 
\end{aligned}$} [label=center];

\draw[line width=.07cm,->,>=stealth ] (-.5,0)  -- (6,0);

\draw[line width=.07cm,->,>=stealth ] (0,-.5)  -- (0,7);
\draw[line width=.04cm,dashed] (0,5)  -- (6,5);
\draw[line width=.04cm,dashed] (5,0)  -- (5,7);


\node (irrelevant) at (5,0) [label=below:$\delta t$] {};
\fill[black] (5,0) circle (2pt);
\node (irrelevant) at (0,5) [label=above right:$(1-\ubar{\lambda})v_{\max}$] {};
\fill[black] (0,5) circle (2pt);
\node (irrelevant) at (0,6) [label=left:$v_{\max}$] {};
\fill[black] (0,6) circle (2pt);

\fill[draw=black, fill=red, opacity=0.8] (7,0)  -- (12,0) -- (12,1) -- (7,2) -- (7,0);


\fill[draw=black, fill=yellow, opacity=0.7] (7,4) -- (12,4) -- (12,5) -- (7,5) -- (7,4);

\fill [draw=black, fill=red, opacity=0.7, domain=2:5, variable=\x]
      (7, 5) 
      -- (7, {(2*(1-2/5)+2/5)/(1*(1+2/5))+5})
      -- ({2+7}, {(2*(1-2/5)+2/5)/(1*(1+2/5))+5})
      -- plot ({\x+7}, {(2*(1-\x/5)+\x/5)/(1*(1+\x/5))+5})
      -- (12, 5)
      -- (7, 5);

\fill [draw=black, fill=red, opacity=0.3, domain=0:5, variable=\x]
      (7, 5)
      -- plot ({\x+7}, {(2*(1-\x/5)+\x/5)/(1*(1+\x/5))+5})
      -- (12, 5)
      -- (7, 5);

\node (irrelevant) at (9.5,4.5) [label=center:$\ubar{\lambda}v_{\max}$] {};
\node (irrelevant) at (9.5,5.5) [label=center:$\xi_i(t)v_{\max}$] {};

\draw[line width=.07cm,->,>=stealth ] (6.5,0)  -- (13,0);
\draw[line width=.04cm,dashed] (7,1)  -- (13,1);
\draw[line width=.04cm,dashed] (7,2)  -- (13,2);
\draw[line width=.04cm,dashed] (7,5)  -- (13,5);
\draw[line width=.04cm,dashed] (7,{(2*(1-2/5)+2/5)/(1*(1+2/5))+5})  -- (13,6);

\draw[line width=.07cm,->,>=stealth ] (6.5,4)  -- (13,4);

\draw[line width=.07cm,->,>=stealth ] (7,-.5)  -- (7,3);
\draw[line width=.04cm,dashed] (12,0)  -- (12,3);

\draw[line width=.07cm,->,>=stealth ] (7,3.5)  -- (7,8);


\node (irrelevant) at (12,0) [label=below:$\delta t$] {};
\fill[black] (12,0) circle (2pt);
\node (irrelevant) at (12,4) [label=below:$\delta t$] {};
\fill[black] (12,4) circle (2pt);
\node (irrelevant) at (7,5) [label=above left:$\ubar{\lambda}v_{\max}$] {};
\fill[black] (7,5) circle (2pt);
\node (irrelevant) at (7,{(2*(1-2/5)+2/5)/(1*(1+2/5))+5}) [label=above left:$\bar{\lambda}v_{\max}$] {};
\fill[black] (7,{(2*(1-2/5)+2/5)/(1*(1+2/5))+5}) circle (2pt);
\node (irrelevant) at (7,7) [label=above left:$\zeta_i(0)v_{\max}$] {};
\fill[black] (7,7) circle (2pt);

\node (irrelevant) at (7,1) [label=left:$A^R_iv_{\max}$] {};
\fill[black] (7,1) circle (2pt);
\node (irrelevant) at (7,2) [label=left:$A^Lv_{\max}$] {};
\fill[black] (7,2) circle (2pt);

\end{tikzpicture}
\end{center}
\caption{Illustration of the bounds on the feedback terms and the exploitation of the remaining available control through the modification of $\zeta_i$ in \eqref{zeta:selection:2}-\eqref{ARight}.}  \label{fig:remaining:control}
\end{figure}
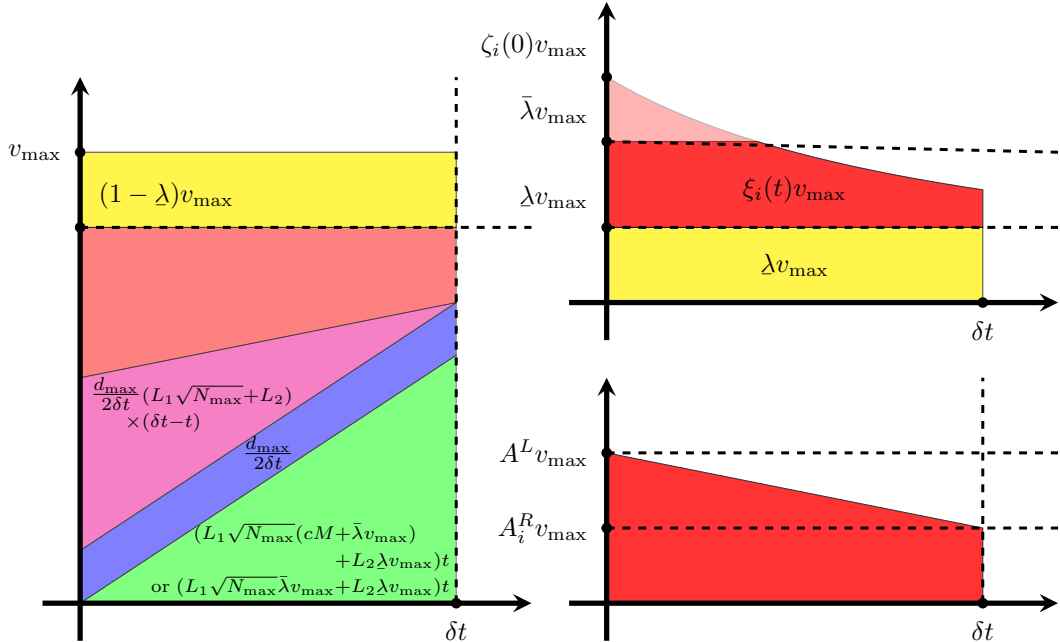

We next provide an improved version of Theorem \ref{discretizations:for:planning}, for the case where the conditions of Lemma \ref{lemma:zero:deviation} are satisfied for all agents, namely, when for every possible $m$-cell configuration of each agent, the estimate of its neighbors' reference trajectories and their reference trajectories corresponding to consistent configurations coincide.

\begin{thm} \label{discretizations:for:planning:zero:dev}
Assume that $\N_i^{m+1}=\emptyset$ holds for all $i\in\N$. Then, the result of Theorem \ref{discretizations:for:planning} remains valid for any $\delta t$ and $d_{\max}$ satisfying
\begin{align}
\delta t & \in\left(0,\frac{(1-\ubar{\lambda})v_{\max}}{L_1\sqrt{N_{\max}}\bar{\lambda}v_{\max}+\ubar{\lambda}L_2v_{\max}}\right) \label{deltat:interval2} \\
d_{\max} & \in\left(0,\min\left\lbrace \frac{2(1-\ubar{\lambda})v_{\max}\delta t}{1+(L_1\sqrt{N_{\max}}+L_2)\delta t}, 2(1-\ubar{\lambda})v_{\max}\delta t \right.\right. \nonumber \\
& \left.\left.\qquad\qquad\qquad-2(L_1\sqrt{N_{\max}}\bar{\lambda}v_{\max}+\ubar{\lambda}L_2v_{\max}) \delta t^2 \right\rbrace\right]. \label{dmax:interval2}
\end{align}
\end{thm}

\begin{proof}
By exploiting Lemma \ref{lemma:zero:deviation} and the fact that according to the hypothesis of the theorem it holds $\N_i^{m+1}=\emptyset$, $\N_{\ell}^{m+1}=\emptyset$ for each pair of agents $i\in\N$, $\ell\in\N_i$, it follows that for corresponding consistent cell configurations $\bf{l}_i$ and $\bf{l}_{\ell}$ it holds $\chi_{\ell}^{[i]}(t)=\chi_{\ell}^{[\ell]}(t)$, $\forall t\ge 0$.
Thus, we can select for each agent $i\in\N$ the neighbor reference trajectory deviation bound $\alpha_i(\cdot)\equiv 0$. The rest of the proof follows the same arguments employed for the proof of Theorem \ref{discretizations:for:planning} and is therefore omitted.
\end{proof}

As in Corollary \ref{corollary:planning}, the following corollary presents the analogous improved reachability result when $d_{\max}$ and $\delta t$ are selected by Theorem \ref{discretizations:for:planning:zero:dev}.

\begin{corollary} \label{corollary:planning2}
Assume that the hypotheses of Theorem \ref{discretizations:for:planning:zero:dev} are fulfilled and let any $\delta t$ and $d_{\max}$ that satisfy \eqref{deltat:interval2} and \eqref{dmax:interval2}. Then, the result of Corollary \ref{corollary:planning} is valid with $A_i^R=A^R:=(1-\ubar{\lambda})v_{\max}-\frac{d_{\max}}{2\delta t}-L_{2}\ubar{\lambda}v_{\max}\delta t-L_{1}\sqrt{N_{\max}}\bar{\lambda}v_{\max}\delta t$.
\end{corollary}

\section{Example and Simulation Results}

As an illustrative example we consider a system of four agents in $\Rat{2}$. Their dynamics are given as
\begin{align*}
\dot{x}_1 & = {\rm sat}_{\rho}(x_2-x_1)+v_1,  \\
\dot{x}_2 & = v_2  \\
\dot{x}_3 & = {\rm sat}_{\rho}(x_2-x_3)+v_3 \\
\dot{x}_4 & = {\rm sat}_{\rho}(x_3-x_4)+v_4
\end{align*}

\noindent where the saturation function ${\rm sat}_{\rho}:\Rat{2}\to\Rat{2}$ is defined as
\begin{equation*}
{\rm sat}_{\rho}(x):=\left\lbrace
\begin{array}{ll}
x, &  {\rm if}\; x<\rho \\
\frac{\rho}{|x|}x, &  {\rm if}\; x\ge\rho
\end{array}
\right.,x\in\Rat{2}.
\end{equation*}
\noindent The agents' neighbors' sets in this example are $\N_1=\{2\}$, $\N_2=\emptyset$, $\N_3=\{2\}$, $\N_4=\{3\}$ and specify the corresponding network topology. The constant $\rho>0$ above represents a bound on the distance between agents 1, 2, and agents 2, 3, that we will require the system to satisfy during its evolution. It is not hard to show that this is possible if we select $v_{\max}:=\frac{\rho}{2}$. Hence we obtain the dynamics bounds and Lipschitz constants $M:=\rho$, $L_1:=1$ and $L_2:=1$ for all agents. By selecting the degree of decentralization $m=2$, it follows that the conditions of Theorem \ref{discretizations:for:planning:zero:dev} are satisfied. By choosing the constant $\bar{\lambda}=1$ in \eqref{zeta:and:lambdas}  we obtain from \eqref{dmax:interval2} and \eqref{deltat:interval2} (notice also that $N_{\max}=1$) that
$$
0<\delta t <\frac{1-\ubar{\lambda}}{1+\ubar{\lambda}}  \\
$$

\noindent and
$$
0<d_{\max} \le\min\left\lbrace \frac{(1-\ubar{\lambda})\delta t}{1+2\delta t}\rho, ((1-\ubar{\lambda})\delta t-(1+\ubar{\lambda}) \delta t^2)\rho\right\rbrace.
$$

\noindent By picking $\delta t=\frac{1-\ubar{\lambda}}{2(1+\ubar{\lambda})}$ we obtain that
$$
0<d_{\max}\le\frac{(1-\ubar{\lambda})^2}{4(1+\ubar{\lambda})}\rho.
$$

\noindent For the simulation results we pick $\ubar{\lambda}=0.4$, $\rho=10$ and evaluate the corresponding value of $\delta t$ above and the maximum acceptable value for the diameter $d_{\max}$ of the cells. We partition the workspace $[-10,10]\times[-10,10]\subset\Rat{2}$ into squares and select the initial conditions of the agents 1, 2, 3 and 4 at $(9,4)$, $(4,4)$, $(-6,6)$ and $(-9,-4)$, respectively, which respect the desired initial bound on the distance between agents 1, 2, and agents 2, 3. We assume that agent 2, which is unaffected by the coupled constraints has constant velocity $v_2=(-1,-4)$ and study reachability properties of the system over the time interval $[0,2]$. The sampled trajectory of agent 2 is visualized with the circles in the figure below, and the blue, magenta cells indicate the union of agent's 3 and 1 reachable cells over the time interval, respectively, given the trajectory of 2. Finally,
given the trajectory of 2 and selecting the discrete trajectory of 3 which is depicted with the red cells in the figure, we obtain with yellow the corresponding reachable cells of agent 4. The simulation results have been implemented in MATLAB with a running time of the order of a few seconds, on a PC with an Intel(R) Core(TM) i7-4600U CPU @ 2.10GHz processor.

\begin{figure}[H]
\includegraphics[width =0.7\textwidth]{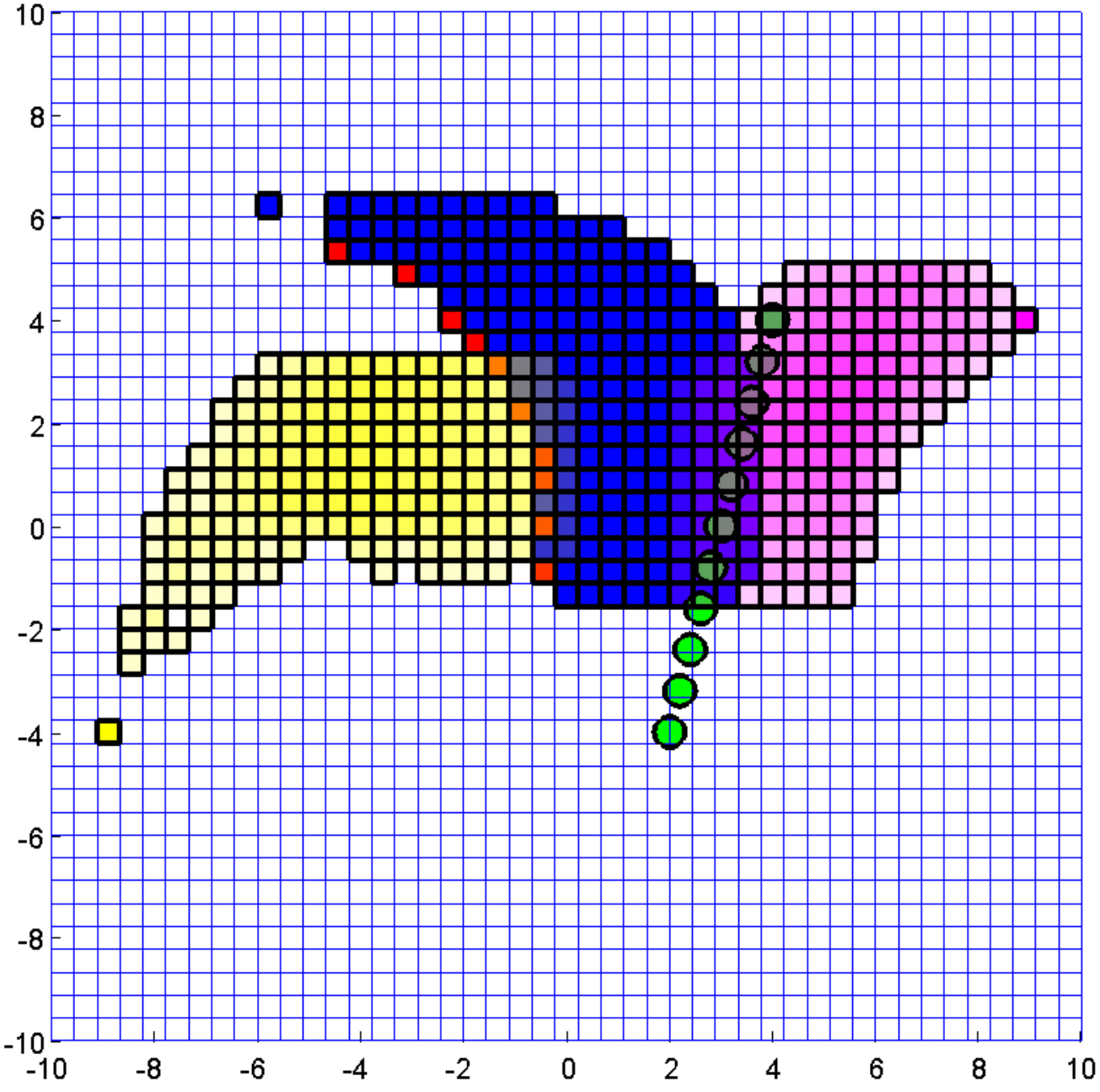}
\caption{Reachable cells of agents 1, 3 and 4, given the trajectory of agent 2 in green.}
\end{figure}

\section{Conclusions}

We have provided an abstraction framework for multi-agent systems which is based on a varying degree of decentralization for the information of each agent with respect to the graph topology of the network. Well posed transitions have been modeled by exploiting a system with disturbances in order to model the possible evolution of the agents neighbors. Sufficient conditions on the space and time discretization have been quantified in order to capture the reachability properties of the symbolic models through well defined transitions. The latter are realized by means of hybrid feedback control laws which take into account the coupled constraints and navigate the agents to their successor cells.

Ongoing work includes the improvement of the agents' reachability properties, based on their local dynamics bounds and Lipschitz constants which will enable the exploitation of a larger part of the free inputs for the transitions to successor cells. In addition we aim at the formulation of an online abstraction framework for heterogeneous agents with updated choices of the discretization and planning parameters.

\section{Appendix}

\subsection{Proof of Lemma \ref{lemma:reftraj:worst:case:dev}.}
\begin{proof}
For the proof, it suffices to show that \eqref{reftraj:deviation:all}  holds for all agents $\kappa\in\bar{\N}_i^m\cap\bar{\N}_{\ell}^m$, since by Lemma~\ref{lemma:mneighbor:contains:mmin1}(i) we have $\bar{\N}_{\ell}^{m-1}\subset\bar{\N}_i^m$. We distinguish the following cases.

\noindent \textit{Case (i).}  $\N_i^{m+1}\ne\emptyset$ and  $\N_{\ell}^{m+1}\ne\emptyset$. For Case (i) we consider the following subcases for each agent $\kappa\in\bar{\N}_i^m\cap\bar{\N}_{\ell}^m$.

\noindent \textit{Case (ia).} $\kappa\in(\bar{\N}_i^m\cap\bar{\N}_{\ell}^m)\cap(\N_i^m\cup\N_{\ell}^m)$. In this case, it follows from \eqref{reference:IVP2:constant:terms} that either  $\chi_{\kappa}^{[i]}(\cdot)\equiv x_{l_{\kappa},G}$ or $\chi_{\kappa}^{[\ell]}(\cdot)\equiv x_{l_{\kappa},G}$. Without any loss of generality we assume that $\kappa\in\bar{\N}_{\ell}^{m-1}$ and thus, that $\chi_{\kappa}^{[i]}(\cdot)\equiv x_{l_{\kappa},G}$. Then, we get from \eqref{dynamics:bound}, the initial value problem which specifies $\chi_{\kappa}^{[\ell]}(\cdot)$, as given by \eqref{reference:IVP2}, and the consistency of $\bf{l}_{\ell}$ with $\bf{l}_i$, which implies $\chi_{\kappa}^{[i]}(0)=\chi_{\kappa}^{[\ell]}(0)=x_{l_{\kappa},G}$, that
\begin{equation}  \label{reftraj:deviation:CaseA}
|\chi_{\kappa}^{[i]}(t)-\chi_{\kappa}^{[\ell]}(t)|=|x_{l_{\kappa},G}-\chi_{\kappa}^{[\ell]}(t)|\le \int_0^t|f_{\kappa}(\chi_{\kappa}^{[\ell]}(s),\mbf{\chi}_{j(\kappa)}^{[\ell]}(s))|ds\le Mt,\forall t\ge 0.
\end{equation}

\noindent \textit{Case (ib).} $\kappa\in(\bar{\N}_i^m\cap\bar{\N}_{\ell}^m)\setminus(\N_i^m\cup\N_{\ell}^m)$. Notice first, that $(\bar{\N}_i^m\cap\bar{\N}_{\ell}^m)\setminus(\N_i^m\cup\N_{\ell}^m)\subset\bar{\N}_i^{m-1}\cap\bar{\N}_{\ell}^{m-1}$ and thus, for each agent $\kappa\in(\bar{\N}_i^m\cap\bar{\N}_{\ell}^m)\setminus(\N_i^m\cup\N_{\ell}^m)$ we have that $\kappa\in\bar{\N}_i^{m-1}$ and $\kappa\in\bar{\N}_{\ell}^{m-1}$. Hence, we obtain from Lemma~\ref{lemma:mneighbor:contains:mmin1}(ii) that $\N_{\kappa}\subset\bar{\N}_i^m$ and $\N_{\kappa}\subset\bar{\N}_{\ell}^m$, respectively, implying that $\N_{\kappa}\subset\bar{\N}_i^m\cap\bar{\N}_{\ell}^m$. Consequently, it follows from Definition~\ref{reference:trajectories} that both $\mbf{\chi}_{j(\kappa)}^{[i]}(\cdot)$ and $\mbf{\chi}_{j(\kappa)}^{[\ell]}(\cdot)$ are well defined. In order to show \eqref{reftraj:deviation:all} for all $\kappa$ of Case~(ib) we will prove the following claim.

\noindent \textit{Claim I}. There exists a $\delta\in(0,t^*)$, such that \eqref{reftraj:deviation:all} holds for all $t\in [0,\delta]$ and all $\kappa$ of Case (ib).

\noindent In order to show Claim I, we select 
\begin{equation} \label{time:delat}
\delta:=\min\left\lbrace \frac{1}{4L_1\sqrt{N_{\max}}},\frac{\ln2}{L_2}\right\rbrace
\end{equation}

\noindent and any $\kappa\in(\bar{\N}_i^m\cap\bar{\N}_{\ell}^m)\setminus(\N_i^m\cup\N_{\ell}^m)$. By recalling that $\mbf{\chi}_{j(\kappa)}^{[i]}(\cdot)$ and $\mbf{\chi}_{j(\kappa)}^{[\ell]}(\cdot)$ are well defined and exploiting the dynamics bound \eqref{dynamics:bound}, \eqref{Nmax}, \eqref{time:delat}, and that $\bf{l}_{\ell}$ is consistent with $\bf{l}_i$, which by virtue of the fact that  $\N_{\kappa}\subset\bar{\N}_i^m\cap\bar{\N}_{\ell}^m$ implies that $\chi_{\nu}^{[\ell]}(0)=\chi_{\nu}^{[i]}(0)=x_{l_{\nu},G}$ for all $\nu\in\N_{\kappa}$, we deduce that
\begin{align} \label{neighbor:difference:bound1}
|\mbf{\chi}_{j(\kappa)}^{[i]}(t)-\mbf{\chi}_{j(\kappa)}^{[\ell]}(t)| & =\left(\sum_{\nu\in\N_{\kappa}}|\chi_{\nu}^{[i]}(t)-\chi_{\nu}^{[\ell]}(t)|^2\right)^{\frac{1}{2}}\le \left(\sum_{\nu\in\N_{\kappa}}(2Mt)^2\right)^{\frac{1}{2}} \nonumber \\
& \le 2M\sqrt{N_{\max}}t\le\frac{M}{2L_1},\forall t\in[0,\delta].
\end{align}

\noindent Next, by virtue of \eqref{dynamics:bound1} and \eqref{dynamics:bound2} we evaluate the difference
\begin{align}
|\chi_{\kappa}^{[i]}(t)-\chi_{\kappa}^{[\ell]}(t)| & \le \int_0^t|f_{\kappa}(\chi_{\kappa}^{[i]}(s),\mbf{\chi}_{j(\kappa)}^{[i]}(s))-f_{\kappa}(\chi_{\kappa}^{[\ell]}(s),\mbf{\chi}_{j(\kappa)}^{[\ell]}(s))|ds \nonumber \\
& \le \int_0^t(|f_{\kappa}(\chi_{\kappa}^{[i]}(s),\mbf{\chi}_{j(\kappa)}^{[i]}(s))-f_{\kappa}(\chi_{\kappa}^{[\ell]}(s),\mbf{\chi}_{j(\kappa)}^{[i]}(s))| \nonumber \\
& +|f_{\kappa}(\chi_{\kappa}^{[\ell]}(s),\mbf{\chi}_{j(\kappa)}^{[i]}(s))-f_{\kappa}(\chi_{\kappa}^{[\ell]}(s),\mbf{\chi}_{j(\kappa)}^{[\ell]}(s))|)ds \nonumber \\
& \le \int_0^t (L_2|\chi_{\kappa}^{[i]}(s)-\chi_{\kappa}^{[\ell]}(s)|+L_1|\mbf{\chi}_{j(\kappa)}^{[i]}(s)-\mbf{\chi}_{j(\kappa)}^{[\ell]}(s)|)ds. \label{tildexi:difference:integral:ineq}
\end{align}

\noindent Hence, from \eqref{reftraj:deviation:all} and \eqref{neighbor:difference:bound1} we obtain that for all $t\in[0,\delta]$ it holds
\begin{equation} \label{tildexi:difference:integral:ineq2}
|\chi_{\kappa}^{[i]}(t)-\chi_{\kappa}^{[\ell]}(t)| \le \frac{Mt}{2}+\int_0^tL_2|\chi_{\kappa}^{[i]}(s)-\chi_{\kappa}^{[\ell]}(s)| ds.
\end{equation}

\noindent In order to derive a bound for $|\chi_{\kappa}^{[i]}(\cdot)-\chi_{\kappa}^{[\ell]}(\cdot)|$ from the integral inequality above, we will use the following version of the Gronwall Lemma, whose proof is given later in the Appendix.

\noindent \textit{Fact I.} Let $\lambda:[a,b]\to\Rat{}$ be a continuously differentiable function satisfying $\lambda(a)=0$ and $\mu$ a nonnegative constant. If a continuous function $y(\cdot)$ satisfies
$$
y(t)\le \lambda(t)+\int_a^t\mu y(s)ds,
$$

\noindent then, on the same interval it holds
$$
y(t)\le \int_a^te^{\mu(t-s)}\dot{\lambda}(s)ds. \quad \triangleleft
$$

\noindent By exploiting Fact I, we obtain from \eqref{tildexi:difference:integral:ineq2} and \eqref{time:delat} that
\begin{equation} \label{tildexi:difference:ineq}
|\chi_{\kappa}^{[i]}(t)-\chi_{\kappa}^{[\ell]}(t)| \le\int_0^te^{L_2(t-s)}\frac{M}{2}ds\le\frac{M}{2}e^{L_2t}t\le\frac{M}{2}e^{L_2\frac{\ln2}{L_2}}t=Mt,\forall t\in[0,\delta].
\end{equation}

\noindent Hence, we have proved Claim I. $\triangleleft$

\noindent We proceed by showing that  \eqref{reftraj:deviation:all} also holds for all $t\in[0,t^*]$ and $\kappa$ of Case (ib).  Assume on the contrary that there exist $\kappa'\in(\bar{\N}_i^m\cap\bar{\N}_{\ell}^m)\setminus(\N_i^m\cup\N_{\ell}^m)$ and $T\in(0,t^*]$ such that
\begin{equation} \label{tildexi:difference:atT}
|\chi_{\kappa'}^{[i]}(T)-\chi_{\kappa'}^{[\ell]}(T)|>MT.
\end{equation}

\noindent By continuity of the functions $\chi_{\kappa}^{[i]}(\cdot)$, we can define
\begin{equation} \label{time:tau}
\tau:=\max\{\bar{t}\in[0,T]:|\chi_{\kappa}^{[i]}(t)-\chi_{\kappa}^{[\ell]}(t)|\le Mt,\forall t\in [0,\bar{t}],\kappa\in(\bar{\N}_i^m\cap\bar{\N}_{\ell}^m)\setminus(\N_i^m\cup\N_{\ell}^m)\}.
\end{equation}

\noindent Then, it follows from \eqref{tildexi:difference:ineq}, \eqref{tildexi:difference:atT} and continuity of the functions $\chi_{\kappa}^{[i]}(\cdot)$, $\chi_{\kappa}^{[\ell]}(\cdot)$, that there exists $\kappa''\in(\bar{\N}_i^m\cap\bar{\N}_{\ell}^m)\setminus(\N_i^m\cup\N_{\ell}^m)$ such that
\begin{equation} \label{tildexi:difference:attau}
|\chi_{\kappa''}^{[i]}(\tau)-\chi_{\kappa''}^{[\ell]}(\tau)|=M\tau.
\end{equation}

\noindent Also, from Claim I and \eqref{time:tau} we get that
\begin{equation} \label{tau:vs:deltat}
0<\tau<t^*.
\end{equation}

\noindent Next, by recalling that since $\kappa''\in(\bar{\N}_i^m\cap\bar{\N}_{\ell}^m)\setminus(\N_i^m\cup\N_{\ell}^m)$, it holds $\N_{\kappa''}\subset\bar{\N}_i^m\cap\bar{\N}_{\ell}^m$, it follows that for each neighbor $\nu\in\N_{\kappa''}$ of $\kappa''$ either  $\nu\in(\bar{\N}_i^m\cap\bar{\N}_{\ell}^m)\cap(\N_i^m\cup\N_{\ell}^m)$, or $\nu\in(\bar{\N}_i^m\cap\bar{\N}_{\ell}^m)\setminus(\N_i^m\cup\N_{\ell}^m)$. Thus, we deduce from \eqref{reftraj:deviation:CaseA} and  \eqref{time:tau}, respectively, that
\begin{equation} \label{tildexj:difference:case1}
|\chi_{\nu}^{[i]}(t)-\chi_{\nu}^{[\ell]}(t)|\le Mt, \forall t\in[0,\tau],\nu\in\N_{\kappa''}.
\end{equation}

\noindent It then follows from \eqref{tildexj:difference:case1} (by similar calculations as in \eqref{tildexi:difference:integral:ineq}) that
\begin{align*}
|\chi_{\kappa''}^{[i]}(\tau)-\chi_{\kappa''}^{[\ell]}(\tau)| & \le \int_0^{\tau} (L_2|\chi_{\kappa''}^{[i]}(s)-\chi_{\kappa''}^{[\ell]}(s)|+L_1|\mbf{\chi}_{j(\kappa'')}^{[i]}(s)-\mbf{\chi}_{j(\kappa'')}^{[\ell]}(s)|)ds \\
& = \int_0^{\tau} \left((L_2|\chi_{\kappa''}^{[i]}(s)-\chi_{\kappa''}^{[\ell]}(s)|+L_1\left(\sum_{\nu\in\N_{\kappa''}}|\chi_{\nu}^{[i]}(s)-\chi_{\nu}^{[\ell]}(s)|^2\right)^{\frac{1}{2}}\right)ds \\
& \le \int_0^{\tau} L_2|\chi_{\kappa''}^{[i]}(s)-\chi_{\kappa''}^{[\ell]}(s)|ds +\int_0^{\tau}L_1M\sqrt{N_{\kappa''}}sds.
\end{align*}

\noindent Hence, we obtain from Fact I and \eqref{Nmax} that
\begin{align}
|\chi_{\kappa''}^{[i]}(\tau)-\chi_{\kappa''}^{[\ell]}(\tau)| & \le \int_0^{\tau} e^{{L_2(\tau-s)}}L_1M\sqrt{N_{\kappa''}}sds \nonumber \\
& \le \int_0^{\tau} e^{{L_2(\tau-s)}}L_1M\sqrt{N_{\max}}sds  \nonumber \\
& =\frac{L_1}{L_2}M\sqrt{N_{\max}}\left(\frac{e^{L_2\tau}}{L_2}-\tau-\frac{1}{L_2}\right). \label{tildexi:difference:attau:ineq}
\end{align}

\noindent It can then be checked by elementary calculations that
\begin{equation} \label{exponential:inequality}
e^{L_2t}-\left(L_2+\frac{L_2^2}{L_1\sqrt{N_{\max}}}\right)t-1<0,\forall t\in(0,t^*),
\end{equation}

\noindent with $t^*$ as specified in \eqref{time:tstar}. Also, from
\eqref{tildexi:difference:attau} and \eqref{tildexi:difference:attau:ineq} we have that
\begin{align*}
M\tau & \le\frac{L_1}{L_2}M\sqrt{N_{\max}}\left(\frac{e^{L_2\tau}}{L_2}-\tau-\frac{1}{L_2}\right) \iff \\
\frac{L_2^2}{L_1\sqrt{N_{\max}}}\tau & \le e^{L_2\tau}-L_2\tau-1 \iff \\
e^{L_2\tau} & -\left(L_2+\frac{L_2^2}{L_1\sqrt{N_{\max}}}\right)\tau-1\ge 0,
\end{align*}

\noindent which contradicts \eqref{exponential:inequality}, since by \eqref{tau:vs:deltat}, it holds $0<\tau<t^*$. Thus, we have shown \eqref{reftraj:deviation:all} for Case (ib).

\noindent \textit{Case (ii).}  $\N_i^{m+1}\ne\emptyset$ and  $\N_{\ell}^{m+1}=\emptyset$. For Case (ii) we also consider the following subcases for each agent $\kappa\in\bar{\N}_i^m\cap\bar{\N}_{\ell}^m$.

\noindent \textit{Case (iia).} $\kappa\in(\bar{\N}_i^m\cap\bar{\N}_{\ell}^m)\cap\N_i^m$. In this case, it follows from \eqref{reference:IVP2:constant:terms} that $\chi_{\kappa}^{[i]}(\cdot)\equiv x_{l_{\kappa},G}$ and thus, by using similar arguments with Case (ia) that \eqref{reftraj:deviation:CaseA} is fulfilled.  

\noindent \textit{Case (iib).} $\kappa\in(\bar{\N}_i^m\cap\bar{\N}_{\ell}^m)\setminus\N_i^m$. Notice that $(\bar{\N}_i^m\cap\bar{\N}_{\ell}^m)\setminus\N_i^m\subset\bar{\N}_i^{m-1}\cap\bar{\N}_{\ell}^m$ and thus, for each agent $\kappa\in(\bar{\N}_i^m\cap\bar{\N}_{\ell}^m)\setminus\N_i^m$ we have that $\kappa\in\bar{\N}_i^{m-1}$ and $\kappa\in\bar{\N}_{\ell}^m$. Hence, we obtain from Lemma~\ref{lemma:mneighbor:contains:mmin1}(ii) and the fact that $\N_{\ell}^{m+1}=\emptyset$, that $\N_{\kappa}\subset\bar{\N}_i^m$ and $\N_{\kappa}\subset\bar{\N}_{\ell}^{m+1}=\bar{\N}_{\ell}^m\cup\N_{\ell}^{m+1}=\bar{\N}_{\ell}^m$, respectively, implying that $\N_{\kappa}\subset\bar{\N}_i^m\cap\bar{\N}_{\ell}^m$. The remaining proof for  this case follows similar arguments with the proof of Case~(ib) and is omitted.

\noindent \textit{Case (iii).}  $\N_i^{m+1}=\emptyset$ and  $\N_{\ell}^{m+1}\ne\emptyset$. We consider again the following subcases for each agent $\kappa\in\bar{\N}_i^m\cap\bar{\N}_{\ell}^m$.

\noindent \textit{Case (iiia).} $\kappa\in(\bar{\N}_i^m\cap\bar{\N}_{\ell}^m)\cap\N_{\ell}^m$. In this case, it follows from \eqref{reference:IVP2:constant:terms} that $\chi_{\kappa}^{[\ell]}(\cdot)\equiv x_{l_{\kappa},G}$ and thus, by using again similar arguments with Case (ia) that \eqref{reftraj:deviation:CaseA} is fulfilled.  

\noindent \textit{Case (iiib).} $\kappa\in(\bar{\N}_i^m\cap\bar{\N}_{\ell}^m)\setminus\N_{\ell}^m$. Notice that $(\bar{\N}_i^m\cap\bar{\N}_{\ell}^m)\setminus\N_{\ell}^m\subset\bar{\N}_i^m\cap\bar{\N}_{\ell}^{m-1}$ and thus, for each agent $\kappa\in(\bar{\N}_i^m\cap\bar{\N}_{\ell}^m)\setminus\N_{\ell}^m$ we have that $\kappa\in\bar{\N}_i^m$ and $\kappa\in\bar{\N}_{\ell}^{m-1}$. Hence, we obtain from Lemma~\ref{lemma:mneighbor:contains:mmin1}(ii) and the fact that $\N_i^{m+1}=\emptyset$, that $\N_{\kappa}\subset\bar{\N}_i^{m+1}=\bar{\N}_i^m\cup\N_i^{m+1}=\bar{\N}_i^m$ and $\N_{\kappa}\subset\bar{\N}_{\ell}^m$, respectively, implying that $\N_{\kappa}\subset\bar{\N}_i^m\cap\bar{\N}_{\ell}^m$. The remaining proof for Case~(iiib) follows again similar arguments with the proof of Case~(ib) and is omitted.

\noindent \textit{Case (iv).}  $\N_i^{m+1}=\emptyset$ and  $\N_{\ell}^{m+1}=\emptyset$. In this case the result follows from the proof Lemma~\ref{lemma:zero:deviation}, which implies that the trajectories $\chi_{\kappa}^{[i]}(\cdot)$ and $\chi_{\kappa}^{[\ell]}(\cdot)$ coincide for all $\kappa\in\bar{\N}_i^m\cap\bar{\N}_{\ell}^m$. The proof is now complete.
\end{proof}

\subsection{Proof of Fact I}

We provide the proof of Fact I.

\begin{proof}
Indeed, from the classical version of the Gronwall Lemma (see for instance \cite{Kh02}) we have that
\begin{align*}
y(t) & \le \lambda(t)+\int_a^t \lambda(s)\mu e^{\mu(t-s)}ds \\
& = \lambda(t)+\int_a^t \lambda(s)\frac{d}{ds}(-e^{\mu(t-s)})ds \\
& = \lambda(t)-\lambda(t)+\lambda(a)e^{\mu(t-a)}+\int_a^t e^{\mu(t-s)}\dot{\lambda}(s)ds \\
& = \int_a^te^{\mu(t-s)}\dot{\lambda}(s)ds.
\end{align*}
\end{proof}

\subsection{Evaluation of explicit expressions for the functions $H_{m}(\cdot)$, $m\ge 2$ in \eqref{functions:Hkappa}.} For $m=2$ and $m=3$ the functions $H_{2}(\cdot)$ and $H_{3}(\cdot)$ are given as
\begin{align*}
H_{2}(t) & = \frac{L_1}{L_2}\sqrt{N_{\max}}M\left(\frac{e^{L_2t}}{L_2}-t-\frac{1}{L_2}\right) \\
H_{3}(t) & = \left(\frac{L_1}{L_2}\sqrt{N_{\max}}\right)^2M\left(e^{L_2t}t-\frac{2}{L_2}e^{L_2t}+t+\frac{2}{L_2}\right).
\end{align*}

\noindent For each $m\ge 4$, it follows by induction that
\begin{align*}
H_{m}(t) & =\left(\frac{L_1}{L_2}\sqrt{N_{\max}}\right)^{m-1}M\left(e^{L_2t}\left[L_2^{m-3}\frac{t^{m-2}}{(m-2)!}+\sum_{j=1}^{m-2}(-1)^j(j+1)L_2^{m-3-j}\frac{t^{m-2-j}}{(m-2-j)!}\right]\right. \\
         & \left.+(-1)^{m-1} \left[t+\frac{m-1}{L_2}\right]\right).
\end{align*}

\subsection{Proof of \eqref{tbar:vs:tstar} and \eqref{function:Hm:bound} in Section 4.} In order to show \eqref{tbar:vs:tstar}, let $g(t):=e^{L_2t}-\left(L_2+\bar{c}\frac{L_2^2}{L_1\sqrt{N_{\max}}}\right)t-1$, whose derivative  $\dot{g}(t)=L_2e^{L_2t}-L_2-\bar{c}\frac{L_2^2}{L_1\sqrt{N_{\max}}}$ satisfies $\dot{g}(0)=-\bar{c}\frac{L_2^2}{L_1\sqrt{N_{\max}}}<0$. Hence, since $g(0)=0$ and $\lim_{t\to\infty}g(t)=\infty$, it holds $\bar{t}\in\RgO$ from \eqref{time:tbar}. Next, by recalling that $t^*$ is the unique solution of \eqref{time:tstar}, we obtain that
\begin{align*}
e^{L_2t^*} & -\left(L_2+\bar{c}\frac{L_2^2}{L_1\sqrt{N_{\max}}}\right)t^*-1 \\
& =e^{L_2t^*}-\left(L_2+\frac{L_2^2}{L_1\sqrt{N_{\max}}}\right)t^*-1+\frac{(1-\bar{c})L_2^2}{L_1\sqrt{N_{\max}}}t^*=\frac{(1-\bar{c})L_2^2}{L_1\sqrt{N_{\max}}}t^*>0,
\end{align*}

\noindent which by virtue of \eqref{time:tbar} implies \eqref{tbar:vs:tstar}. We next show \eqref{function:Hm:bound} by induction on $\kappa\in\{1,\ldots,m\}$. Notice that by \eqref{functions:Hkappa}, \eqref{function:Hm:bound} holds for $\kappa=1$. Assuming that it is valid for certain $\kappa\in\{1,\ldots,m\}$ we show that also holds for $\kappa=\kappa+1$. Hence, by exploiting \eqref{functions:Hkappa}, \eqref{time:tbar} and \eqref{function:Hm:bound} with $m=\kappa$ we obtain
\begin{align*}
H_{\kappa+1}(t)-\bar{c}^{\kappa}Mt & =\int_0^t e^{L_2(t-s)}L_1\sqrt{N_{\max}}H_{\kappa}(s)ds-\bar{c}^{\kappa}Mt \\
& \le\bar{c}^{\kappa-1}L_1\sqrt{N_{\max}}\left(\int_0^t e^{L_2(t-s)}sds-\frac{\bar{c}t}{L_1\sqrt{N_{\max}}}\right) \\
& =\bar{c}^{\kappa-1}L_1\sqrt{N_{\max}}\left(\frac{1}{L_2}\left(\frac{e^{L_2t}}{L_2}-t-\frac{1}{L_2}\right)-\frac{\bar{c}t}{L_1\sqrt{N_{\max}}}\right) \\
& =\bar{c}^{\kappa-1}L_1\sqrt{N_{\max}}L_2^2\left(e^{L_2t}-\left(L_2+\frac{\bar{c}L_2^2}{L_1\sqrt{N_{\max}}}\right) t-1\right)\le 0,\forall t\in[0,\bar{t}],
\end{align*}

\noindent which establishes \eqref{function:Hm:bound} with $m=k+1.$

\section{Acknowledgements}

This work was supported by the H2020 ERC Starting Grant BUCOPHSYS, the Knut and Alice Wallenberg Foundation, the Swedish Foundation for Strategic Research (SSF), and the Swedish Research Council
(VR).

\bibliographystyle{abbrv}
\bibliography{longtitles,decentralization_degree_references}

\end{document}